\documentclass[12pt]{article}

\usepackage{titlesec}
\titleformat{\section}
  {\normalfont\fontsize{12}{15}\bfseries}{\thesection}{1em}{}
\titleformat{\subsection}
  {\normalfont\fontsize{12}{15}\bfseries}{\thesubsection}{1em}{}

\usepackage{setspace} 

\usepackage{times}
\usepackage{rotating}
\usepackage{bbm}
\usepackage{amsfonts}
\usepackage{amssymb}
\usepackage{amsmath}
\usepackage{amsthm}
\usepackage{graphics}
\usepackage{epsfig,amssymb,latexsym,verbatim}
\usepackage{graphicx}
\usepackage{multirow}
\usepackage{multicol}
\usepackage{float}
\usepackage{graphicx, amssymb}
\usepackage{hyperref}
\usepackage{color}
\usepackage{algorithm}
\usepackage{algpseudocode}
 \usepackage{setspace}
 \usepackage{nccmath}
 \usepackage[titletoc]{appendix}
 \usepackage{enumitem}
\usepackage{apacite}
\newfloat{algorithm}{tbp}{loa}
\floatname{algorithm}{Algorithm}

\newtheorem{theorem}{{\bf Theorem}}

\newtheorem{proposition}{{\bf Proposition}}
\newtheorem{lemma}{{\bf Lemma}}
\newtheorem{example}{{\bf Example}}
\newtheorem{assumption}{{\bf Assumption}}
\newtheorem{remark}{{\bf Remark}}
\newtheorem{definition}{{\bf Definition}}
\newcommand{\captionfonts}{\normalsize}

\makeatletter  
\long\def\@makecaption#1#2{%
  \vskip\abovecaptionskip
  \sbox\@tempboxa{{\captionfonts #1: #2}}%
  \ifdim \wd\@tempboxa >\hsize
    {\captionfonts #1: #2\par}
  \else
    \hbox to\hsize{\hfil\box\@tempboxa\hfil}%
  \fi
  \vskip\belowcaptionskip}
\makeatother   
\begin{document}
\doublespacing

\hspace{13.9cm}1

\vspace{20mm}

{\bf Nearly Optimal Learning using Sparse Deep ReLU Networks in Regularized Empirical Risk Minimization with Lipschitz Loss}

\ \\
{\bf Ke Huang$^{\displaystyle 1}$, Mingming Liu$^{\displaystyle 1}$, Shujie Ma$^{\displaystyle 1}$}\\
{$^{\displaystyle 1}$Department of Statistics, University of California, Riverside, Riverside 92521, California, United States}\\
%
%\ \\[-2mm]

\thispagestyle{empty}
\markboth{}{NC instructions}
\ \vspace{-0mm}\\

%\title{Nearly Optimal Learning using Sparse Deep ReLU Networks in Regularized Empirical Risk Minimization with Lipschitz Loss\tnoteref{t1}}
%%\tnotetext[t1]{The research of this project is supported in part by the U.S. NSF grants DMS-17-12558, DMS-20-14221 and DMS-
%23-10288 and the UCR Academic Senate CoR Grant.}

%\author[label2]{Ke Huang}
%\ead{khuan049@ucr.edu}

%\author[label2]{Mingming Liu}
%\ead{mliu034@ucr.edu}

%author[label2]{Shujie Ma\corref{cor1}}
%\ead{shujie.ma@ucr.edu}
% \affiliation[label2]{organization={Department of Statistics},
%             addressline={University of California, Riverside},
%             city={Riverside},
%             postcode={92521},
%             state={California},
%             country={United States}}

\begin{center} {\bf Abstract} \end{center}
We propose a sparse deep ReLU network (SDRN) estimator of the regression function obtained from regularized
empirical risk minimization with a Lipschitz loss function. Our framework can be applied to a variety of
regression and classification problems.  We establish novel non-asymptotic excess risk
bounds for our SDRN estimator when the regression function belongs to a Sobolev space with mixed derivatives. We obtain a new nearly optimal risk rate in the sense that the SDRN estimator can achieve nearly the same optimal minimax convergence rate as one-dimensional nonparametric
regression with the dimension only involved in a logarithm term when the feature dimension is fixed. The estimator has a slightly slower rate when the dimension grows
with the sample size. We show that the depth of the SDRN estimator grows with the sample size in logarithmic order, and the total number of nodes
and weights grows in polynomial order of the sample size to have the nearly optimal risk rate. The proposed SDRN can go deeper with fewer parameters to well estimate the regression and overcome the overfitting problem encountered by conventional feed-forward neural networks.

%%%%%%%%%%%%%%%%%%%%%%%%%%%%%%%%%%%%%%%%%%%%%%
%% Please use \tableofcontents for articles %%
%% with 50 pages and more                   %%
%%%%%%%%%%%%%%%%%%%%%%%%%%%%%%%%%%%%%%%%%%%%%%
%\tableofcontents

%%%%%%%%%%%%%%%%%%%%%%%%%%%%%%%%%%%%%%%%%%%%%%
%%%% Main text entry area:

%\clearpage\pagebreak\newpage \pagenumbering{arabic} \newlength{\gnat} %
%\setlength{\gnat}{22pt} \baselineskip=\gnat

\section{Introduction\label{SEC:INTRODUCTION}}

Advances in modern technologies have facilitated the collection of
large-scale data that are growing in both sample size and the number of
variables. Although conventional parametric models, such as generalized
linear models are convenient for studying the relationships between
variables, they may not be flexible enough to capture complex patterns
in large-scale data. With a large sample size, the bias due to model
misspecification becomes more prominent compared to sampling variability,
and may lead to false conclusions. The problem of model misspecification
can be solved by nonparametric regression methods that are capable of
approximating the unknown target function well without a restrictive
structural assumption. Theoretically, we hope that both the bias and the
variance of the functional estimator decrease as the sample size increases.
Moreover, the bias is reduced by increasing the variance and vice versa, so
that a tradeoff between bias and variance can be achieved for an accurate
prediction.

In the classical multivariate regression context with a smoothness condition
imposed on the target function, the conventional nonparametric smoothing
methods such as local kernels and splines (e.g. \citeNP{S94,CFM94,FG96,R97,W06,MRY2015}) suffer from the so-called
\textquotedblleft curse of dimensionality\textquotedblright\ \cite{B61},
i.e., the convergence rate of the resulting functional estimator
deteriorates sharply as the dimension of the predictors increases. As such
it is desirable to develop analytic tools that can alleviate the curse of
dimensionality while preserving sufficient flexibility, to accommodate the
large volume as well as the high dimensionality of the modern data.

In recent years, the investigation of statistical properties for regression using deep neural networks has received increasing attention in the statistical and machine learning communities. Deep neural networks with multiple hidden layers are powerful and effective machine learning tools for prediction and classification, and have been successfully applied to many fields, including
computer vision, language processing, speech recognition, time series forecasting, and biomedical
studies (e.g. \citeNP{AB09,LBH15,S15,LS16,WC18,CSG221,FMZ21}). Several pioneering works (e.g. \citeNP{BK2019,Schmidt2020,CJLZ19}) have established convergence rates for neural network estimators of a regression function when it has a compositional structure (e.g. \citeNP{PMRML17,MLP17,BK2019,Schmidt2020}), or the covariates are assumed to lie in a low-dimensional
manifold (e.g. \citeNP{CW13,SCC18,NI2020,SYZ2020,CJLZ19}). The compositional structure or the low-dimensional manifold assumption is assumed to alleviate the dimensionality problem. Moreover, most existing works (e.g. \citeNP{KM11,BK2019,Schmidt2020,CJLZ19,LWWZ22}) estimate the regression function using the least-squares method. 

In this paper, we develop novel non-asymptotic excess risk bounds and convergence rates for a sparse deep ReLU network (SDRN) estimator of the regression function through regularized empirical risk minimization (ERM) with a Lipschitz loss function satisfying mild
conditions. $L_2$ regularization is employed by our approach to prevent possible over-fitting. We consider a unified framework in which the family of the loss functions is a general class. It includes the quadratic, Huber, quantile, and logistic loss functions as special cases, so many regression and
classification problems can be solved by our framework. Classification (e.g. \citeNP{K07,KS19}) is a
crucial task of supervised learning, and robust regression is an important
tool for analyzing data with heavy tails. We adopt a network structure of sparsely connected deep
neural networks with the rectified linear unit (ReLU) activation function given in \cite{MD19} that allows the network to go deeper with fewer parameters than the conventional feedforward neural networks (FNNs), so it enjoys parsimony, scalability, and lower computational complexity, which are important for deep learning models  \cite{BL07,ES16,MP16,MLP17}. On the contrary, conventional FNNs can be computationally expensive and may suffer from overfitting \cite{SHKSS2014}. Different from \cite{MD19} which focuses on finding a sparse neural network with known and given weights to approximate a target function, we obtain the SDRN estimator of a regression function with unknown weights trained through minimizing the regularized ERM. We also establish a new nearly optimal risk rate for the SDRN estimator when the regression function belongs to the Sobolev spaces of functions with
square-integrable mixed second derivatives (also called Korobov spaces),
commonly assumed for the sparse grid approaches dealing with high-dimensional partial
differential equations \cite{BG04,G06,SW10,MD19,MZ22}. This rate can break the notorious `curse of dimensionality' suffered by the conventional nonparametric methods. 

\textcolor{black}{ The main contributions of our paper are summarized as follows:
\begin{itemize}[leftmargin=*]
    \item We derive novel non-asymptotic excess risk bounds for the sparse deep ReLU network
(SDRN) estimator obtained from regularized empirical risk minimization. The SDRN architecture allows the network to go deeper with fewer parameters than the regular FNNs, and it can overcome the overfitting problem suffered by conventional FNNs. Our established risk bounds consist of the estimation and approximation errors, for both of which we derive explicit forms as a function of the sample size, dimension of the feature space, and network complexity. The existing work \cite{MD19} 
uses an accuracy value $\epsilon >0$ for the approximation error only without data fitting. Moreover, we derive a
new non-asymptotic bound for the network complexity. This bound has not been provided in \cite{MD19}.
These newly established bounds provide important theoretical guidance on
how the network complexity should be related to the sample size and the dimension, so that a
tradeoff between the two errors can be achieved to secure an optimal fitting from the dataset.
\item We show that the SDRN estimator can break the \textit{curse of dimensionality} when the regression function is in the  Korobov space by allowing the covariate dimension $d$ to increase with sample size $n$ with a rate slightly slower than $\log (n)$. The existing works on FNN regression, for example, \cite{BK2019} and \cite{Schmidt2019},  still require the dimension $d$ to be fixed.  We further
show that our SDRN estimator can achieve nearly the optimal minimax convergence rate as one-dimensional nonparametric
regression with the dimension $d$ only involved in a logarithm term when the dimension is fixed. The SDRN estimator has a suboptimal rate (slightly slower than the
optimal rate) when the dimension increases with the sample size. 
\item We provide an architecture of the SDRN, and estimate the unknown parameters (unknown weights) in the SDRN through regularized ERM with Lipschitz loss, while  \cite{MD19} focus on the study of the approximation of the sparse neural network with given and known weights to the target function. In addition, our SDRN can be applied to both regression and classification problems. We discuss several commonly used loss functions that satisfy the Lipschitz condition. 
%We provide empirical studies to validate the theoretical predictions regarding error bounds. Additionally, our SDRN is applicable to both regression and classification problems. We discuss several commonly used loss functions satisfying Lipschitz's loss condition and provide numerical analysis for both sets of problems. %\cite{MD19} uses an accuracy value $\epsilon >0$ for the approximation error without data fitting. 
\item We show that the depth of SDRN
increases with the sample size $n$ at a logarithmic rate, and the number of nodes and weights
only need to grow with $n$ at a polynomial rate to ensure the convergence properties of the SDRN estimator. The proposed SDRN
can go deeper with fewer parameters to well estimate the regression and overcome the overfitting problem encountered by conventional FNNs.
%These new results provide a theoretical basis for empirical studies by deep neural networks and are also demonstrated by our numerical analysis.
\end{itemize}
}

\iffalse
\textcolor{black}{Our main contributions can be summarized as follows:
\begin{enumerate}
    \item We focus on the statistical performance of sparse deep ReLU network
(SDRN) estimator obtained from empirical risk minimization (ERM) with Lipschitz loss. We establish non-asymptotic excess risk bounds for SDRN. 
These bounds encompass the explicit forms for both estimation and approximation errors relative to data and sample size, and demonstrate how the error bounds are influenced by sample size.
Additionally, we discuss several commonly used loss functions satisfying Lipschitz's loss condition. These are not provided and discussed in \cite{MD19}.
\item  We show that within Korobov spaces, SDRN alleviates the \textit{curse of dimensionality} by allowing dimension $d$ to increase with sample size $n$ at a slow rate of $d = \mathcal{O}((\log_2n)^{1-\kappa})$, and achieve the excess risk rate slightly slower than $n^{-4/5}$.  Classical nonparametric regression estimators and existing neural network estimators restrict dimensions to be fixed. When $d$ is fixed, we show that SDRN achieves the optimal minimax rate as one-dimensional nonparametric regression and can alleviate the \textit{curse of dimensionality} with the impact of dimension on a logarithm order. 
\item  We develop the SDRN structure, including the construction and the unknown parameters estimation, while the parameters in \cite{MD19} are given. We provide empirical studies that validate the theoretical predictions regarding error bounds. 
\end{enumerate}
}
\fi

The paper is organized as follows. Section \ref{SEC:MODEL} provides the
basic setup, Section \ref{SEC:NN} discusses approximation of the target
function by the ReLU networks, Section \ref{SEC:SDRN algrithom} constructs the sparse deep ReLU network (SDRN) estimator, Section \ref{SEC:ReLUestimator} establishes the theoretical properties for the SDRN estimator obtained from empirical risk
minimization, Section \ref%
{discussion_ass} further discusses the conditions imposed on the loss function,
Section \ref{sec:simulation} reports results from simulation studies, and
Section \ref{sec:real} illustrates the proposed method through real data
applications. Some concluding remarks are given in Section \ref%
{sec:Discussion}. All the technical proofs are provided in the
Appendix.

\textbf{Notations: }Let $\mathbf{a}_{d}=(a,...,a)^{\top }$ be a $d$%
-dimensional vector of $a$'s. let $|A|$ be the cardinality of a set $A$. The
vectorization of a $m\times n$ matrix $\boldsymbol{A}$, denoted vec$\left(
\boldsymbol{A}\right) $, is the $mn\times 1$ column vector by stacking the
columns of the matrix $\boldsymbol{A}$. Denote $|\boldsymbol{a}%
|_{p}=(\sum_{i=1}^{m}|a_{i}|^{p})^{1/p}$ as the L$^{p}$-norm of a vector $%
\boldsymbol{a}\mathbf{=(}a_{1},\ldots ,a_{m})^{\top }$, and $|\boldsymbol{a}%
|_{\infty }=\max_{i=1,...,m}|a_{i}|$. For two vectors $\boldsymbol{a}=%
\mathbf{(}a_{1},\ldots ,a_{m})^{\top }$ and $\boldsymbol{b}=\mathbf{(}%
b_{1},\ldots ,b_{m})^{\top }$, denote $\boldsymbol{a\cdot b}%
=\sum_{i=1}^{m}a_{i}b_{i}$. Moreover, for any arithmetic operations
involving vectors, they are performed element-by-element. For any two values
$a$ and $b$, denote $a\vee b=\max (a,b)$. For two sequences of positive
numbers $a_{n}$ and $b_{n}$, $a_{n}\ll b_{n}$ means that $%
b_{n}^{-1}a_{n}=o(1)$, $a_{n}\lesssim b_{n}$ means that there exists a
constant $C\in (0,\infty )$ and $n_{0}\geq 1$ such that $a_{n}\leq Cb_{n}$
for $n\geq $ $n_{0}$, and $a_{n}\asymp b_{n}$ means that there exist
constants $C,C^{\prime }\in (0,\infty )$ and $n_{0}\geq 1$ such that $%
a_{n}\leq Cb_{n}$ and $b_{n}\leq C^{\prime }a_{n}$ for $n\geq $ $n_{0}$.

\section{Basic setup\label{SEC:MODEL}}

We consider a general setting of many supervised learning problems. Let $%
Y\in \mathcal{Y\subset }\mathbb{R}$ be a real-valued response variable and $%
\boldsymbol{X}\mathbf{=(}X_{1},\ldots ,X_{d})^{\top }$ be $d$-dimensional
independent variables with values in a compact support $\mathcal{X\subset }%
\mathbb{R}^{d}$. Without loss of generality, we let $\mathcal{X}=[0,1]^{d}$.
Let $(\boldsymbol{X}_{i}^{\top },Y_{i})^{\top }$, $i=1,...,n$ be i.i.d.
samples (a training set of $n$ examples) drawn from the distribution of $(%
\boldsymbol{X}^{\top },Y)^{\top }$. We consider the mapping $\ f:\mathcal{X}%
\rightarrow \mathbb{R}$. Our goal is to estimate the unknown regression function
$f\left( \boldsymbol{x}\right) $ using sparse deep neural networks from the
training set.

Let $\mu :$ $\mathcal{X}\times \mathcal{Y}\rightarrow \lbrack 0,1]$ be a
Borel probability measure of $(\boldsymbol{X}^{\top },Y)^{\top }$. For every
$\boldsymbol{x}\in \mathcal{X}$, let $\mu (y|\boldsymbol{x})$ be the
conditional (w.r.t. $\boldsymbol{x}$) probability measure of $Y$. Let $\mu
_{X}$ be the marginal probability measure of $\boldsymbol{X}$. For any $%
1\leq p\leq \infty $, let $\mathcal{L}^{p}\left( \mathcal{X}\right) =\{f:%
\mathcal{X}\rightarrow \mathbb{R}$, $f$ is Lebesgue measurable on $\mathcal{X%
}$ and $||f||_{L^{p}}<\infty \}$, where $||f||_{L^{p}}=(\int_{\boldsymbol{x}%
\in \mathcal{X}}|f\left( \boldsymbol{x}\right) |^{p}d\boldsymbol{x})^{1/p}$
for $1\leq p<\infty $, and $||f||_{L^{\infty }}=||f||_{\infty }=\sup_{%
\boldsymbol{x}\in \mathcal{X}}|f\left( \boldsymbol{x}\right) |$. For $1\leq
p<\infty $, denote $||f||_{p}=(\int_{\boldsymbol{x}\in \mathcal{X}}|f\left(
\boldsymbol{x}\right) |^{p}d\mu _{X}(\boldsymbol{x}))^{1/p}$ and $%
||f||_{p,n}=(n^{-1}\sum_{i=1}^{n}|f\left( \boldsymbol{X}_{i}\right)
|^{p})^{1/p}$. Let $\rho :\textcolor{black}{\mathbb{R}}\times \mathcal{Y}\rightarrow \mathbb{R}
$ be a loss function. The true target function $f_{0}$ is defined as \
\begin{equation}
f_{0}=\arg \min_{f\in \mathcal{L}^{p}\left( \mathcal{X}\right) }\mathcal{E}%
(f)\text{, where }\mathcal{E}(f)=\int_{\mathcal{X\times Y}}\rho (f(%
\boldsymbol{x}),y)d\mu (\boldsymbol{x,}y).  \label{EQ:f0}
\end{equation}%
Next, we introduce the Korobov spaces, in which the functions need to satisfy
a certain smoothness condition. The partial derivatives of $f$ with
multi-index $\boldsymbol{k}=(k_{1},...,k_{d})^{\top }\in \mathbb{N}^{d}$ is
given as $D^{\boldsymbol{k}}f=\frac{\partial ^{|\boldsymbol{k}|_{1}}f}{%
\partial x_{1}^{k_{1}}\cdots \partial x_{d}^{k_{d}}}$, where $\mathbb{N}%
=\{0,1,2,...,\}$ and $|\boldsymbol{k}|_{1}=\sum_{j=1}^{d}k_{j}$.

\begin{definition}
For $2\leq p\leq \infty $, the Sobolev spaces of mixed second derivatives
(also called Korobov spaces) $W^{2,p}(\mathcal{X)}$ are define by
\begin{equation*}
W^{2,p}(\mathcal{X})=\{f\in \mathcal{L}^{p}\left( \mathcal{X}\right) :D^{%
\boldsymbol{k}}f\in \mathcal{L}^{p}\left( \mathcal{X}\right) ,|\boldsymbol{k}%
|_{\infty }\leq 2\}\text{, where }|\boldsymbol{k}|_{\infty
}=\max_{j=1,...,d}k_{j}.
\end{equation*}
\end{definition}

\begin{assumption}
\label{ass1} We assume that $f_{0}\in W^{2,p}(\mathcal{X})$, for a given $%
2\leq p\leq \infty $.
\end{assumption}

\begin{remark}
Assumption \ref{ass1} imposes a smoothness condition on the target function %
\cite{BG04,G06,MD19}. The Korobov spaces $W^{2,p}(\mathcal{X)}$ are subsets
of the regular Sobolev spaces defined as $S^{2,p}(\mathcal{X})=\{f\in
\mathcal{L}^{p}\left( \mathcal{X}\right) :D^{\boldsymbol{k}}f\in \mathcal{L}%
^{p}\left( \mathcal{X}\right) ,|\boldsymbol{k}|_{1}\leq 2\}$ assumed in the
traditional nonparametric regression setting \cite{W06}. For instance, when
$d=2,$ $|\boldsymbol{k}|_{\infty }=\max (k_{1},k_{2})\leq 2$ implies $|%
\boldsymbol{k}|_{1}=k_{1}+k_{2}\leq 4$. If $f\in W^{2,p}(\mathcal{X)}$, it
needs to satisfy
\begin{equation*}
\frac{\partial f}{\partial x_{j}},\frac{\partial ^{2}f}{\partial x_{j}^{2}},%
\frac{\partial ^{2}f}{\partial x_{1}\partial x_{2}},\frac{\partial ^{3}f}{%
\partial x_{1}^{2}\partial x_{2}},\frac{\partial ^{3}f}{\partial
x_{1}\partial x_{2}^{2}},\frac{\partial ^{4}f}{\partial x_{1}^{2}\partial
x_{2}^{2}}\in \mathcal{L}^{p}\left( \mathcal{X}\right) .
\end{equation*}%
If $f\in S^{2,p}(\mathcal{X})$, it needs to satisfy $\frac{\partial f}{%
\partial x_{j}},\frac{\partial ^{2}f}{\partial x_{j}^{2}},\frac{\partial
^{2}f}{\partial x_{1}\partial x_{2}}\in \mathcal{L}^{p}\left( \mathcal{X}%
\right) $. It is worth noting that no nonparametric regression methods can
avoid the \textquotedblleft curse of dimensionality\textquotedblright\ if
the target function belongs to the regular Sobolev spaces. Functions in the
Korobov spaces need to be smoother than those in the regular Sobolev spaces,
and many popular structured nonparametric models satisfy this condition %
(see \citeNP{G06}). Note that when $d=1$ (one-dimensional nonparametric
regression), the Korobov and the Sobolev spaces are the same, i.e., if $%
f\in W^{2,p}(\mathcal{X)}$ or $f\in S^{2,p}(\mathcal{X)}$, it needs to
satisfy $\frac{\partial f}{\partial x_{1}},\frac{\partial ^{2}f}{\partial
x_{1}^{2}}\in \mathcal{L}^{p}\left( \mathcal{X}\right) $.
\end{remark}

\textcolor{black}{
Below, we provide several examples of regression functions that belong to the Korobov space given in Definition 1, so they satisfy the condition in Assumption 1. These regression models are popularly used in the non- and semi-parametric regression literature. Let $\boldsymbol{x} = (x_1, x_2, \dots, x_d)^{\top}$.
\begin{example}
\textbf{Additive model} \cite{Stone85} is defined as $f_{0}(\boldsymbol{x})=\sum_{j=1}^{d}f_{j}(x_{j})$, where $f_{j}\left( \cdot
\right) $ is an unknown but smooth function of the $j^{\text{th}}$
covariate, for $j=1,\dots ,d$. In the literature, to ensure nonparametric estimators of $f_{j}\left( \cdot \right) $ have a good
convergence rate, it requires a smoothness condition on each univariate function $%
f_{j}\left( \cdot \right) $, such as $D^{2}f_{j}\in \mathcal{L}%
^{p}\left( \mathcal{X}\right) $, i.e., the second derivative of $f_{j}\left(
\cdot \right) $ exists and is integrable, then the regression function $%
f_{0}\in W^{2,p}(\mathcal{X)}$, satisfying Assumption \ref{ass1}.
\end{example}
\begin{example}
\textbf{Functional ANOVA model} \cite{Stone97} consists of the main and
interaction effect terms. It is typically expressed as
\begin{equation*}
f_{0}(\boldsymbol{x})=\sum_{j=1}^{d}f_{j}(x_{j})+\sum_{i\neq j}f_{jj^{\prime
}}(x_{j^{\prime }},x_{j}),
\end{equation*}%
where $f_{j}$ are unknown but smooth functions for the main effects and $f_{jj^{\prime }}$ are unknown functions for the second-order interaction
effects. When the univariate functions $f_{j}$ and the bivariate functions $f_{jj^{\prime }}$ satisfy
the smoothness conditions such that $D^{2}f_{j}\in \mathcal{L}%
^{p}\left( \mathcal{X}\right) $ and  $D^{\boldsymbol{k}}f_{jj^{\prime }}\in
\mathcal{L}^{p}\left( \mathcal{X}\right) $ for $|\boldsymbol{k}|_{\infty
}=\max (k_{1},k_{2})\leq 2$, where $\boldsymbol{k}=(k_1,k_2)^{\top}$, the regression function $f_{0}\in W^{2,p}(\mathcal{X)}$,
satisfying Assumption \ref{ass1}.
\end{example}
\begin{example}
\textbf{Sparse tensor decomposition model} \cite{Schmidt2020} assumes that the regression function has the form
\begin{equation*}
f_{0}(\boldsymbol{x})=\sum_{\ell=1}^{L}\alpha
_{\ell}\prod\limits_{j=1}^{d}f_{j\ell}(x_{j}),
\end{equation*}%
for fixed $L$, real coefficients $\alpha _{\ell}$ and univariate functions $%
f_{j\ell}$. When the unknown univariate functions $f_{j\ell}$ satisfy the
smoothness condition such that $D^{2}f_{j\ell}\in \mathcal{L}%
^{p}\left( \mathcal{X}\right) $, the regression function $f_{0}\in W^{2,p}(%
\mathcal{X)}$, satisfying Assumption \ref{ass1}.
\end{example}
}

\begin{assumption}
\label{ass2}For any $y\in \mathcal{Y}$, the loss function $\rho \left( \cdot
,y\right) $ is convex and it satisfies the Lipschitz property such that
there exists a constant $0<C_{\rho }<\infty $, for almost every $(%
\boldsymbol{x},y)\in \mathcal{X}\times \mathcal{Y}$, $|\rho \left( f_{1}(%
\boldsymbol{x}),y\right) -\rho \left( f_{2}(\boldsymbol{x}),y\right) |\leq
C_{\rho }|f_{1}(\boldsymbol{x})-f_{2}(\boldsymbol{x})|$, for any $%
f_{1},f_{2}\in \mathcal{F}$ , where $\mathcal{F}$ is a neural network space given in Section \ref{SEC:SDRN algrithom}. %\ref{SEC:ReLUestimator}.
\end{assumption}

\begin{remark}
The above Lipschitz assumption is satisfied by many commonly used loss
functions. Several examples are provided below.
\end{remark}

\begin{example}
\textbf{Huber loss }is popularly used for robust regression, and it is
defined as
\begin{equation}
\rho \left( f(\boldsymbol{x}),y\right) =\left\{
\begin{array}{cc}
2^{-1}(y-f(\boldsymbol{x}))^{2} & \text{if \ }|f(\boldsymbol{x})-y|\leq
\delta \\
\delta |y-f(\boldsymbol{x})|-\delta ^{2}/2 & \text{if \ }|f(\boldsymbol{x}%
)-y|>\delta%
\end{array}%
\right. .  \label{def:huber}
\end{equation}%
It satisfies Assumption \ref{ass2} with $C_{\rho }=\delta $.
\end{example}

\begin{example}
\textbf{Quantile loss }is another popular loss function for robust
regression, and it is defined as
\begin{gather}
\rho \left( f(\boldsymbol{x}),y\right) =(y-f(\boldsymbol{x}))(\tau -I\{y-f(%
\boldsymbol{x})\leq 0\})  \label{def:quantile}
\end{gather}
for $\tau \in (0,1)$. It satisfies Assumption \ref{ass2} with $C_{\rho }=1$.
\end{example}

\begin{example}
\textbf{Logistic loss }is used in logistic regression for binary responses
as well as for classification. The loss function is $\rho \left( f(%
\boldsymbol{x}),y\right) =\log (1+e^{f(\boldsymbol{x})})-yf(\boldsymbol{x})$
for $y\in \{0,1\}$. It satisfies Assumption \ref{ass2} with $C_{\rho }=2$.
\end{example}

\section{Approximation of the target function by ReLU networks\label{SEC:NN}}

We consider feedforward neural networks which consist of a collection of
input variables, one output unit, and a number of computational units (nodes) in different hidden layers. In our setting, the $d$-dimensional covariates $%
\boldsymbol{X}$ are the input variables, and the approximated function is
the output unit. Each computational unit is obtained from the units in the
previous layer. Following \cite{AB09}, we measure the network complexity by
using the depth of the network defined as the number of layers, the total
number of units (nodes), and the total number of weights, which is the sum
of the number of connections and the number of units. Moreover, $\sigma :%
\mathbb{R}\rightarrow $ $\mathbb{R}$ is an activation function which is
chosen by practitioners. In this paper, we use the rectified linear unit
(ReLU) function given as $\sigma \left( x\right) =\max (0,x)$.

Our goal is to propose a ReLU network estimator for the regression function (\ref%
{EQ:f0}) obtained via the regularized ERM. We are also interested in studying whether the ReLU estimator can
break the \textquotedblleft curse of dimensionality \textquotedblright ,
when the target function satisfies the condition given in (\ref{ass1}).
Below we will first provide a reasoning why a function $f\left( \cdot
\right) $ in the Korobov space such that $f\in W^{2,p}(\mathcal{X)}$ can be
well approximated by a deep ReLU network. This will provide a mathematical
grounding for the construction of our sparse deep ReLU network estimator
introduced in Section \ref{SEC:SDRN algrithom}.

We present several results given in \cite{Y17} to show that for the function
$v(\boldsymbol{u})=\prod\nolimits_{j=1}^{d}u_{j}$ with $\boldsymbol{u}\in
\lbrack 0,1]^{d}$, where $\boldsymbol{u}=(u_{1},...,u_{d})^{\top }$, it can
be well approximated by a ReLU network. This result will be used to
construct the ReLU network approximator for the target function $f\left(
\cdot \right) $. Consider the \textquotedblleft tooth\textquotedblright\
function $g:\left[ 0,1\right] \rightarrow \left[ 0,1\right] $ given as $%
g(u)=2u$ for $u<1/2$ and $g(u)=2(1-u)$ for $u\geq 1/2$, and the iterated
functions $g_{r}(u)=\underset{r}{\underbrace{g\circ g\circ \cdot \cdot \cdot
\circ g}}(u)$. Let
\begin{equation*} \label{def:phi_true}
\varphi _{R}(u)=u-\sum\nolimits_{r=1}^{R}\frac{g_{r}(u)}{2^{2r}}.
\end{equation*} %
It is clear that $\varphi _{R}(0)=0$. It is shown in \cite{Y17} that for the
function $v(u)=u^{2}$ with $u\in \lbrack 0,1]$, it can be approximated by $%
\varphi _{R}(u)$ such that
\begin{equation*}
||v-\varphi _{R}||_{\infty }\leq 2^{-2R-2}.
\end{equation*}%
Moreover, the tooth function $g$ can be implemented by a ReLU network: $%
g\left( u\right) =2\sigma (u)-4\sigma (u-1/2)+2\sigma (u-1)$ which has 1
hidden layer and 3 computational units. Therefore, $\varphi _{R}(u)$ can be
constructed by a ReLU network with the depth $R+2$ , the computational units
$3R+1$, and the number of weights $12R-5+3R+1=15R-4$. 

Next, we approximate the function $%
v(u_{1},u_{2})=u_{1}u_{2}=2^{-1}((u_{1}+u_{2})^{2}-u_{1}^{2}-u_{2}^{2})$ for
$u_{1}\in \lbrack 0,1]$ and $u_{2}\in \lbrack 0,1]$ by a ReLU network as
follows. By the above results, we have $\textcolor{black}{|\varphi _{R}((u_{1}+u_{2})/2%
)-((u_{1}+u_{2})/2)^{2}|\leq 2^{-2R-2}}$, $|2^{-2}\varphi
_{R}(u_{1})-2^{-2}u_{1}^{2}|\leq 2^{-2}2^{-2R-2}$ and $|2^{-2}\varphi
_{R}(u_{2})-2^{-2}u_{2}{}^{2}|\leq 2^{-2}2^{-2R-2}$. Let
\begin{equation*}
\widetilde{\varphi }_{R}(u_{1},u_{2})=2\left\{ \varphi _{R}(\frac{u_{1}+u_{2}%
}{2})-\frac{\varphi _{R}(u_{1})}{2^{2}}-\frac{\varphi _{R}(u_{2})}{2^{2}}%
\right\} ,
\end{equation*}%
and $\widetilde{\varphi }_{R}(u_{1},u_{2})$ can be implemented by a ReLU
network having the depth, the computational units and the number of weights
being $c_{1}R+c_{2}$, where the constants $c_{1}$ and $c_{2}$ can be
different for these three measures. Moreover, $\widetilde{\varphi }%
_{R}(u_{1},u_{2})=0$ if $u_{1}u_{2}=0$. For all $u_{1}\in \lbrack 0,1]$ and $%
u_{2}\in \lbrack 0,1]$,
\begin{align}
& \textcolor{black}{\left\vert \widetilde{\varphi }_{R}(u_{1},u_{2})-v(u_{1},u_{2})\right\vert} \notag \\
& =2\left\vert \left\{ \varphi _{R}(\frac{u_{1}+u_{2}}{2})-\frac{\varphi
_{R}(u_{1})}{2^{2}}-\frac{\varphi _{R}(u_{2})}{2^{2}}\right\} -\left\{ (%
\frac{u_{1}+u_{2}}{2})^{2}-\frac{u_{1}{}^{2}}{2^{2}}-\frac{u_{2}{}^{2}}{2^{2}%
}\right\} \right\vert  \notag \\
& \leq 2\left( 2^{-2R-2}+2^{-2}2^{-2R-2}+2^{-2}2^{-2R-2}\right) =3\cdot
2^{-2R-2}.  \label{EQ:ftilda}
\end{align} 

At last, we can approximate the function $v(\boldsymbol{u}%
)=\prod\nolimits_{j=1}^{d}u_{j}$ for $\boldsymbol{u}\in \lbrack 0,1]^{d}$
from a binary tree structure constructed based on the function $\widetilde{\varphi }_{R}(\cdot, \cdot)$ given in (\ref{EQ:ftilda}). The resulting network is denoted by \textcolor{black}{$\widetilde{\varphi }_R%
(\boldsymbol{u})$}. It can be shown from mathematical induction \cite{MD19} that for any $%
\boldsymbol{u}\in \lbrack 0,1]^{d}$,
\begin{equation}
|\textcolor{black}{\widetilde{\varphi }_R(\boldsymbol{u})}-v(\boldsymbol{u})|\leq
(1+2+2^{2}+\cdot \cdot \cdot +2^{\left\lfloor \log _{2}d\right\rfloor
-1})\cdot 3\cdot 2^{-2R-2}\leq 3\cdot 2^{-2R-2}(d-1),  \label{EQ:phitilda}
\end{equation}%
where $\left\lfloor a\right\rfloor $ is the largest integer no greater than $%
a$. Moreover, $\textcolor{black}{\widetilde{\varphi }_{R}(\boldsymbol{u})}=0$ if $h(\boldsymbol{%
u})=0$. The ReLU network used to approximate $\phi _{\boldsymbol{\ell },%
\boldsymbol{s}}(\boldsymbol{x})$ has depth $\mathcal{O}(R)\times \log _{2}d=%
\mathcal{O}(R\log _{2}d)$, the computational units $\mathcal{O}(R)\times
(d+2^{-1}d+\cdot \cdot \cdot +2^{-\left\lfloor \log _{2}d\right\rfloor +1}d)=%
\mathcal{O}(Rd)$, and the number of weights $\mathcal{O}(Rd)$.

For any function $f\in W^{2,p}(\mathcal{X)}$, it has a unique expression in
a hierarchical basis \cite{BG04} such that $f(\boldsymbol{x})=\sum_{%
\boldsymbol{0}_{d}\leq \boldsymbol{\ell }\leq \boldsymbol{\infty }%
}\sum_{s\in I_{\boldsymbol{\ell }}}\gamma _{_{\boldsymbol{\ell },\boldsymbol{%
s}}}^{0} \phi _{\boldsymbol{\ell },\boldsymbol{s}}(\boldsymbol{x})$, where $%
\phi _{\boldsymbol{\ell },\boldsymbol{s}}(\boldsymbol{x})=\prod%
\limits_{j=1}^{d}\phi _{\ell _{j},s_{j}}(x_{j})$ are the tensor product
piecewise linear basis functions defined on the grids $\Omega _{\boldsymbol{%
\ell }}$ of level $\boldsymbol{\ell }=(\ell _{1},...,\ell _{d})^{\top }$, $%
I_{\boldsymbol{\ell }}$ are the index sets of level $\boldsymbol{\ell }$,
and the hierarchical coefficients $\gamma _{_{\boldsymbol{\ell },\boldsymbol{%
s}}}^{0}\in \mathbb{R}$ are given in (\ref{gammaexpression}). We refer to
Section \ref{SEC:discrete} in the Appendix for a detailed discussion on the
hierarchical basis functions. Section \ref{SEC:discrete} shows that for any $%
f\in W^{2,p}(\mathcal{X)}$, it can be well approximated by the hierarchical
basis functions with sparse grids such that $f(\boldsymbol{x})\approx
\sum\nolimits_{|\boldsymbol{\ell |}_{1}\leq m}\sum\nolimits_{s\in I_{%
\boldsymbol{\ell }}}\gamma _{_{\boldsymbol{\ell },\boldsymbol{s}}}^{0}\phi _{%
\boldsymbol{\ell },\boldsymbol{s}}(\boldsymbol{x})$. Then the hierarchical
space with sparse grids is given as
\begin{equation*}
V_{m}^{\left( 1\right) }=\text{span}\{\phi _{\boldsymbol{\ell },\boldsymbol{s%
}}:s\in I_{\boldsymbol{\ell }},|\boldsymbol{\ell }|_{1}\leq m\}.
\end{equation*}%
\textcolor{black}{We provide explicit upper and lower bounds for the
dimension (cardinality) of the space $V_{m}^{\left( 1\right) }$ in Proposition \ref{PROP:cardinality}.}

Based on the result given in (\ref{EQ:phitilda}), we let each $u_{j}=\phi
_{\ell _{j},s_{j}}(x_{j})$, \textcolor{black}{ so $\boldsymbol{u} = \boldsymbol{\phi}_{\boldsymbol{\ell}, \boldsymbol{s}}(\boldsymbol{x}) = (\phi
_{\ell _{1},s_{1}}(x_{1}), \phi
_{\ell _{2},s_{2}}(x_{2}), \dots, \phi
_{\ell _{d},s_{d}}(x_{d}))$, $v(\boldsymbol{u})=\prod\limits_{j=1}^{d}\phi
_{\ell _{j},s_{j}}(x_{j})=\phi _{\boldsymbol{\ell },\boldsymbol{s}}(%
\boldsymbol{x})$} and thus the hierarchical basis functions $\phi _{%
\boldsymbol{\ell },\boldsymbol{s}}(\boldsymbol{x})$ can be well approximated
by the ReLU network $\widetilde{\varphi }_R$. Then the ReLU network approximator of the unknown function $f(%
\boldsymbol{x})$ is
\begin{equation}
\widetilde{f}(\boldsymbol{x})=\sum_{|\boldsymbol{\ell |}_{1}\leq
m}\sum_{s\in I_{\boldsymbol{\ell }}}\gamma _{_{\boldsymbol{\ell },%
\boldsymbol{s}}}^{0}\textcolor{black}{\widetilde{\varphi }_R(\boldsymbol{\phi} _{\boldsymbol{\ell },%
\boldsymbol{s}}(\boldsymbol{x}))}=\sum_{|\boldsymbol{\ell |}_{1}\leq m}%
\widetilde{g}_{\boldsymbol{\ell }}(\boldsymbol{x}).
\label{EQ:ReLUapproximator}
\end{equation}

\begin{assumption}
\label{ass3}
\textcolor{black}{Let $p_X(\boldsymbol{x})$ be the density function of $\mu_X(\boldsymbol{x})$. Assume that for all $\boldsymbol{x}\in \mathcal{X}$, $0\leq p_{X}(\boldsymbol{x})\leq c_{\mu }$ for some constant $c_{\mu }\in (0,\infty )$. }
\end{assumption}

\begin{remark}
    \textcolor{black}{We assume the density function $p_X(\boldsymbol{x})$ is upper bounded by $c_\mu$. This condition is easily met and is used to establish approximation error bounds.}
\end{remark}

The following proposition provides the approximation error of the
approximator $\widetilde{f}\left( \cdot \right) $ obtained from the ReLU
network to the true unknown function $f\left( \cdot \right) $.

\begin{proposition}
\label{THM:error_ReLU}For any $f\in W^{2,p}(\mathcal{X)}$, $2\leq p\leq
\infty $, under Assumption \ref{ass3}, one has that for $d\geq 2$,
\begin{equation*}
||\widetilde{f}-f||_{2}\leq \left\{ (3/2)2^{-2R}+6^{-1}c_{\mu
}2^{-2m}\{(2/3)(m+3)\}^{d-1}\right\} ||D^{\boldsymbol{2}}f||_{L^{2}}.
\end{equation*}%
\textcolor{black}{The ReLU network that is used to construct the approximator $\widetilde{f}$
has the number of computational units \textcolor{black}{$%
\mathcal{O}(
2^{m}d^{\frac{3}{2}}R ( 2e\frac{m+d}{d-1}) ^{d-1}) $, the number
of weights $ \mathcal{O}%
( 2^{m}d^{\frac{3}{2}}R ( 2e\frac{m+d}{d-1}) ^{d-1}) $, and depth $\mathcal{O}(R\log _{2}d)$.} }
\end{proposition}

\begin{remark}
\textcolor{black}{From the mathematical expression (\ref%
{EQ:ReLUapproximator}) and the construction of $\widetilde{\varphi}_R$}, we see that the approximator $\widetilde{f}(%
\boldsymbol{\cdot })$ of the unknown function $f(\boldsymbol{\cdot })$ is
constructed from a sparse deep ReLU network, as the nodes on each layer are
not fully connected with the nodes from the previous layer, and the depth of
the network has the order of $R\log _{2}d$ which increases with $R$.
\end{remark}

\begin{remark}
\cite{MD19} showed that the approximation error of the deep ReLU network can
achieve accuracy $\epsilon >0$. We further derive an explicit form of the
bound to see how it depends on the dimension and the network complexity. In
Theorem \ref{THM:rate}, we will show that $m$ and $R$ need to grow with the
sample size $n$ slowly at a logarithmic rate to achieve a tradeoff between
bias and variance, so the depth of the ReLU network grows with $n$ at a
logarithmic rate, and the number of computational units increases with $n$ at
a polynomial rate.
\end{remark}

\section{\textcolor{black}{Sparse deep ReLU network estimator\label{SEC:SDRN algrithom}}}
In this section, we introduce the Sparse Deep ReLU Network estimator (SDRN)  for the regression function (\ref%
{EQ:f0}). 

\subsection{SDRN architecture}
\label{sec: SDRN architecture}
Our SDRN is constructed based on 3 subnetworks given in Figures \ref{Fig:plot_func}-\ref{Fig:f} obtained from 3 steps presented as follows. Step 4 gives the output as a weighted linear combination of the network nodes. Specifically, Step 1 generates subnetwork 1 consisting of fully-connected feedforward neural networks. step 2 builds a sparse network (subnetwork 2) as a simple linear combination of the subnetworks from Step 1. Step 3 obtains the network nodes in the last hidden layer constructed from a binary tree. In Step 4, the output is obtained from a linear combination of the network nodes obtained in Step 3. Our considered SDRN can go deeper with fewer parameters to well approximate a smooth function as discussed in Section \ref{SEC:NN}, and it can overcome the overfitting problem \cite{SHKSS2014} encountered by the conventional FNNs.

\iffalse
\textcolor{red}{In this subsection, we provide procedures and pseudocodes for SDRN construction, which is based on 4 steps, as illustrated in Figures \ref{Fig:plot_func}-\ref{Fig:f}. The first 3 steps provide the details of constructing each sub-ReLU neural network and the fourth step involves the linearly combination of these sub-ReLU neural networks.}
\fi

\textbf{Step 1:} for any input $u\in \left[ 0,1\right] $, we first generate
a fully-connected feedforward neural network with $r$ hidden layers and $3$
nodes in each layer, denoted as
\begin{equation*}
g_{r}(u,\boldsymbol{\theta }_{r})=\sigma (\boldsymbol{\theta }_{r1}\sigma
\left( \boldsymbol{\theta }_{r-1,1}\cdots \sigma (\boldsymbol{\theta }%
_{21}\sigma (\boldsymbol{\theta }_{11}u+\boldsymbol{\theta }_{10})+%
\boldsymbol{\theta }_{20})+\cdots \boldsymbol{\theta }_{r-1,0}\right) +%
\boldsymbol{\theta }_{r0}), 
\end{equation*}%
where $\boldsymbol{\theta }_{r}=$vec$(\boldsymbol{\theta }_{10},\boldsymbol{%
\theta }_{11},...,\boldsymbol{\theta }_{r0,}\boldsymbol{\theta }_{r1})$, in which the elements  $\boldsymbol{\theta }_{j1}$ are $3\times 3$ matrices for $2\leq j\leq
r$, while the element $\boldsymbol{\theta }_{11}$ and $\boldsymbol{\theta }_{j0}$, for $1\leq j\leq r$, are $3\times 1$ vectors. \textcolor{black}{We further introduce a ReLU neural network $\varphi_R(u, \boldsymbol{\theta}, \boldsymbol{\eta})$ constructed by $g_r(u, \boldsymbol{\theta}_r)$, $r=1, \dots, R$, given by 
$\varphi _{R}(u,\boldsymbol{\theta },\boldsymbol{\eta })=\eta_{0}u+\sum\nolimits_{r=1}^{R} \boldsymbol{\eta}_{r}^{\top }g_{r}(u,\boldsymbol{\theta }_{r}),$
 where $\boldsymbol{\eta}_r =2^{-2r}\widetilde{\boldsymbol{\eta}}$, $\boldsymbol{\eta} = \{\eta_0, \widetilde{\boldsymbol{\eta}}^\top \}^\top$ is a 4-dimensional vector 
%where $\boldsymbol{\eta }=(\eta _{0}, \boldsymbol{\eta}_{1}^{\top },..., \boldsymbol{\eta}_{R}^{\top })^{\top }$ is a $(3R+1)$-dimensional vector 
and $\boldsymbol{\theta }=$vec$(\boldsymbol{\theta }_{10},\boldsymbol{\theta }_{11},...,%
\boldsymbol{\theta }_{R0,}\boldsymbol{\theta }_{R1})$ is a $(12R-6)$%
-dimensional vector. This network structure is depicted in Figure \ref{Fig:plot_func}. We denote $\varphi_R(u, \boldsymbol{\theta}, \boldsymbol{\eta})$ as Sub1.}

\begin{figure}[tbp]
\centering
\vspace{0cm} $\includegraphics[scale=0.55]{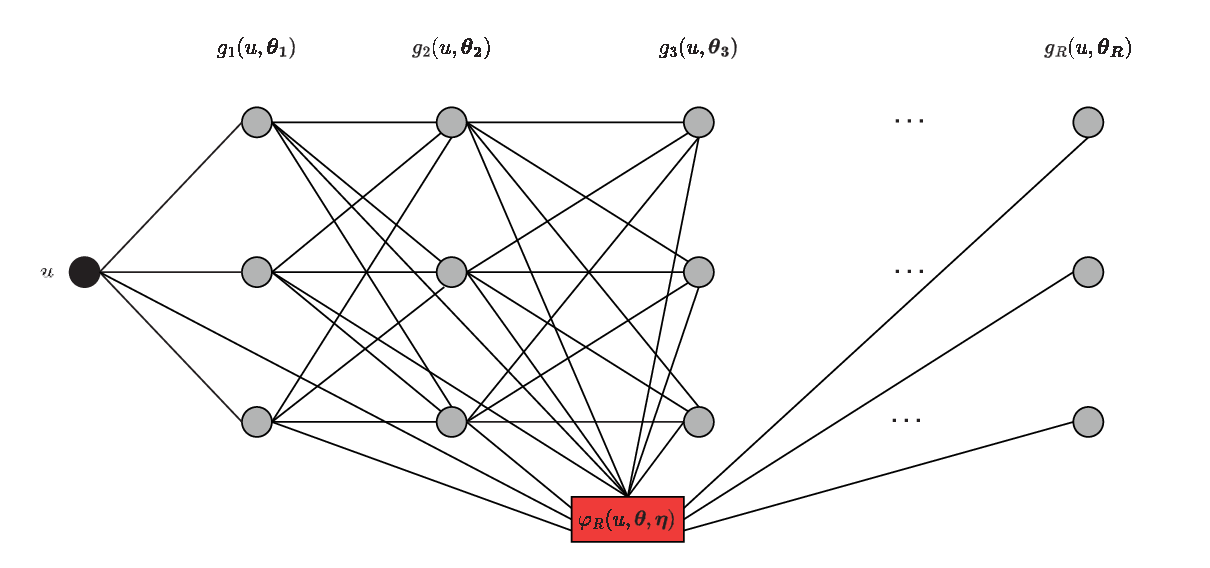}$%
\caption{The construction of the function $\varphi_{R}(u, \boldsymbol{\theta}, \boldsymbol{\eta} )$ by a ReLU network, denoted as Sub1.}%, denoted as subnetwork 1 (Sub1).}
\label{Fig:plot_func}
\end{figure}

\textbf{Step 2:} \textcolor{black}{for any two-dimensional inputs $u_{1},u_{2}\in \left[ 0,1%
\right] ^{2}$, based on the ReLU network Sub1, we build a sparse ReLU network $%
\widetilde{\varphi }_{R}(u_{1},u_{2},\boldsymbol{\theta },\boldsymbol{\eta },%
\boldsymbol{\omega })=\omega _{1}\varphi _{R}(u_{1},\boldsymbol{\theta },%
\boldsymbol{\eta })+\omega _{2}\varphi _{R}(u_{2},\boldsymbol{\theta },%
\boldsymbol{\eta })+\omega _{3}\varphi _{R}((u_{1}+u_{2})/2,%
\boldsymbol{\theta },\boldsymbol{\eta })$, where $\boldsymbol{\omega }=(\omega _{1},\omega _{2},\omega _{3})^{\top }$. The corresponding network is shown in Figure \ref{Fig:phitilda} and is denoted as Sub2.}

\begin{figure}[tbp]
\centering
\vspace{0cm} $\includegraphics[scale=0.6]{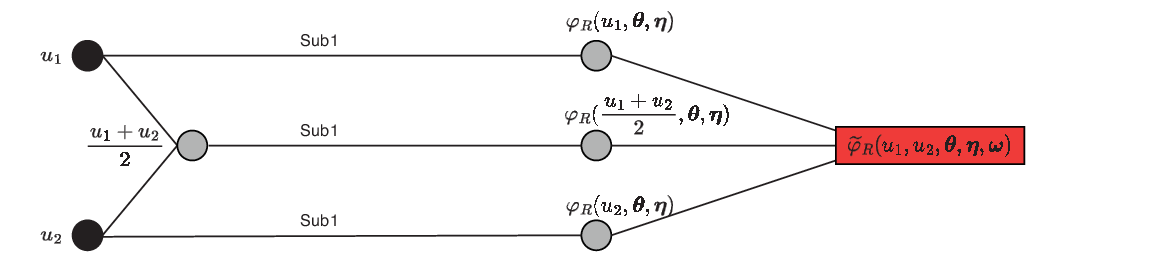}$
\caption{The construction of $\protect\widetilde{\varphi}_{R}(u_1,u_2, \boldsymbol{\theta}, \boldsymbol{\eta}, \boldsymbol{\omega})$ from the Sub1's, denoted as Sub2.}
\label{Fig:phitilda}
\end{figure}

\textbf{Step 3:} \textcolor{black}{for any $d$-dimensional inputs $\boldsymbol{u}\in \left[ 0,1%
\right] ^{d}$, where $\boldsymbol{u}=(u_1, u_2,\dots,u_d)^{\top }$, we construct the sub-ReLU neural network $\widetilde{\varphi }_R%
(\boldsymbol{u},\boldsymbol{\theta },\boldsymbol{\eta },\boldsymbol{%
\omega })$ based on a binary tree structure. This involves using Sub2 to integrate pairs of nodes from the preceding layer. Figure \ref{Fig:phtilda} shows the construction of $\widetilde{\varphi}_R(\boldsymbol{u}, \boldsymbol{\theta}, \boldsymbol{\eta}, \boldsymbol{\omega})$, clearly illustrating that $\widetilde{\varphi}_R(\boldsymbol{u}, \boldsymbol{\theta}, \boldsymbol{\eta}, \boldsymbol{\omega})$ is a ReLU neural network with sparse connections. We denote $\widetilde{\varphi}_R(\boldsymbol{u}, \boldsymbol{\theta}, \boldsymbol{\eta}, \boldsymbol{\omega})$ as Sub3.}

\begin{figure}[tbp]
\centering
\vspace{0cm} $\includegraphics[scale=0.6]{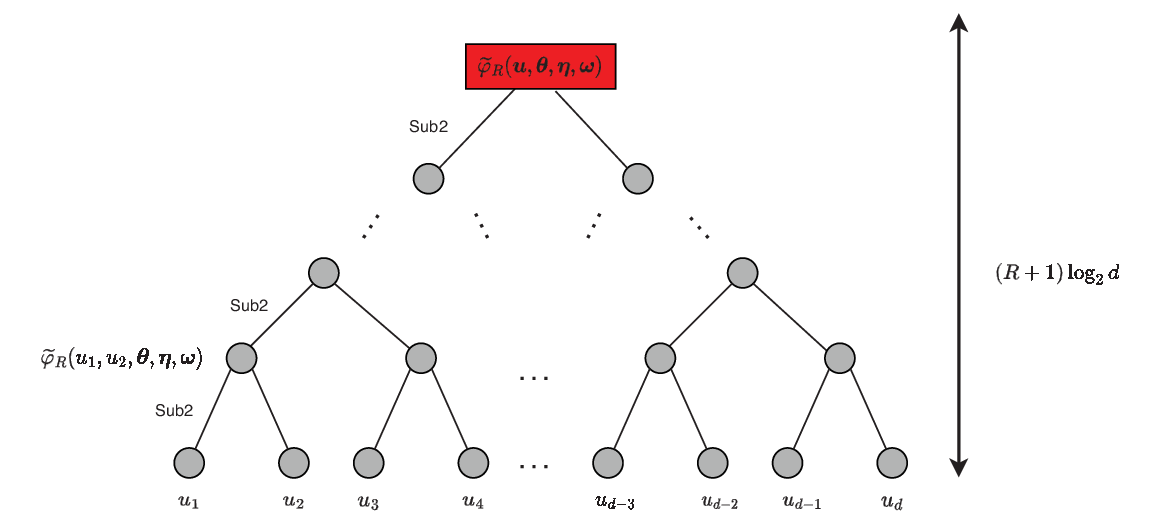}$
\caption{The construction of $\widetilde{\varphi}_R(\boldsymbol{u}, \boldsymbol{\theta}, \boldsymbol{\eta}, \boldsymbol{\omega})$ from the Sub2's, denoted as Sub3.}
\label{Fig:phtilda}
\end{figure}

\textbf{Step 4:} with d-dimensional inputs $\boldsymbol{x} = (x_1, x_2, \dots, x_d)$, letting \textbf{\ }$u_j=\phi _{\ell _{j},s_{j}}(x_{j})$ for $1\leq j\leq d$, and \textcolor{black}{ $\boldsymbol{u} = \boldsymbol{\phi}_{\boldsymbol{\ell}, \boldsymbol{s}} = (\phi _{\ell _{1},s_{1}}(x_{1}), \dots, \phi _{\ell _{d},s_{d}}(x_{d}))^\top$,} then the output from the SDRN is built based on the linear combination of \textcolor{black}{ all sub-ReLU networks $\widetilde{\varphi }_R(\boldsymbol{\phi}_{\boldsymbol{\ell}, \boldsymbol{s}}, \boldsymbol{\theta }, \boldsymbol{\eta }, \boldsymbol{\omega })$, and can be expressed as
\begin{equation*}
f(\boldsymbol{x},\boldsymbol{\theta },\boldsymbol{\eta },\boldsymbol{\omega }%
,\boldsymbol{\gamma }) = \sum\nolimits_{|\boldsymbol{\ell |}_{1}\leq
m}\sum\nolimits_{\boldsymbol{s}\in I_{\boldsymbol{\ell }}}\gamma _{_{\boldsymbol{\ell }, \boldsymbol{s}}}\textcolor{black}{\widetilde{\varphi }_R(\boldsymbol{\phi} _{\boldsymbol{\ell },\boldsymbol{s}},\boldsymbol{\theta },\boldsymbol{\eta }, \boldsymbol{\omega })} 
\end{equation*}}%
, where $\boldsymbol{\gamma }=\{\gamma _{_{\boldsymbol{\ell },\boldsymbol{s}%
}}:\boldsymbol{s}\in I_{\boldsymbol{\ell }},|\boldsymbol{\ell |}_{1}\leq m\}^{\top }$, Figure \ref{Fig:f} shows the resulting network architecture.

\begin{figure}[tbp]
\centering
\vspace{0cm} $\includegraphics[scale=0.5]{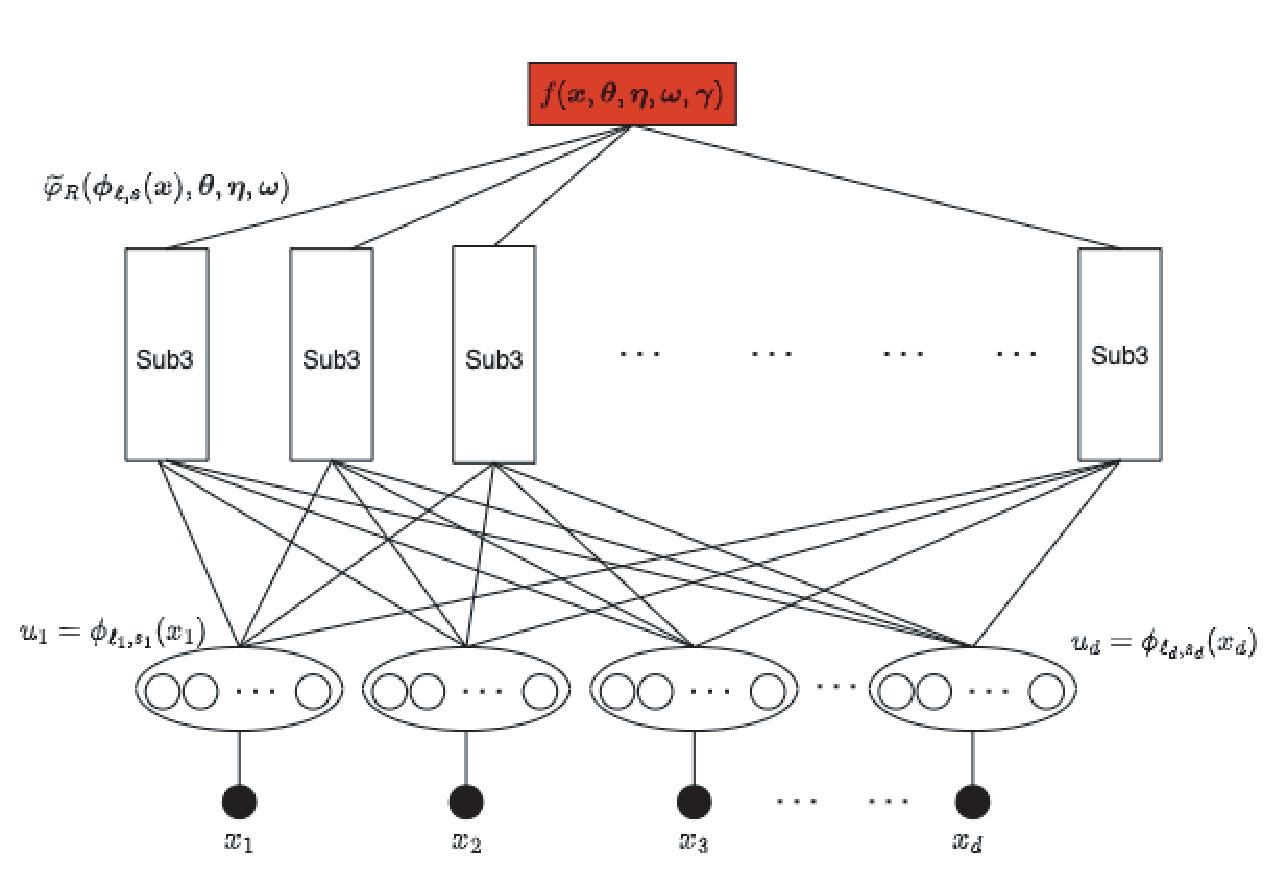}$
\caption{The construction of $f(\boldsymbol{x},\boldsymbol{\theta },\boldsymbol{\eta },\boldsymbol{\omega }
,\boldsymbol{\gamma })$ from the Sub3's.}
\label{Fig:f}
\end{figure}

\textcolor{black}{We provide a pseudocode for constructing the SDRN based on the previous 4 steps in Algorithm \ref{alg:pseudocode}.}
\begin{algorithm}
\caption{\textcolor{black}{Pseudocode to construct SDRN, refer to Section \ref{sec: SDRN architecture} for detailed explanation.  Let $\boldsymbol{u}=\boldsymbol{\phi}_{\boldsymbol{\ell}, \boldsymbol{s}}(\boldsymbol{x}) = (\phi_{l_1, s_1}(x_1), \phi_{l_2, s_2}(x_2), \dots, \phi_{l_d, s_d}(x_d))$  be a $d$-dimensional vector and denote the last element in $\boldsymbol{u}$ as $\boldsymbol{u}[-1]$, and $u$, $u_1$, $u_2$ are scalars.}}\label{alg:pseudocode}    
\begin{algorithmic}
    \Require Input $\boldsymbol{x}_{n \times d}$: $d$-dimensional data with sample size $n$ 
    \Require $m$: control the number of basis functions
    %\Require $\{\phi_{\boldsymbol{\ell}, \boldsymbol{s}}(\boldsymbol{x})$, $|\boldsymbol{\ell}|_1 \le m$, $\boldsymbol{s} \in I_{\boldsymbol{\ell}}$\}: a set of basis functions with each $ \phi _{\boldsymbol{\ell}, \boldsymbol{s}} = (\phi_{\ell_1, s_1}(x_1), \dots, \phi_{l_d, s_d}(x_d))^{\top}$
    \Require  $\boldsymbol{\Theta }=(\boldsymbol{\theta }^{\top }, \boldsymbol{\eta }^{\top },
    \boldsymbol{\omega }^{\top },\boldsymbol{\gamma }^{\top })^{\top }$: parameter vector
\State $R=3\max(\lfloor 0.2 \log_2 n \rfloor, m)$ (number of hidden layers in $\varphi_R(u, \boldsymbol{\theta}, \boldsymbol{\eta})$)
\State
\Function{$\varphi_R$}{$u, \boldsymbol{\theta}, \boldsymbol{\eta}$} (Sub1)
\State $s_r \gets \eta_0 \cdot u$ (Initialize the starting sum value)
\State $h_r \gets u$ (layer's input)
\For{$r$ from $1$ to $R$} (Iterate from the first hidden layer to the last)
    \State $g_r \gets \sigma(\boldsymbol{\theta}_{r1} \cdot h_r + \boldsymbol{\theta}_{r0})$  (Compute the layer's output)
    \State $s_r \gets s_r + \boldsymbol{\eta}_r^{\top} \cdot g_r$ (Add the weighted output of this layer)
    \State  $h_r \gets g_r$ (Update layer's input for the next layer)
\EndFor
    \State \textbf{return} $s_r$
\EndFunction
\State
\Function{$\widetilde{\varphi}_R$}{$u_1, u_2, \boldsymbol{\theta}, \boldsymbol{\eta}, \boldsymbol{\omega}$}
 (Sub2)
\State \textbf{return}  $ \omega_1 \varphi_R(u_1, \boldsymbol{\theta}, \boldsymbol{
\eta}) + \omega_2 \varphi_R(u_2, \boldsymbol{\theta}, \boldsymbol{
\eta}) + \omega_3 \varphi_R(\frac{u_1+u_2}{2}, \boldsymbol{\theta}, \boldsymbol{
\eta})$ 
\EndFunction

\State
    \Function{$\widetilde{\varphi}_R$}{$\boldsymbol{u}, \boldsymbol{\theta}, \boldsymbol{\eta}, \boldsymbol{\omega}$} (Sub3: binary tree structure)
    %\State Given a basis $\boldsymbol{u} \gets \phi_{\boldsymbol{\ell}, \boldsymbol{s}}(\boldsymbol{x}) = (\phi_{\ell_1, s_1}(x_1), \dots, \phi_{l_d, s_d}(x_d))^{\top}$
    \State $p \gets$ length($\boldsymbol{u}$) ($p$ represents the number of elements in $\boldsymbol{u}$)
    \While {$p \ge 2$}
    \State $\boldsymbol{u}^{\prime} \gets [\text{ }]$ (Initialize an empty vector $\boldsymbol{u}^{\prime}$ to store elements in next layer)
    \For{$j$ from $1$ to $\lfloor p/2 \rfloor$}
    \State $\boldsymbol{u}^{\prime}[j] \gets \widetilde{\varphi}_R(\boldsymbol{u}[2j-1], \boldsymbol{u}[2j], \boldsymbol{\theta}, \boldsymbol{\eta}, \boldsymbol{\omega})$ 
    \State (Call Sub2 to combine the pair of elements for the next layer $\boldsymbol{u}^{\prime}$)
    \EndFor
    \If{$p$ is odd}
    \State $\boldsymbol{u}^{\prime} \gets [\boldsymbol{u}^{\prime} \quad \boldsymbol{u}[-1]]$ 
    \State(If there's an unpaired element left, carry it over to $\boldsymbol{u}^{\prime}$)
    \EndIf
    \State $\boldsymbol{u} \gets \boldsymbol{u}^{\prime}$ (Replace $\boldsymbol{u}$ with the new layer, continue to the next layer)
    \State $p \gets \text{length}(\boldsymbol{u})$ (Update $p$ to the current length of $\boldsymbol{u}$)
    \EndWhile
    \State \textbf{return}  $\boldsymbol{u}$  
    \EndFunction
\State
\State \textbf{Output} $f(\boldsymbol{x}, \boldsymbol{\theta}, \boldsymbol{\eta}, \boldsymbol{\omega}, \boldsymbol{\gamma}) \gets \sum\nolimits_{|\boldsymbol{\ell |}_{1}\leq
m}\sum\nolimits_{s\in I_{\boldsymbol{\ell }}}\gamma _{_{\boldsymbol{\ell }, \boldsymbol{s}}}\widetilde{\varphi }_R(\boldsymbol{\phi} _{\boldsymbol{\ell },\boldsymbol{s}}(\boldsymbol{x}),\boldsymbol{\theta },\boldsymbol{\eta }, \boldsymbol{\omega })$
\State (Call Sub3 for each $\boldsymbol{\phi}_{\boldsymbol{\ell}, \boldsymbol{s}}(\boldsymbol{x})$, then linearly combine $\widetilde{\varphi}_R(\boldsymbol{\phi}_{\boldsymbol{\ell}, \boldsymbol{s}}(\boldsymbol{x}), \boldsymbol{\theta}, \boldsymbol{\eta}, \boldsymbol{\omega}))$
\end{algorithmic}
\end{algorithm}

\subsection{SDRN estimator}

We denote the vector of the whole parameters in $f(\boldsymbol{x},%
\boldsymbol{\theta },\boldsymbol{\eta },\boldsymbol{\omega },\boldsymbol{%
\gamma })$ as $\boldsymbol{\Theta }=(\boldsymbol{\theta }^{\top },%
\boldsymbol{\eta }^{\top },\boldsymbol{\omega }^{\top },\boldsymbol{\gamma }%
^{\top })^{\top }$. Then the SDRN estimator can be written as $f(\boldsymbol{x},%
\boldsymbol{\Theta })=f(\boldsymbol{x},\boldsymbol{\theta },\boldsymbol{\eta
},\boldsymbol{\omega },\boldsymbol{\gamma })$.  Since the ReLU network
          function $\widetilde{f}$ in (\ref{EQ:ReLUapproximator}) is constructed in
the same way as above, there exists a set of parameters $\boldsymbol{\Theta }%
^{0}=(\boldsymbol{\theta }^{0\top },\boldsymbol{\eta }^{0\top },\boldsymbol{%
\omega }^{0\top },\boldsymbol{\gamma }^{0\top })^{\top }$ such that $\widetilde{f}$
can be written $\widetilde{f}(\boldsymbol{x},\boldsymbol{\Theta }^{0})$. Define the ReLU network class as
\textcolor{black}{
\begin{eqnarray}
\mathcal{F(}R ,m,B_{0},B_{1}) &=&\left\{ f:\mathcal{X}\rightarrow \mathbb{%
R},f(\boldsymbol{x},\boldsymbol{\Theta })=\sum\nolimits_{|\boldsymbol{\ell |}%
_{1}\leq m}\sum\nolimits_{s\in I_{\boldsymbol{\ell }}}\gamma _{_{\boldsymbol{%
\ell },\boldsymbol{s}}}\widetilde{\varphi }_R(\boldsymbol{\phi} _{\boldsymbol{\ell },%
\boldsymbol{s}}(\boldsymbol{x}),\boldsymbol{\theta },\boldsymbol{\eta },%
\boldsymbol{\omega }),\right.  \notag \\
&&\left. \gamma _{_{\boldsymbol{\ell },\boldsymbol{s}}}\in
\mathbb{R}, ||f||_{\infty }\leq B_{0},|\boldsymbol{\Theta }|_{2}\leq
B_{1}\right\} ,  \label{DEF:F}
\end{eqnarray}%
}
with $B_{0}\geq \max (||f_{0}||_{\infty },||\widetilde{f}||_{\infty })$ and $%
B_{1}\geq |\boldsymbol{\Theta }^{0}|_{2}$. Then $\widetilde{f}\in \textcolor{black}{ \mathcal{F(%
}R ,m,B_{0},B_{1})}$.

Define the empirical risk as $\mathcal{E}_{n}(f;\boldsymbol{\Theta }%
)=n^{-1}\sum_{i=1}^{n}\rho (f(\boldsymbol{X}_{i},\boldsymbol{\Theta }),Y_{i})
$, and the regularized empirical risk as $\mathcal{E}_{n}^{P}(f;\boldsymbol{%
\Theta })=n^{-1}\sum_{i=1}^{n}\rho (f(\boldsymbol{X}_{i},\boldsymbol{\Theta }%
),Y_{i})+2^{-1}\lambda \boldsymbol{\Theta }^{\top }\boldsymbol{\Theta }$,
where $\lambda >0$ is a tuning parameter for the $L_{2}$ (ridge) penalty.
The $L_{2}$ penalty is often used to prevent over-fitting. When $\lambda =0$%
, the regularized empirical risk is reduced to be the empirical risk. Then,
the penalized SDRN estimator $\widehat{f}$ of $f_{0}$ satisfies%
\begin{equation}
\mathcal{E}_{n}^{P}(\widehat{f};\widehat{\boldsymbol{\Theta }})\leq
\min_{f\in \textcolor{black}{ \mathcal{F(%
}R ,m,B_{0},B_{1})}}\{\mathcal{E}_{n}^{P}(f;%
\boldsymbol{\Theta })\}+\varpi _{n},  \label{EQ:ReLUestimator}
\end{equation}%
where $\widehat{f}(\boldsymbol{x})=f(\boldsymbol{x},\widehat{\boldsymbol{%
\Theta }})$, and $\varpi _{n}=o(1)$ satisfying a condition given in Theorem %
\ref{THM:sampliingerror1}. Note that the penalized SDRN estimator $\widehat{f%
}$ obtained from (\ref{EQ:ReLUestimator}) does not need to be the global
minimizer of the objective function $\mathcal{E}_{n}^{P}(f;\boldsymbol{%
\Theta })$; it can be any local solution such that the difference of the
objective function values evaluated at the local solution and the global
minimizer is bounded by a small term $\varpi _{n}$. We adopt the Adam algorithm given in
\cite{KingmaBa2015} for obtaining the estimate of $\boldsymbol{\Theta }$. This algorithm considers first-order gradient-based optimization,
and it is straightforward to implement and has little memory requirements. We use the default settings for the hyperparameters used in the Adam algorithm given in \cite{KingmaBa2015}.

\section{Theory\label{SEC:ReLUestimator}}

For a given estimator $\widehat{f}$, we define the overall error as $%
\mathcal{E}(\widehat{f})-\mathcal{E}(f_{0})$, which is used to measure how
close the estimator $\widehat{f}$ to the true target function $f_{0}$. Let
\begin{equation}
f^{\ast }=\arg \min_{f\in \textcolor{black}{ \mathcal{F(%
}R ,m,B_{0},B_{1})}}\mathcal{E}(f)%
\text{, where }\mathcal{E}(f)=\int_{\mathcal{X\times Y}}\rho (f(\boldsymbol{x%
},\boldsymbol{\Theta }),y)d\mu (\boldsymbol{x,}y).  \label{EQ:f0RN}
\end{equation}%
Then the overall error of the estimator $\widehat{f}$ can be split into
the approximation error $\mathcal{E}(f^{\ast })-\mathcal{E}(f_{0})$ and the
sampling error $\mathcal{E}(\widehat{f})-\mathcal{E}(f^{\ast })$ such that
\begin{equation*}
\underset{\text{overall error}}{\underbrace{\mathcal{E}(\widehat{f})-%
\mathcal{E}(f_{0})}}=\underset{\text{approximation error}}{\underbrace{%
\mathcal{E}(f^{\ast })-\mathcal{E}(f_{0})}}+\underset{\text{estimation error}%
}{\underbrace{\mathcal{E}(\widehat{f})-\mathcal{E}(f^{\ast })}}.
\end{equation*}%
We will establish the upper bounds for the approximation error and the
estimation error, respectively, as follows.

We introduce the following Bernstein condition that is required for
obtaining the probability bound for the estimation error of our SDRN
estimator.

\begin{assumption}
\label{ass4}There exists a constant $0<a_{\rho }<\infty $ such that
\begin{equation}
a_{\rho }||f-f^{\ast }||_{2}^{2}\leq \mathcal{E}(f)-\mathcal{E}(f^{\ast })
\label{Bernstein}
\end{equation}%
for any $f\in \textcolor{black}{ \mathcal{F(%
}R ,m,B_{0},B_{1})}$.
\end{assumption}

\begin{remark}
The Bernstein condition given in (\ref{Bernstein}) for Lipschitz loss
functions is used in the literature in order to establish probability bounds
of estimators obtained from empirical risk minimization \cite{ACL19}. A
more general form is $a_{\rho }||f-f^{\ast }||_{2}^{2\kappa }\leq \mathcal{E}%
(f)-\mathcal{E}(f^{\ast })$ for some $\kappa \geq 1$. The parameter $\kappa $
can affect the estimator's rate of convergence. For proof convenience, we
let $\kappa =1$ which is satisfied by many commonly used loss functions. We
will give a detailed discussion on this Bernstein condition, and will
present different examples in Section \ref{discussion_ass}. 
\end{remark}

\begin{remark}
\label{rho}From the Lipschitz condition given in Assumption \ref{ass2}, we
have that there exists a constant $0<M_{\rho }<\infty $ such that $|\rho
\left( f(\boldsymbol{x}),y\right) |\leq M_{\rho }$, for almost every $(%
\boldsymbol{x},y)\in \mathcal{X}\times \mathcal{Y}$ and any $f\in \textcolor{black}{ \mathcal{F(%
}R ,m,B_{0},B_{1})}$.
\end{remark}

Another condition is given below and it is used for controlling the
approximation error from the ReLU networks.

\begin{assumption}
\label{ass5}There exists a constant $0<b_{\rho }<\infty $ such that
\begin{equation}
\mathcal{E}(f)-\mathcal{E}(f_{0})\leq b_{\rho }||f-f_{0}||_{2}^{2}
\label{Bernstein2}
\end{equation}%
for any $f\in \textcolor{black}{ \mathcal{F(%
}R ,m,B_{0},B_{1})}$.
\end{assumption}

\begin{remark}
Assumption \ref{ass5} is introduced for controlling the approximation error $%
\mathcal{E}(f^{\ast })-\mathcal{E}(f_{0})$, but it is not required for
establishing the upper bound of the sampling error $\mathcal{E}(\widehat{f})-%
\mathcal{E}(f^{\ast })$. The approximation error $\mathcal{E}(f^{\ast })-%
\mathcal{E}(f_{0})$ can be well controlled based on the result from
Proposition \ref{THM:error_ReLU} together with Assumption \ref{ass5}.
Without this assumption, the approximation error will have a slower rate.
Assumption \ref{ass5} is satisfied by the quadratic, logistic, quantile and
Huber loss functions under mild conditions. More discussions on this
assumption will be provided in Section \ref{discussion_ass}. 
\end{remark}

Under Condition (\ref{Bernstein2}) given in Assumption \ref{ass5}, by the
definition of $f^{\ast }$ given in (\ref{EQ:f0RN}) and Proposition \ref%
{THM:error_ReLU}, the approximation error
\begin{equation*}
\mathcal{E}(f^{\ast })-\mathcal{E}(f_{0})\leq \mathcal{E}(\widetilde{f})-%
\mathcal{E}(f_{0})\leq b_{\rho }||\widetilde{f}-f_{0}||_{2}^{2}.
\end{equation*}%
Since $f_{0}$ satisfies Assumption \ref{ass1}, then $||D^{\boldsymbol{2}%
}f_{0}||_{L^{2}}\leq C_{f}$ for some constant $C_{f}\in (0,\infty )$. Next
proposition presents an upper bound for the approximation error when the
unknown function $f_{0}$ is approximated by the SDRN obtained from the ERM
in (\ref{EQ:f0RN}).

\begin{proposition}
\label{THM:error_ReLUbound}Under Assumptions \ref{ass1}, \ref{ass3} and \ref%
{ass5}, and $R\geq m$ and $d\geq 2$, one has%
\begin{equation*}
\mathcal{E}(f^{\ast })-\mathcal{E}(f_{0})\leq \zeta _{m,d},
\end{equation*}%
where
\begin{equation}
\zeta _{m,d}=4^{-1}b_{\rho }(3+c_{\mu }/3)^{2}2^{-4m}\{(2/3)(m+3)\}^{2\left(
d-1\right) }||D^{\boldsymbol{2}}f_0||_{L^{2}}^{2},  \label{XiRmd}
\end{equation}
in which $c_{\mu }$ and $b_{\rho }$ are given in Assumptions \ref{ass3} and %
\ref{ass5}, respectively.
\end{proposition}
Note that without Assumption \ref{ass5}, we obtain a looser bound for $%
\mathcal{E}(f^{\ast })-\mathcal{E}(f_{0})=\mathcal{O(}\zeta _{m,d}^{1/2})$
based on the result $\mathcal{E}(\widetilde{f})-\mathcal{E}(f_{0})\leq
C_{\rho }||\widetilde{f}-f_{0}||_{2}$ which is directly implied from
Assumption \ref{ass2}.

Next we establish the bound for the sampling error $\mathcal{E}(\widehat{f})-%
\mathcal{E}(f^{\ast })$. Let $\mathcal{N}(\delta ,\mathcal{F},||\cdot
||_{\infty })$ be the covering number, that is, the minimal number of $%
||\cdot ||_{\infty }$- balls with radius $\delta $ that covers $\mathcal{F}$
and whose centers reside in $\mathcal{F}$. In the theorem below, we provide
an upper bound for the estimation error $\mathcal{E}(\widehat{f})-\mathcal{E}%
(f^{\ast })$.

\begin{theorem}
\label{THM:sampliingerror1}Under Assumptions \ref{ass1}-\ref{ass4}, we have
that for any $\epsilon >0$ and $\varpi _{n}+\lambda B_{1}^{2}<(1/2)\epsilon $%
,
\begin{equation*}
P\left\{ \mathcal{E}(\widehat{f})-\mathcal{E}(f^{\ast })>(3/2)\epsilon
\right\} \leq \mathcal{N}(\sqrt{2}C_{\rho }^{-1}\epsilon /8,\textcolor{black}{ \mathcal{F(%
}R ,m,B_{0},B_{1})},||\cdot ||_{\infty })\exp \left( -n\epsilon /C^{\ast
}\right)
\end{equation*}%
, where $C^{\ast }=64(C_{\rho }^{2}a_{\rho }^{-1}+4M_{\rho }/3)$, in which $%
C_{\rho },a_{\rho }$ and $M_{\rho }$ are constants given in Assumptions \ref%
{ass2} and \ref{ass4} and Remark \ref{rho}.
\end{theorem}

\begin{theorem}
\label{THM:sampliingerror2}Under the same assumptions as given in Theorem %
\ref{THM:sampliingerror1}, there exist constants $c,C\in (0,\infty )$ such
that
\begin{equation*}
P\left( \mathcal{E}(\widehat{f})-\mathcal{E}(f^{\ast })>\frac{3C^{\ast
}CWL\log (W)}{2n}\max (1,\log \frac{C^{\ast \ast }n}{cWL\log (W/L)\varsigma }%
)\right) \leq \varsigma .
\end{equation*}%
where $C^{\ast }$ is given in Theorem \ref{THM:sampliingerror1}, $C^{\ast \ast }=16C_{\rho }B_0C^{\ast -1}$, $\textcolor{black}{\varsigma \asymp  (\frac{16C_{\rho }B_0}{\epsilon })^{WL\log (W)} \exp \left(\frac{-n\epsilon}{C^{\ast }}\right)} $, in which
the number of parameters in the sparse deep ReLU network is $W\ \asymp |V_{m}^{\left( 1\right) }| +R$ and the number of layers is $L\asymp R\log
_{2}d$.
\end{theorem}

Based on the upper bound for the estimation error given in Theorem \ref%
{THM:sampliingerror2}, and the bound for the approximation error given in (%
\ref{XiRmd}), we can further obtain the risk rate of the SDRN estimator $%
\widehat{f}$ presented in the following theorems.

\begin{theorem}
\label{THM:rate}Under Assumptions \ref{ass1}-\ref{ass5}, $2^{m}\asymp n^{1/5}
$, $R\asymp \log _{2}n$ and $m\leq R$, when $d=\mathcal{O}\left( (\log
_{2}n)^{1-\kappa }\right) $ for an arbitrary small constant $\kappa >0$,
then the penalized SDRN estimator $\widehat{f}$ given in (\ref%
{EQ:ReLUestimator}) with \textcolor{black}{$\varpi _{n}=\mathcal{O}(n^{-\frac{4}{5}+ \frac{\nu}{2}}(\log
_{2}n)^{\frac{3\kappa}{2} +2})\mathcal{\ }$} and \textcolor{black}{$\lambda =\mathcal{O}(n^{-\frac{4}{5}+\frac{\nu}{2}}(\log _{2}n)^{\frac{3\kappa}{2} +2})$} has the risk rate
\begin{equation*}
\mathcal{E}(\widehat{f})-\mathcal{E}(f_{0})=o_{p}(n^{-4/5+\nu }(\log
_{2}n)^{-2}), \text{ for an arbitrarily small $\nu >0$}. 
\end{equation*}%
The approximation error
satisfies $\mathcal{E}(f^{\ast })-\mathcal{E}(f_{0})=o(n^{-4/5+\nu }(\log
_{2}n)^{-2})$, and the estimation error satisfies $\mathcal{E}(\widehat{f})-%
\mathcal{E}(f^{\ast })=\mathcal{O}_{p}(n^{-4/5+\nu /2}(\log
_{2}n)^{7/2-\kappa /2})$. The ReLU network that is used to construct the
estimator $\widehat{f}$ has depth $\mathcal{O}[\log _{2}n \allowbreak \{\log _{2}(\log
_{2}n)\}]$, the number of computational units $\mathcal{O}\{\left( \log
_{2}n\right) ^{3/2(1-\kappa )}n^{1/5+\nu /2}\}$, and the number of weights $%
\mathcal{O}\{\left( \log _{2}n\right) ^{3/2(1-\kappa )}n^{1/5+\nu /2}\}$.
\end{theorem}

\begin{proposition}
\label{THM:rate_fixed}Under the same conditions in Theorem \ref%
{THM:sampliingerror2}, when $d\ $is fixed, then the penalized SDRN estimator
$\widehat{f}$ given in (\ref{EQ:ReLUestimator}) with $\varpi _{n} =\mathcal{O}%
(n^{-4/5+\nu /2} (\log _{2}n)^{3\kappa /2+2})\mathcal{\ }$ and $\lambda =%
\mathcal{O}(n^{-4/5+\nu /2}(\log _{2}n)^{3\kappa /2+2})$ has the risk rate
\begin{equation*}
\mathcal{E}(\widehat{f})-\mathcal{E}(f_{0})=\mathcal{O}_{p}(n^{-4/5}(\log
_{2}n)^{\left( d+3\right) \vee \left( 2d-2\right) }).
\end{equation*}%
Moreover, the approximation error satisfies $\mathcal{E}(f^{\ast })-\mathcal{%
E}(f_{0})=\mathcal{O}(n^{-4/5}(\log _{2}n)^{2d-2})$, and the estimation
error satisfies $\mathcal{E}(\widehat{f})-\mathcal{E}(f^{\ast })=\mathcal{O}%
_{p}(n^{-4/5}(\log _{2}n)^{d+3})$. The ReLU network that is used to
construct the estimator $\widehat{f}$ has depth $\mathcal{O}(\log _{2}n)$,
the number of computational units $\mathcal{O}\{\left( \log _{2}n\right)
^{d}n^{1/5}\}$, and the number of weights $\mathcal{O}\{\left( \log
_{2}n\right) ^{d}n^{1/5}\}$.
\end{proposition}
\iffalse
\begin{remark}
Note that Assumption \ref{ass5} is not required to obtain the
convergence rate of the sampling error $\mathcal{E}(\widehat{f})-\mathcal{E}%
(f^{0})$, it is only needed for the rate of the approximation error $%
\mathcal{E}(f^{\ast })-\mathcal{E}(f_{0})$. Without this assumption, the
rate of $\mathcal{E}(f^{\ast })-\mathcal{E}(f_{0})$ is slower.
\end{remark}
\fi
\begin{remark}
We focus on deriving the optimal risk rate for the SDRN estimator of the
unknown function $f_{0}$ when it belongs to the Korobov space of mixed
derivatives of order $\beta =2$. Then the derived rate can be written as $%
n^{-2\beta /(2\beta +1)}(\log _{2}n)^{2d}$ when $d$ is fixed. It is possible
to derive a similar estimator for a smoother regression function that has
mixed derivatives of order $\beta >2$ when Jacobi-weighted Korobov spaces %
\cite{SW10} are considered. This can be an interesting topic for the future
work.
\end{remark}

\begin{remark}
It is worth noting that for the classical nonparametric regression
estimators such as spline estimators \cite{S82}, the optimal minimax risk
rate is $n^{-4/(4+d)}$, if the regression function belongs to the Sobolev
spaces $S^{2,p}(\mathcal{X})$. This rate suffers from the curse of
dimensionality as $d$ increases.

\cite{BK2019} showed that their least squares neural network estimator can
achieve the rate $n^{-2\beta /(2\beta +d^{\ast })}$ (up to a log factor), if
the regression function satisfies a $\beta $-smooth generalized hierarchical
interaction model of order $d^{\ast }$. When $\beta =2$, the rate is $%
n^{-4/(4+d^{\ast })}$. The rates mentioned above require $d$ to be fixed.
\cite{BK2019} consider a smooth activation function, while \cite{Schmidt2020}
have established a similar optimal rate for ReLU activation function.

Proposition \ref{THM:rate_fixed} shows that when $f_{0}$ belongs to the Korobov spaces
$W^{2,p}(\mathcal{X)}$, our SDRN estimator has the risk rate $n^{-4/5}(\log
_{2}n)^{2d+1}$ and it achieves the optimal minimax rate (up to a log factor)
as one-dimensional nonparametric regression, if the dimension $d$ is fixed.
The effect of $d$ is passed on to a logarithm order, so the curse of
dimensionality can be alleviated. When $d$ increases with $n$ with an order $%
(\log _{2}n)^{1-\kappa }$, the risk rate is slightly slower than $n^{-4/5}$.

\cite{mao2022approximation} derived an approximation error for deep convolutional neural networks (DCNNs) when the target function belongs to Korobov spaces, but they did not provide an estimation error. Estimators obtained from ERM have two errors: the approximation error and the estimation error. The mean squared error used to evaluate the overall performance of a machine learning estimator comes from both errors. When their approximation error is in the same order as (similar to) ours, their network size needs to be significantly larger than ours, leading to a larger estimation error. Specifically, to reach the approximation accuracy, the total number of free parameters $\mathcal{N}$ in their DCNNs given in equation (2.2) of  \cite{mao2022approximation}  is
\begin{equation}
\label{eq: para_DCNN}
    \mathcal{N} \le 13385d^2(\log_2d)^2(log_2N)^2N\text{, with}\quad N \ge 2^{16}.
\end{equation}
By the construction of our proposed SDRN in Section \ref{sec: SDRN architecture}, the total number of free parameters in SDRN is bounded by
\begin{align}
|\boldsymbol{\Theta }|& =|\boldsymbol{\gamma }|+|\boldsymbol{\theta }|+|%
\boldsymbol{\eta }|+|\boldsymbol{\omega }|=|V_{m}^{(1)}|+|\boldsymbol{\theta
}|+|\boldsymbol{\eta }|+|\boldsymbol{\omega }|  \notag  \label{eq: para_SDRN}
\\
& \leq \sum\nolimits_{|\boldsymbol{\ell |}_{1}\leq m}2^{\sum_{j=1}^{d}\ell
_{j}\wedge 2-d}+12R+1,
\end{align}%
where $m=\left\lfloor 0.2\log _{2}n\right\rfloor$ and $R=3\max
(\left\lfloor 0.2\log _{2}n\right\rfloor ,m)$ used in our numerical
analysis. We see that our SDRN requires much fewer free parameters than DCNNs to achieve the same order of the approximation error. For example, for $n=2000$ and $d=10$, the network size of our SDRN is bounded by 67657, while it is $14770624\times(log_2N)^2N$ with $N \ge 2^{16}$ for DCNNs. Since the estimation error depends on the model complexity, our SDRN has a smaller estimation error and overall better performance.

\end{remark}

\begin{proposition}
\label{THM:rate_fixed_lower}Under the same conditions in Theorem \ref%
{THM:rate}, when the sample size is sufficiently large, the lower bound for
the overall error of the penalized SDRN estimator $\widehat{f}$ given in (%
\ref{EQ:ReLUestimator}) is given as
\begin{equation*}
P\left( \mathcal{E}(\widehat{f})-\mathcal{E}(f_{0})>c_{1}n^{-4/5}(%
{ \log _{2}n)^{2}\log _{2}d}\right) \geq 1/100,
\end{equation*}%
for some constant $c_{1}>0$.
\end{proposition}

\begin{remark}
Proposition \ref{THM:rate_fixed_lower} presents a lower bound for the
overall error of the SDRN estimator $\widehat{f}$. When $d$ is fixed, the
rate of the lower bound is comparable to the upper bound rate given in
Proposition \ref{THM:rate_fixed}, which is $n^{-4/5}(\log _{2}n)^{\left(
d+3\right) \vee \left( 2d-2\right) }$, but it is smaller than the upper
bound as one can show that $({\log _{2}n)^{2}\log _{2}d<}(\log
_{2}n)^{\left( d+3\right) \vee \left( 2d-2\right) }$. Thus, our SDRN
estimator achieves a nearly tight optimal risk rate.
\end{remark}

%%%%%%%%% adding simulation + real data

\section{Discussions on Assumptions \protect\ref{ass4} and \protect\ref{ass5}%
}

\label{discussion_ass}

We first state a general condition given in Assumption \ref{ass6} presented
below. We will show that if a loss function satisfies this condition, then
it will satisfy Assumption \ref{ass4} (Bernstein condition) and Assumption %
\ref{ass5}.

\begin{assumption}
\label{ass6}For all $y\in \mathcal{Y}$, the loss function $\rho \left( \cdot
,y\right) $ is strictly convex and it has a bounded second derivative such
that $\rho ^{\prime \prime }\left( \cdot ,y\right) \in \lbrack 2a_{\rho
},2b_{\rho }]$ almost everywhere, for some constants $0<a_{\rho }\leq
b_{\rho }<\infty $.
\end{assumption}

Assumption \ref{ass6} is satisfied by a variety of classical loss functions
such as quadratic loss and logistic loss. For example, for the quadratic
loss $\rho \left( f(\boldsymbol{x}),y\right) =(y-f(\boldsymbol{x}))^{2}$,
clearly $\rho ^{\prime \prime }\left( \cdot ,y\right) =2$, so $a_{\rho
}=b_{\rho }=1$.

Let $f_{_{0}}$ solve $\int_{\mathcal{Y}}\rho ^{\prime }(f_{0}(\boldsymbol{x}%
),y)d\mu (y|\boldsymbol{x})=0$ and $f_{0}\in W^{2,p}(\mathcal{X})$. Then $%
f_{_{0}}$ is the target function that minimizes the expected risk given in (%
\ref{EQ:f0}). Lemma \ref{LEM:f-f0} given below will show that Assumptions %
\ref{ass4} and \ref{ass5} are implied from Assumption \ref{ass6}.

\begin{lemma}
\label{LEM:f-f0}Under Assumption \ref{ass6}, for any $f\in \textcolor{black}{ \mathcal{F(%
}R ,m,B_{0},B_{1})}$, one has $a_{\rho }||f-f^{\ast
}||_{2}^{2}\leq \mathcal{E}(f)-\mathcal{E}(f^{\ast })$ and $\mathcal{E}(f)-%
\mathcal{E}(f_{0})\leq b_{\rho }||f-f_{0}||_{2}^{2}$.
\end{lemma}

It is easy to see that the quantile and Huber loss functions do not satisfy
Assumption \ref{ass6}. In the lemmas below we will show that under mild
conditions, Assumptions \ref{ass4} and \ref{ass5} are satisfied by the
quantile and Huber loss functions.

\begin{lemma}
\label{LEM:quantile}Assume that for all $\boldsymbol{x}\in \mathcal{X}$, it
is possible to define a conditional density function \textcolor{black}{$\mu ^{\prime }(u|\boldsymbol{x})$} of $Y|\boldsymbol{X}=%
\boldsymbol{x}$ such that $1/C_{1}\leq \mu ^{\prime }(u|\boldsymbol{x})\leq
1/C_{2}$ for some $C_{1}\geq C_{2}>0$ for all $u\in \{u\in \mathbb{R}$: $%
|u-f^{\ast }(\boldsymbol{x})|\leq 2B_{0}$ or $|u-f_{0}(\boldsymbol{x})|\leq
2B_{0}\}$. Then for any $f\in \textcolor{black}{ \mathcal{F(%
}R ,m,B_{0},B_{1})}$,
the quantile loss given in (\ref{def:quantile}) satisfies $a_{\rho
}||f-f^{\ast }||_{2}^{2}\leq \mathcal{E}(f)-\mathcal{E}(f^{\ast })$ and $%
\mathcal{E}(f)-\mathcal{E}(f_{0})\leq b_{\rho }||f-f_{0}||_{2}^{2}$ with $%
a_{\rho }=(2C_{1})^{-1}$ and $b_{\rho }=(2C_{2})^{-1}$.
\end{lemma}

\begin{lemma}
\label{LEM:Huber}Assume that for all $\boldsymbol{x}\in \mathcal{X}$, $%
1/c_{1}\leq \mu (u+\delta |\boldsymbol{x})-\mu (u-\delta |\boldsymbol{x}%
)\leq 1/c_{2}$ for some $c_{1}\geq c_{2}>0$ for all $u\in \{u\in \mathbb{R}$%
: $|u-f^{\ast }(\boldsymbol{x})|\leq 2B_{0}$ or $|u-f_{0}(\boldsymbol{x}%
)|\leq 2B_{0}\}$, where $\mu (u|\boldsymbol{x})$ is the conditional
cumulative function of $Y$ given $Y|\boldsymbol{X}=\boldsymbol{x}$. Then for
any \textcolor{black}{$f\in \mathcal{F(}R ,m,B_{0},B_{1})$}, the Huber loss
given in (\ref{def:huber}) satisfies $a_{\rho }||f-f^{\ast }||_{2}^{2}\leq
\mathcal{E}(f)-\mathcal{E}(f^{\ast })$ and $\mathcal{E}(f)-\mathcal{E}%
(f_{0})\leq b_{\rho }||f-f_{0}||_{2}^{2}$ with $a_{\rho }=(2c_{1})^{-1}$ and
$b_{\rho }=(2c_{2})^{-1}$.
\end{lemma}

The proofs of Lemmas \ref{LEM:f-f0}-\ref{LEM:Huber} are provided in Section %
\ref{proof_LEM:f-f0} of the Appendix.

\begin{remark}
    \textcolor{black}{We impose the upper and lower bounds for the conditional density function for $u$ near $f^{\ast}(\boldsymbol{x})$ and $f_0(\boldsymbol{x})$. The upper bounds hold because $\mathcal{X} = [0,1]^d$ is compact. The lower bounds ensure the density function is non-zero when $u$ is near $f^{\ast}$ and $f_0$. These mild conditions are easily satisfied.}
\end{remark}

\section{Simulation studies}

\label{sec:simulation} In this section, we conduct simulation studies to
assess the finite-sample performance of the proposed methods.

\subsection{Date generating process}

\label{subsec:simu_models}  To illustrate the methods, we generate data from the following nonlinear models: 
\begin{align*}
\text{Model 1}& :\mathbb{E}(Y_{i}|X_{i})=X_{i1}^{\ast 2}X^{\ast}_{i2} + \cos (2.5(X^{\ast}_{i3}+X^{\ast}_{i4}) - 2) + 2X_{i5}^{\ast 2} + \frac{2X^{\ast}_{i5}}{X_{i3}^{\ast 2}+X_{i4}^{\ast 4}+2}; \\
\text{Model 2}& :\mathbb{E}(Y_{i}|X_{i})=\mathbb{P}(Y_{i}=1|X_{i})=\frac{%
e^{\eta _{i}}}{1+e^{\eta _{i}}}; \\
& \eta _{i}= 4\cos(2X^{\ast}_{i4}X^{\ast}_{i5} - X^{\ast}_{i2}) - 9 X^{\ast 2}_{i3} \sqrt{X^{\ast}_{i4}X^{\ast}_{i5} + X^{\ast}_{i1}} - 11X^{\ast}_{i3} + 12X^{\ast 2}_{i1};
\end{align*}%
%}
where $X_i=(X_{i1},...,X_{id})^{\top }$ is the $d$-dimensional vector of the covariates, for $1 \le i \le n$. In Models 1 and 2, we let $X^{\ast}_{ij}=\frac{5}{d} \sum_{j^\prime = d(j -1)/5+1}^{dj /5} X_{ij^\prime}$, so each $X^{\ast}_{ij}$ is the average of $d/5$ covariates   for $1 \le j \le 5$. 
\iffalse
We partition the $d$ predictors into 5 subgroups, each denoted as $X^{\ast}_{ij^\prime}$. Each $X^{\ast}_{ij^\prime}$ represents the average of $p/5$ predictors and is calculated as $\frac{5}{p} \sum_{j = p(j^\prime -1)/5+1}^{pj^\prime /5} X_{ij}$ for $1 \le j^\prime \le 5$. 
\fi

For Model 1, we generate the responses from $Y_{i}=E(Y_{i}\mid X_{i})+\epsilon _{i}$, where $\epsilon _{i}$ are independently generated from
the standard normal distribution and Laplace distribution, respectively, for
$1\leq i\leq n$.  For each setting, we run $n_{rep}=100$ replications. Let $n=2000,4000$ and $d=5,10,100.$ When $d=5,10$, we generate the covariates from $X_{i}\sim \mathcal{U}\left( [0,1]^{d}\right)$. When $d=100$, we generate the covariates from $X_{ij}=\Phi(Z_{ij})$, where $\Phi$ is the cdf of the standard normal and $Z_{ij} = F_i^\top L_j + \zeta_{ij}$, in which $\zeta_{ij} \sim \mathcal{N} \left(0, 0.5^2 \right)$, $F_i\sim \mathcal{N} \left( 0, \Sigma \right)$ and $L_j\sim \mathcal{N} \left( 0, \Sigma \right)$ with $\Sigma = \{0.5^{|k-k^\prime|}\}_{1 \le k, k^\prime \le 15}$. We keep $L_j$ fixed for each replication.

\iffalse
We consider two settings for both Model 1 and Model 2. The first setting is for $p = 5, 10$ with $X_{i}\sim \mathcal{U}\left( [0,1]^{p}\right)$.  The second setting is for $p=100$ with $X_{ij}=F(Z_{ij})$, where $F$ is the cdf of standard normal distribution and $Z_{ij} = F_i^\top L_j + \zeta_{ij}$, designed using a factor model with $\zeta_{ij} \sim \mathcal{N} \left(0, 0.5^2 \right)$. Both $F_i$ and $L_j$ are generated from multivariate normal $ \mathcal{N} \left( 0, \Sigma \right)$, where $\Sigma = 0.5^{|k-k^\prime|}$ for $1 \le k, k^\prime \le 15$. We keep $L_j$ fixed for each replication. When dealing with $p=100$, for the purpose of alleviating the computing load and ensuring robust performance, we use Principal Component Analysis (PCA) to extract the leading principal components. We select the principle components with the explained variance greater than $75\%$.

\fi

The nonlinear functions in Model 1 are designed based on the nonlinear patterns of the response changing with the covariates in the Boston housing data application analyzed in Section \ref{sec:real}. Figure \ref{simulation_plot} shows the scatterplot of three simulated covariates from Model 1; we can observe similar nonlinear patterns of age, NOX, and RM in Figure \ref{boston_scater} in Section \ref{sec:real}.

\begin{figure}[th]
\begin{center}
$%
\begin{array}{cc}
\includegraphics[width=14cm]{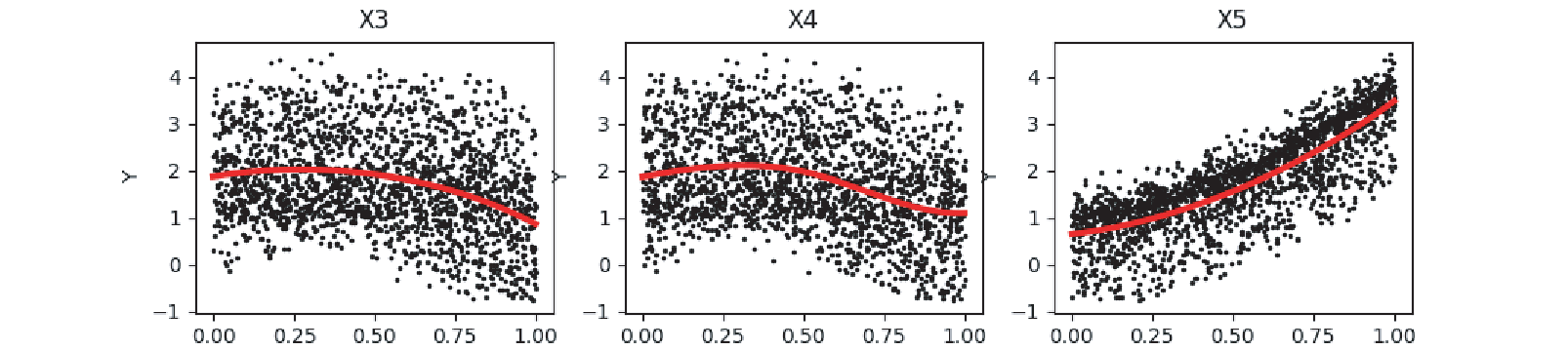} 
\end{array}%
$%
\end{center}
\caption{Scatterplot of three simulated covariates, where the red line represents the fitted mean curve by using cubic B-splines.}
\label{simulation_plot}
\end{figure}

We evaluate the numerical results based on the same test dataset $(y_{i}^{\ast },x_{i}^{\ast })$ for $1\leq i\leq n$ generated from our considered models.  For Model 1, let $\widehat{f}(x_{i}^{\ast })$ be the estimate of the regression function $f(x_{i}^{\ast })$.
We report
\begin{eqnarray*}
\text{average bias}^{2} &=&\frac{1}{n}\sum\nolimits_{i=1}^{n}\{\frac{1}{%
n_{rep}}\sum\nolimits_{j=1}^{n_{rep}}\widehat{f}(x_{i}^{\ast
})-f(x_{i}^{\ast })\}^{2}, \\
\text{average variance} &=&\frac{1}{n}\sum\nolimits_{i=1}^{n}\{\frac{1}{%
n_{rep}}\sum\nolimits_{j=1}^{n_{rep}}\widehat{f}(x_{i}^{\ast })^{2}-(\frac{1%
}{n_{rep}}\sum\nolimits_{j=1}^{n_{rep}}\widehat{f}(x_{i}^{\ast }))^{2}\}, \\
\text{average mse} &\text{=}&\frac{1}{n}\sum\nolimits_{i=1}^{n}\frac{1}{%
n_{rep}}\sum\nolimits_{j=1}^{n_{rep}}\{\widehat{f}(x_{i}^{\ast
})-f(x_{i}^{\ast })\}^{2}.
\end{eqnarray*}

For Model 2, we report the average values of square root mse (RMSE) for the estimated conditional probability, accuracy, sensitivity, specificity, precision, recall, and F1 score calculated based on the $n_{rep}=100$ simulation replicates. Let $\widehat{p}(x_i^{\ast}) = \widehat{p}(y_i^{\ast} = 1 | x_i^{\ast})$ be the estimated conditional probability of $p(x_i^{\ast}) = p(y_i^{\ast} = 1 | x_i^{\ast})$. The average RMSE is calculated by
\begin{equation*}
\text{average RMSE = } \sqrt{\frac{1}{n}\sum\nolimits_{i=1}^{n}\frac{1}{%
n_{rep}}\sum\nolimits_{j=1}^{n_{rep}}\{\widehat{p}(x_i^{\ast}) - p(x_i^{\ast})\}^{2}}.
\end{equation*} 

\subsection{Simulation results}

%%% for model 1
We use $m=\max (\left\lfloor 0.2\log _{2}n\right\rfloor
+c,0)$ and $R=3\max (\left\lfloor 0.2\log _{2}n\right\rfloor ,m)$ for
the SDRN estimator, where the constant $c$ is a tuning parameter.  
The choices of $m$ and $R$
satisfy the conditions in Theorem \ref{THM:rate}. In our procedure, we have two tuning parameters $c$ and $\lambda$. The constant $c$ determines the complexity of the neural networks, and $\lambda$ is used in the ridge penalty term.  Both of them control the fitting result. A larger value of $c$ and a smaller value of $\lambda$ can yield larger variance and smaller bias, while a smaller $c$ and a larger  $\lambda$ lead to smaller variance and larger bias. We choose the optimal values for  $c$ and  $\lambda$ by minimizing the prediction errors in the test dataset.

 For the large dimension setting with  $p=100$, to ensure robust performance, we first perform dimension reduction using PCA on the covariate matrix to reduce the high-dimensional features into lower-dimensional latent vectors, which are the leading principal components, and then use the principal components that explain the variance greater than %$75\%$
 \textcolor{black}{80\% for Model 1 and 75\% for Model 2} as the predictors in our regression analysis. Then the vector of the principal components would be the inputs fitting into our SDRN model, so its dimension becomes the actual dimension of the inputs for the model.  
 Dimensionality reduction is still required for machine learning algorithms when the dimension of features is very large compared to the sample size.   

Table \ref{tab:model1} reports numerical results for the SDRN estimates from the optimal fitting with the optimal tuning parameter values for Model 1. Specifically, it shows the average mean squared
error (MSE), average $\text{bias}^{2}$ and average variance of the SDRN
estimates obtained from the quadratic and quantile $(\tau =0.5,\ 0.25)$ loss functions, respectively, based on the 100 simulation replications. Table \ref{tab:model1} shows that in general, the SDRN  estimates achieve a good balance between the bias$^{2}$ and variance at the optimal values of the tuning parameters. Moreover, we see that the MSE values decrease as the sample size increases. This corroborates with our theoretical results given in Theorem \ref{THM:rate}. The MSE values also increase gradually as the dimension $d$ becomes larger.  Our method can be applied to mean regression as well as quantile regression at any given quantile level. We observe that when the error terms are generated from a normal distribution, the SDRN estimate obtained from the quadratic loss (mean regression) has the smallest MSE, but the estimate obtained from the quantile loss at $\tau=0.5$ (median regression) has a comparable MSE value comparing to the estimate from the quadratic loss. However, when the errors are generated from the Laplace distribution which has heavy tails, we see that the SDRN estimate from the median regression using the quantile loss with  $\tau=0.5$ yields obviously smaller MSE values than the estimate from the mean regression using the quadratic loss.

\begin{table}[tbp]
\caption{The average MSE, $\text{bias}^{2}$ and variance of the SDRN
estimators obtained from the quadratic and quantile $\left( \protect\tau %
=0.5,\ 0.25\right) $ loss functions based on the 100 simulation replications for Model 1.}
\label{tab:model1} \centering
\vskip 0.4cm
\scalebox{0.6}{
\begin{tabular}{lcccccc|cccccc}
\hline
 & \multicolumn{6}{c|}{n=2000} & \multicolumn{6}{c}{n = 4000} \\ \hline
 & \multicolumn{3}{c|}{Normal error} & \multicolumn{3}{c|}{Laplace error} & \multicolumn{3}{c|}{Normal error} & \multicolumn{3}{c}{Laplace error} \\ \hline
 & quadratic & $\tau =$ 0.5 & \multicolumn{1}{c|}{ $\tau =$ 0.25} & quadratic & $\tau =$ 0.5 & $\tau =$ 0.25 & quadratic & $\tau =$ 0.5 & \multicolumn{1}{c|}{ $\tau =$ 0.25} & quadratic & $\tau =$ 0.5 & $\tau =$ 0.25 \\ \hline
d = 5 &  &  & \multicolumn{1}{c|}{} &  &  &  &  &  & \multicolumn{1}{c|}{} &  &  &  \\ \hline
bias2 & 0.0125 & 0.0142 & \multicolumn{1}{c|}{0.0179} & 0.0131 & 0.0144 & 0.0169 & 0.0128 & 0.0113 & \multicolumn{1}{c|}{0.0119} & 0.0130 & 0.0118 & 0.0130 \\
var & 0.0098 & 0.0116 & \multicolumn{1}{c|}{0.0124} & 0.0209 & 0.0102 & 0.0220 & 0.0052 & 0.0080 & \multicolumn{1}{c|}{0.0091} & 0.0097 & 0.0072 & 0.0137 \\
mse & {\textcolor{red}{0.0223}} & 0.0257 & \multicolumn{1}{c|}{0.0304} & 0.0341 & {\textcolor{red}{0.0246}} & 0.0389 & {\textcolor{red}{0.0180}} & 0.0193 & \multicolumn{1}{c|}{0.0211} & 0.0226 & {\textcolor{red}{0.0190}} & 0.0267 \\ \hline
d=10 & \multicolumn{3}{c|}{} & \multicolumn{3}{c|}{} & \multicolumn{3}{c|}{} & \multicolumn{3}{c}{} \\ \hline
bias2 & 0.0293 & 0.0322 & \multicolumn{1}{c|}{0.0352} & 0.0356 & 0.0312 & 0.0376 & 0.0262 & 0.0285 & \multicolumn{1}{c|}{0.0250} & 0.0293 & 0.0252 & 0.0264 \\
var & 0.0167 & 0.0182 & \multicolumn{1}{c|}{0.0168} & 0.0237 & 0.0161 & 0.0227 & 0.0110 & 0.0095 & \multicolumn{1}{c|}{0.0126} & 0.0192 & 0.0119 & 0.0175 \\
mse & {\textcolor{red}{0.0460}} & 0.0504 & \multicolumn{1}{c|}{0.0520} & 0.0592 & {\textcolor{red}{0.0472}} & 0.0603 & {\textcolor{red}{0.0372}} & 0.0379 & \multicolumn{1}{c|}{0.0376} & 0.0486 & {\textcolor{red}{0.0372}} & 0.0440 \\ \hline
d=100 & \multicolumn{3}{c|}{} & \multicolumn{3}{c|}{} & \multicolumn{3}{c|}{} & \multicolumn{3}{c}{} \\ \hline
bias2 & 0.0352 & 0.0379 & \multicolumn{1}{c|}{0.0386} & 0.0388 & 0.0372 & 0.0406 & 0.0351 & 0.0354 & \multicolumn{1}{c|}{0.0360} & 0.0351 & 0.0350 & 0.0377 \\
var & 0.0200 &  0.0199 & \multicolumn{1}{c|}{0.0206} & 0.0222 & 0.0199 & 0.0250 & 0.0155 & 0.0168 & \multicolumn{1}{c|}{0.0168} & 0.0191 & 0.0164 & 0.0197\\
mse & \textcolor{red}{0.0552} & 0.0578 & \multicolumn{1}{c|}{0.0592} & 0.0610 & \textcolor{red}{0.0571} & 0.0655 & \textcolor{red}{0.0506} & 0.0521 & \multicolumn{1}{c|}{0.0528} & 0.0542 & \textcolor{red}{0.0514} & 0.0575 \\ \hline
\end{tabular}}
\end{table}

%%% for model 1 (d=5 n=2000)
%%% for model 1 (d=5 n=2000)

%%%%%%%% for model 1 n=5000

Next, we compare the performance of our proposed SDRN
estimator with that of four other popular machine learning methods, including
the fully-connected feedforward neural networks (FNN), the
gradient boosted machines (GBM), the random forests (RF), and the generalized additive models (GAM). For FNN, we employ ReLU as the activation function. For GAM, a cubic regression
spline basis is used. We report the results from the optimal fitting with the optimal tuning parameters minimizing the MSE value in the test dataset. Table \ref{tab:model1 compare} reports the average MSE, $\text{bias%
}^{2}$ and variance for the five methods based on the 100 replicates when $%
n=2000$, $d=5, 10$. Across all five methods, the quadratic loss and the quantile $(\tau =0.5)$ loss are used for the normal and Laplace errors, respectively. We observe
that our SDRN\ has the smallest MSE values in each case. Among all
methods, the GAM method has the largest bias due to model misspecification.

We use Model 2 to illustrate the performance of the SDRN estimate for classification. Table \ref{tab:model2} shows the average RMSE of the estimated conditional probabilities and the metrics to evaluate the classification performance, including the average values of accuracy, precision, recall, F1 score, and specificity based on 100 simulation replications for Model 2. The evaluation is performed on the test dataset, while the parameter estimates are obtained using the training dataset. The RMSE value becomes smaller when the sample size is larger, indicating that the estimated conditional probabilities are approaching the true values when the sample size is larger. This result corroborates with our theoretical property given in Proposition  \ref{THM:rate_fixed}. However, the RMSE value increases slightly from $d=5$ to $d=10$, because the effect of the dimension $d$ is only reflected on the logarithm order given in Proposition \ref{THM:rate_fixed}. We also observe that the classification accuracy is improved as the sample size becomes larger.  

We then compare the performance of the proposed SDRN with that of the four popular machine learning methods including FNN, GBM, RF and GAM for the classification task in Model 2.  Table \ref{tab:model2 compare} reports the numerical results of the five methods obtained from the optimal fitting. The optimal values of all tuning parameters are chosen by minimizing the prediction errors in the test data. We observe that SDRN outperforms the other four methods in terms of RMSE, accuracy, recall, 
and F1 score. The F1 score conveys the balance between precision and recall. Overall, FNN has the second-best performance following SDRN. Its accuracy is slightly smaller than that of SDRN, but its RMSE value is significantly larger than that of SDRN when $d=10$, indicating that FNN needs to have a greater network complexity than SDRN to achieve similar classification accuracy.

\iffalse
For both Model 1 and Model 2, we set the parameters in the Adam optimization algorithm with exponential decay rates $\beta_1 = 0.9$, $\beta_2 = 0.999$, and $\epsilon = 10^{-8}$, as suggested in the literature. We select the learning rates from $1 \times 10^{-4}$, $5 \times 10^{-4}$, $1 \times 10^{-3}$, and $5 \times 10^{-3}$, choosing the one that provides the best performance. 
\fi

\begin{table}[tbp]
\caption{The average MSE, $\text{bias}^{2}$ and variance of the five methods
obtained from the quadratic loss for normal error and quantile $\left(
\protect\tau =0.5\right) $ loss for Laplace error based on the 100
simulation replications for Model 1 when $n=2000$, $d=5, 10$.}
\label{tab:model1 compare}\centering
\vskip 0.4cm
\scalebox{0.8}{
\begin{tabular}{cccccc|ccccc}
\hline
\multicolumn{6}{c|}{Quadratic (Normal)} & \multicolumn{5}{c}{Quantile 0.5 (Laplace)} \\ \hline
d=5 & SDRN & FNN & GBM & RF & GAM & SDRN & FNN & GBM & RF & GAM \\ \hline
bias2 & 0.0125 & 0.0062 & 0.0157 & 0.0386 & 0.1296 & 0.0144 & 0.0050 & 0.0238 & 0.0500 & 0.1251 \\
var & 0.0098 & 0.0198 & 0.0358 & 0.0523 & 0.0066 & 0.0102 & 0.0247 & 0.0327 & 0.0656 & 0.0098 \\
mse & 0.0223 & 0.0260 & 0.0514 & 0.0909 & 0.1362 & 0.0246 & 0.0298 & 0.0565 & 0.1157 & 0.1350 \\ \hline
\multicolumn{1}{l}{d=10} & SDRN & FNN & GBM & RF & GAM & SDRN & FNN & GBM & RF & GAM \\ \hline
bias2 & 0.0293 & 0.0167 & 0.0355 & 0.0518 & 0.0712 & 0.0312 & 0.0149 & 0.0428 & 0.0544 & 0.0603 \\
var & 0.0167 & 0.0306 & 0.0318 & 0.0515 & 0.0136 & 0.0161 & 0.0331 & 0.0286 & 0.0726 & 0.0142 \\
mse & 0.0460 & 0.0473 & 0.0673 & 0.1033 & 0.0848 & 0.0472 & 0.0480 & 0.0714 & 0.1270 & 0.0745 \\ \hline
\end{tabular}
}
\end{table}

\begin{table}[tbp]
\caption{The average of RMSE, accuracy, sensitivity, precision, recall, and F1 score of the SDRN estimate based on the 100 simulation replications when $n=2000$ and $n=4000$ for Model 2.}
\label{tab:model2}\centering
\vskip 0.4cm
\scalebox{0.6}{
\begin{tabular}{lcccccc|cccccc}
\hline
 & \multicolumn{6}{c|}{n=2000} & \multicolumn{6}{c}{n=4000} \\ \hline
 & RMSE & Accuracy & Precision & Recall & F1 & Specificity & RMSE & Accuracy & Precision & Recall & F1 & Specificity \\ \hline
d=5 & 0.0617 & 0.9212 & 0.9112 & 0.9216 & 0.9163 & 0.9208 & 0.0553 & 0.9227 & 0.9156 & 0.9220 & 0.9188 & 0.9234 \\
d=10 & 0.0667 & 0.8673 & 0.8594 & 0.8533 & 0.8562 & 0.8794 & 0.0589 & 0.8731 & 0.8613 & 0.8643 & 0.8627 & 0.8806 \\
d=100 & 0.1838 & 0.7248 & 0.7091 & 0.6600 & 0.6832 & 0.7778 & 0.1812 & 0.7435 & 0.7347 & 0.6757 & 0.7034 & 0.7991 \\ \hline
\end{tabular}
}
\end{table}

\begin{table}[tbp]
\caption{The average of RMSE, accuracy, sensitivity, precision, recall, and F1 score of the five methods based on the 100 simulation replications for Model 2 when $n=2000$, $d=5, 10$.}
\label{tab:model2 compare}\centering
\vskip 0.4cm
\scalebox{0.8}{
\begin{tabular}{c|ccccc|ccccc}
\hline
\multicolumn{1}{l|}{} & \multicolumn{5}{c|}{d=5} & \multicolumn{5}{c}{d=10} \\ \hline
 & SDRN & FNN & GAM & GBM & RF & SDRN & FNN & GAM & GBM & RF \\ \hline
RMSE & 0.0617 & 0.0688 & 0.0702 & 0.0886 & 0.0896 & 0.0667 & 0.0824 & 0.0850 & 0.1020 & 0.1382 \\
Accuracy & 0.9212 & 0.9182 & 0.9174 & 0.9115 & 0.9135 & 0.8673 & 0.8660 & 0.8648 & 0.8610 & 0.8608 \\
Precision & 0.9112 & 0.9086 & 0.9072 & 0.8990 & 0.9007 & 0.8594 & 0.8603 & 0.8591 & 0.8558 & 0.8569 \\
Recall & 0.9216 & 0.9179 & 0.9175 & 0.9138 & 0.9165 & 0.8533 & 0.8486 & 0.8471 & 0.8416 & 0.8397 \\
F1 & 0.9163 & 0.9131 & 0.9123 & 0.9063 & 0.9085 & 0.8562 & 0.8543 & 0.8530 & 0.8486 & 0.8481 \\
Specificity & 0.9208 & 0.9185 & 0.9172 & 0.9094 & 0.9108 & 0.8794 & 0.8811 & 0.8801 & 0.8777 & 0.8789 \\ \hline
\end{tabular}
}
\end{table}

\section{Real data application\label{sec:real}}

In this section, we illustrate our proposed method by using two datasets, the Boston housing data with continuous responses and the BUPA liver disorders data with binary responses. We use the two datasets to illustrate the prediction and classification performance of our SDRN method on continuous and binary responses, respectively. Each dataset is randomly split into 75\% training data and 25\% test
data. The training data is used to fit the model, while the test data is
used to evaluate the performance. Then, we compare our SDRN with five
methods, including LM/GLM (linear model/generalized linear model), FNN, GBM,
RF and GAM. For all methods, the tuning parameters are selected by 5-fold
cross validations based on a grid search.

\subsection{Boston housing data}

The Boston housing data set \cite{harrison1978hedonic} contains 506 census tracts of Boston from the 1970
census. Each census tract represents one observation. Thus, there are 506
observations and 14 attributes in the dataset, where MEDV (the median value
of owner-occupied homes) is the response variable. Following \cite%
{fan2005profile}, seven explanatory variables are considered: CRIM (per
capita crime rate by town), RM (average number of rooms per dwelling), TAX
(full-value property-tax rate per USD 10,000), NOX (nitric oxides
concentration in parts per 10 million), PTRATIO (pupil-teacher ratio by
town), AGE (proportion of owner-occupied units built prior to 1940) and
LSTAT (percentage of the lower status of the population). Since the value of the
MEDV variable is censored at 50.0 (corresponding to a median price of
\$50,000), we remove the 16 censored observations and use the remaining 490
observations for analysis.

\begin{figure}[th]
\begin{center}
$%
\begin{array}{cc}
\includegraphics[width=6.5cm]{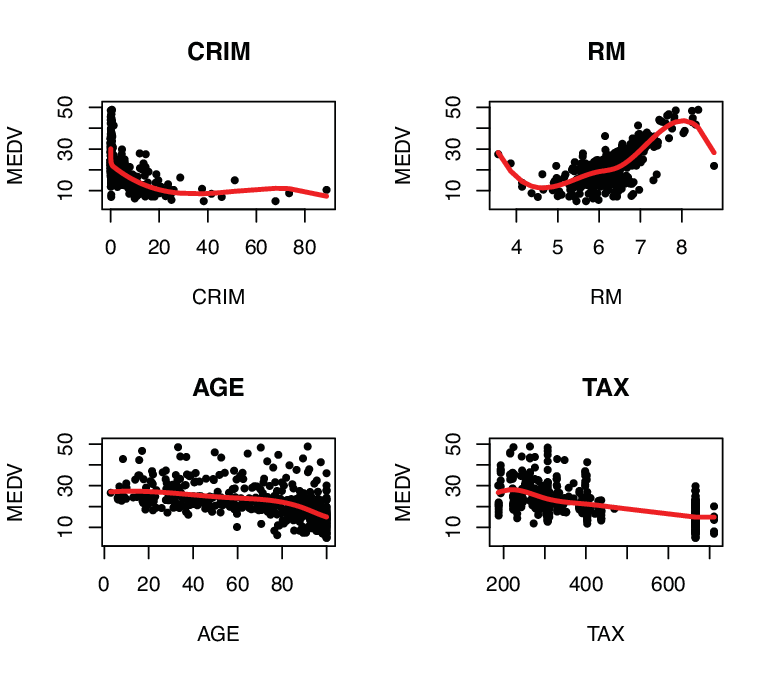} & %
\includegraphics[width=6.5cm]{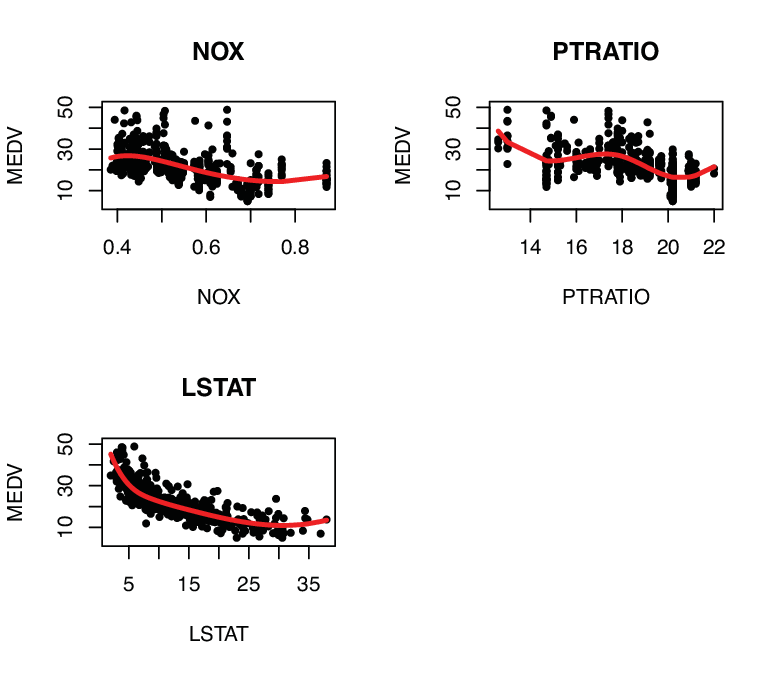}%
\end{array}%
$%
\end{center}
\caption{Scatter plot of MEDV versus each covariate, where the red line
represents the fitted mean curve by using cubic B-splines.}
\label{boston_scater}
\end{figure}

For preliminary analysis of nonlinear patterns, Figure \ref{boston_scater}
shows the scatter plots of the response MEDV against each covariate with the
red lines representing the fitted mean curves by using cubic B-splines. We
observe that the MEDV value has a clear nonlinear changing pattern with
these covariates. The MEDV value has an overall increasing pattern with RM,
whereas it decreases as CRIM, NOX, PTRATIO, TAX, and LSTAT increase. The MEDV
value starts decreasing slowly as AGE increases. However, when the AGE
passes 60, it starts dropping dramatically.

Next, we use our SDRN method with quadratic loss to fit a mean regression of
this data and compare it with LM, FNN, GBM, RF and GAM methods. Table \ref%
{compare_boston} shows the mean squared prediction error (MSPE) from the six
methods. We observe that SDRN outperforms other methods with the smallest
MSPE. The LM method has the largest MSPE, as it cannot capture the nonlinear
relationships between MEDV and the covariates. GAM has the second-largest
MSPE due to its restrictive additive structure without allowing interaction
effects. %The coefficient of determination $R^{2}$ for SDRN is 0.881, while it is 0.735 for LM.

To explore the nonlinear patterns between MEDV and each covariate, in Figure \ref{fig:fitted MEDV} we plot the estimated mean function of MEDV versus each covariate (solid lines), and the estimated median
function of MEDV versus each covariate (dashed lines), obtained from our
 SDRN method with the quadratic loss and the quantile $(\tau =0.5)$ loss,
respectively. The plots are constructed by regressing the estimated mean and median functions using cubic B-splines on each covariate.

We see that overall the fitted MEDV has a negative relationship with all covariates except RM. The estimated MEDV decreases slowly when AGE increases from 0 to 60, and then it begins to drop progressively after AGE passes
60. It suggests that the percentage of old houses in a neighborhood overall has an adverse effect on the house price, and its effect can be prominent when the percentage exceeds a certain level such as $60\%$.   We see that the MEDV value decreases slightly when TAX increases from 200 to 300, and then its value has a steady drop when TAX increases from 300 to over 600, indicating that TAX overall has a negative impact on the house price as well. The MEDV value decreases quickly when  LSTAT increases from 0 to 10, and then it levels off. This result is expected as it is difficult for the lower-status population to afford expensive houses. Moreover, the MEDV value drops gradually as the NOX level rises. This can be explained by that when the houses in a neighborhood are more expensive, this area is often less dense and has more green spaces. MEDV drops sharply as CRIM increases from 0 to 20, and it remains stable from 20 to over 80. MEDV has a high value when the PTRATIO value is small (less than 10), indicating that children in rich neighborhoods may go to private schools. The MEDV value remains stable when PTRATIO increases from 15 to 18, and then it drops steadily when PTRATIO increases from 18 to 20. Lastly, we see that overall MEDV has a steadily increasing pattern with RM. In Figure  \ref{fig:fitted MEDV}, we plot the estimated mean and median functions of MEDV using SDRN versus each covariate. In general, we see that the estimated mean and median functions have similar nonlinear patterns. There is a slight difference between the two curves in the tail parts for CRIM and RM. The difference may be caused by the outliers that can affect the fitting result of mean regression. 

\begin{figure}[tbp]
    \centering
    \includegraphics[width = 11cm]{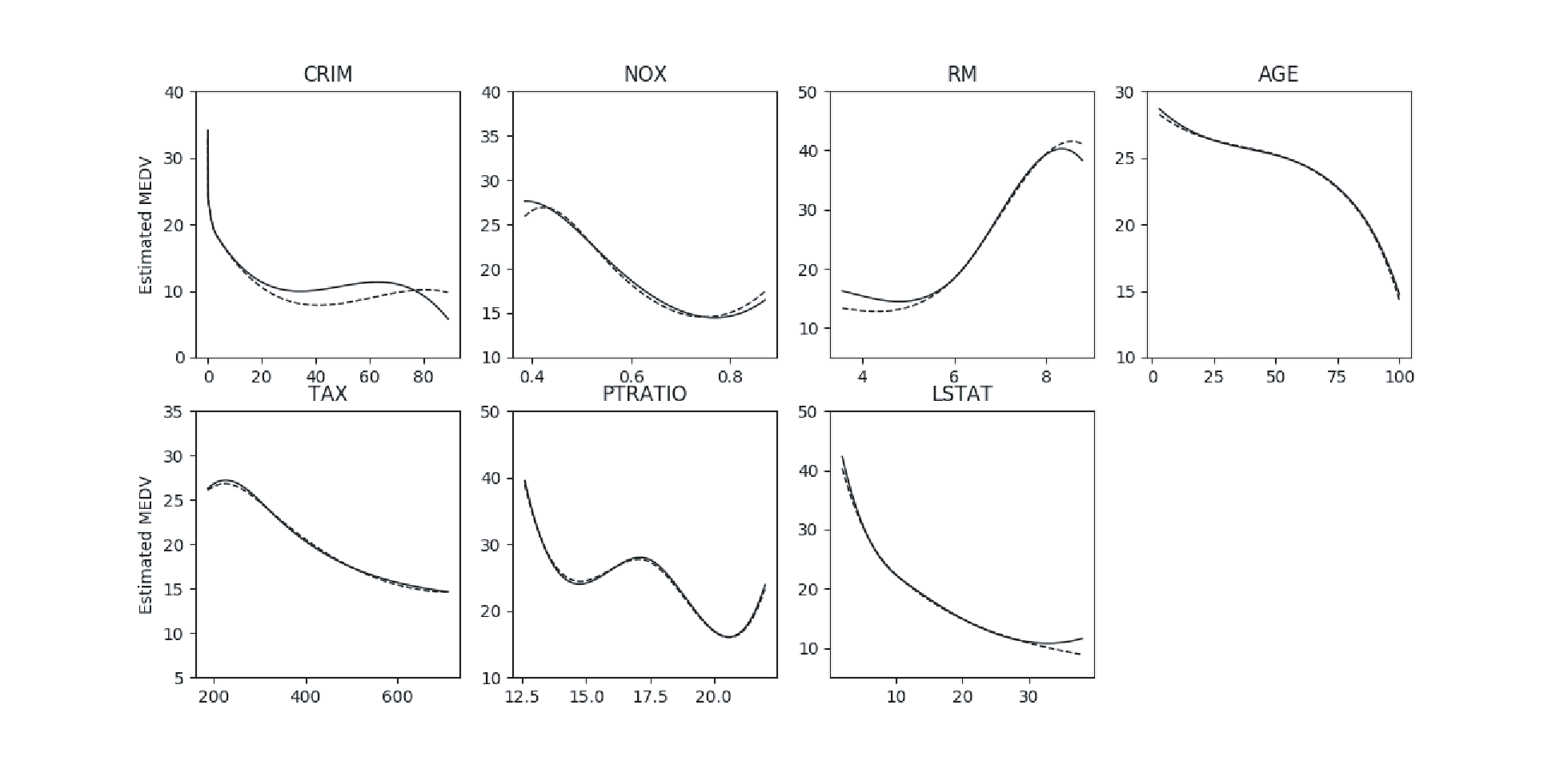}
    \caption{The estimated mean (solid lines) and median (dashed lines) curves of MEDV from SDRN against each covariate.}
    \label{fig:fitted MEDV}
\end{figure}

%\begin{figure}[tbp]
%    \centering
%    \includegraphics[width = 10cm]{Boston_boxplot.png}
%    \caption{Boxplots of MEDV when $30<\text{CRIM}<60$, $\text{CRIM}>60$, and $\text{RM}<5$, where yellow line represents median, green traingle represents mean.}
%    \label{fig:enter-label}
%\end{figure}

%% boston fitted curve
%\begin{figure}[tbp]
%\centering
%\scalebox{0.9}{
%\includegraphics[width=4.8cm]{b1.pdf}\quad
%\includegraphics[width=4.8cm]{b2.pdf}\quad
%\includegraphics[width=4.8cm]{b3.pdf}
%}
%\par
%\medskip
%\par
%\scalebox{0.9}{
%\includegraphics[width=4.8cm]{b5.pdf}\quad
%\includegraphics[width=4.8cm]{b6.pdf}\quad
%\includegraphics[width=4.8cm]{b7.pdf}
%}
%\par
%\medskip
%\par
%\scalebox{0.9}{
%\includegraphics[width=4.8cm]{b4.pdf}

%}
%\caption{The estimated mean (solid lines) and median (dashed lines) curves
%of MEDV against each covariate, while other covariates are fixed at their
%mean values for Boston housing data.}
%\label{boston_fit}
%\end{figure}

%%% boston mse table
\begin{table}[tbp]
\caption{The mean squared prediction error (MSPE) from six different methods
using quadratic loss for the Boston housing data.}
\label{compare_boston}\centering
\vskip 0.4cm
\begin{tabular}{c|cccccc}
\hline
& \textbf{SDRN} & LM & FNN & GBM & RF & GAM \\ \hline
MPSE & \textbf{7.626} & 15.554 & 7.801 & 7.831 & 9.955 & 12.908 \\ \hline
\end{tabular}%
\end{table}

\subsection{BUPA data}

The BUPA Liver Disorders dataset is available at the UCI Machine Learning
Repository \cite{uci}. It has 345 rows and 7 columns, with each row
constituting the record of a single male individual. The first 5 variables
are blood tests that are considered to be sensitive to liver disorders due
to excessive alcohol consumption; they are mean corpuscular volume (mcv),
alkaline phosphotase (alkphos), alanine aminotransferase (sgpt), aspartate
aminotransferase (sgot) and gamma-glutamyl transpeptidase (gammagt). We use
them as covariates. The 6th variable is the number of half-point equivalents
of alcoholic beverages drunk per day. Following \cite%
{mcdermott2016diagnosing}, we dichotomize it to a binary response by letting
$Y_{i}=1$ if the number of drinks is greater than 3, otherwise $Y_{i}=0$.
The 7th column in the dataset was created by BUPA researchers for training and
test data selection.

%We first calculate the McFadden-pseudo $R^{2}$ for GLM and SDRN with
%logistic loss, respectively. The pseudo $R^{2}$ for GLM is 0.2355, while it
%is 0.2584 for SDRN, indicating that the SDRN\ method yields a better
%prediction. In addition, 
We use this example to compare the performance of the classification among different methods.  
Table \ref{compare_bupa} shows the accuracy,
precision, recall, F1, and AUC (area under the ROC curve) for the group with the number of drinks
greater than 3 for the six methods with logistic loss. We see that
SDRN has the largest accuracy, recall, F1, and AUC. The recall, F1, and AUC for GLM are much smaller than other methods, possibly due to model misspecification and nonlinearity of the dataset. %Moreover, we use McFadden's pseudo $R^2=1-\frac{\log \hat{L}(M_{\text{full}})}{\log \hat{L}(M_{\text{null}})}$ to further evaluate the model fitting, where $\hat{L}(M_\text{full})$ is the estimated likelihood with all predictors and $\hat{L}(M_\text{null})$ is the estimated likelihood without
%any predictors. The higher value of the pseudo $R^2$ indicates a better model fitting. The pseudo $R^2$ from SDRN is $0.1179$, and it is larger than the pseudo $R^2=0.0702$ from GLM.

To explore the nonlinear patterns, Figure \ref{fit_bupa}
shows the scatter plots and the fitted mean curve of the estimated conditional probabilities versus the mcv, alkphos, sgpt, sgot, and gammagt, respectively. The fitted mean curves (red lines) are drawn by regressing the estimated conditional probabilities on each predictor using cubic B-splines. We can observe the estimated conditional probability rises as mcv, gammagt increase, which suggests that the mcv and gammagt levels can be strong indicators of alcohol consumption. Both sgpt and sgot are enzymes related to liver health and seem to have similar patterns in correlation with alcohol consumption, as depicted in Figure \ref{fit_bupa}. We can see the estimated conditional probability increases quickly as the level of sgpt and sgot are elevated and remain to be high as these two enzymes pass a certain value. Considering that higher levels of sgpt and sgot are associated with liver damage and excessive chronic alcohol consumption can lead to liver cell damage, it is plausible to observe a positive correlation between the levels of these two enzymes and alcohol consumption. The
estimated conditional probability has a quadratic nonlinear relationship with alkphos.
Abnormal\ (either low or high)\ levels of alkphos are connected to a few
health problems. Low levels of alkphos indicate a deficiency in zinc and
magnesium, or a rare genetic disease called hypophosphatasia, which affects
bones and teeth. High levels of alkphos can be an indicator of liver disease
or bone disorder.

\begin{table}[tbp]
\caption{Accuracy, Precision, Recall, F1, and AUC for the group with the
number of drinks greater than 3 of the BUPA data for different methods with
logistic loss.}
\label{compare_bupa}\centering
\vskip 0.4cm
\begin{tabular}{c|cccccc}
\hline
& \textbf{SDRN} & GLM & FNN & GBM & RF & GAM \\ \hline
Accuracy & \textbf{0.678} & 0.655 & 0.667 & 0.632 & 0.667 & 0.644 \\
Precision & \textbf{0.636} & 0.652 & 0.618 & 0.561 & 0.643 & 0.583 \\
Recall & \textbf{0.568} & 0.405 & 0.568 & 0.622 & 0.486 & 0.568 \\
F1 & \textbf{0.600} & 0.500 & 0.592 & 0.590 & 0.554 & 0.575 \\
AUC & \textbf{0.664} & 0.623 & 0.654 & 0.631 & 0.643 & 0.634 \\ \hline
\end{tabular}%
\end{table}

\begin{figure}[tbp]
    \centering
    \includegraphics[width = 14cm]{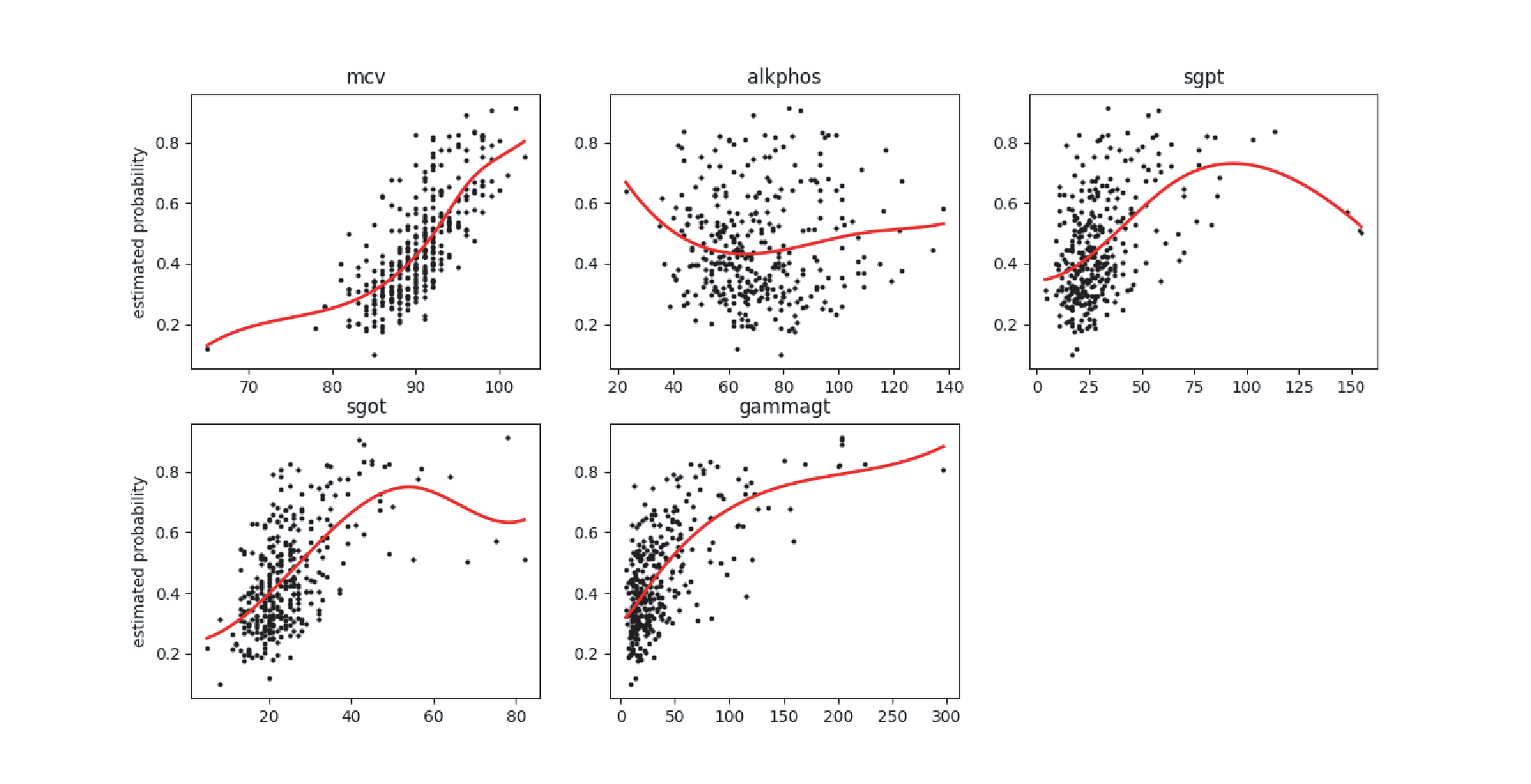}
    \caption{The estimated conditional probabilities from SDRN versus mcv, alkphos, sgot and gammagt for the BUPA data.}
    \label{fit_bupa}
\end{figure}

%\begin{figure}[tbp]
%\centering
%\scalebox{0.85}{
%\includegraphics[width=6cm]{bupa1.pdf}\quad
%\includegraphics[width=6cm]{bupa2.pdf}
%}
%\par
%\medskip
%\par
%\scalebox{0.85}{
%\includegraphics{BUPA_cond_plot.png}
%\includegraphics[width=6cm]{bupa3.pdf}\quad
%\includegraphics[width=6cm]{bupa4.pdf}
%}
%\caption{The estimated log-odds functions versus mcv, alkphos, sgot and
%gammagt , respectively, while other covariates are fixed at their mean
%values, for the BUPA data.}
%\label{fit_bupa}
%\end{figure}

\section{Discussion\label{sec:Discussion}}

In this paper, we propose a sparse deep ReLU\ network estimator (SDRN)
obtained from empirical risk minimization with a Lipschitz loss function
satisfying mild conditions. Our framework can be applied to a variety of
regression and classification problems in machine learning. In general, deep
neural networks are effective tools for lessening the curse of
dimensionality under the condition that the target functions have certain
special properties. We show that our SDRN estimator achieves a fast convergence rate when the regression function belongs to the Korobov spaces with mixed partial derivatives. 

We derive non-asymptotic excess risk bounds for the SDRN estimator. Our
framework allows the dimension of the feature space to increase with the
sample size at a rate slightly slower than $\log (n)$. We further show
that our SDRN estimator can achieve the same optimal minimax rate (up to
logarithmic factors) as one-dimensional nonparametric regression when the
dimension is fixed, and the dimensionality effect is passed on to a
logarithmic factor, so the curse of dimensionality is alleviated. The SDRN
estimator has a slightly slower rate when the dimension increases with the sample
size. Moreover, the depth and the total number of nodes and weights of the
network need to increase with the sample size with certain rates established
in the paper. These statistical properties provide an important theoretical
basis and guidance for the analytic procedures in data analysis.
Practically, we illustrate the proposed method through simulation studies
and several real data applications. The numerical studies support our
theoretical results.

Our proposed method provides a reliable solution for mitigating the curse of
dimensionality for modern data analysis. Meanwhile, it has opened up several
interesting new avenues for further work. One extension is to derive a
similar estimator for smoother regression functions with mixed derivatives
of order greater than two; Jacobi-weighted Korobov spaces \cite{SW10,LTY2020} may
be considered for this scenario. Our method can be extended to other
settings such as semiparametric models, longitudinal data and $L_{1}$
penalized regression. Moreover, it can be a promising tool for the estimation of
the propensity score function or the outcome regression function used in
treatment effect studies. These interesting topics deserve thorough
investigations for future research.

%\section*{Declaration of competing interest}
%The authors declare that they have no known competing financial interests or personal relationships that could have appeared to influence the work reported in this paper.

\section*{Acknowledgments}
We thank the editor and the anonymous reviewer for their valuable comments and suggestions for greatly improving the manuscript.
This research was supported in part by the U.S. NSF grants DMS-17-12558, DMS-20-14221 and DMS-23-10288 and the UCR Academic Senate CoR Grant.

%\section*{Data availability}
%The Boston housing data set is available at \url{http://lib.stat.cmu.edu/datasets/boston}. The BUPA Liver Disorders dataset is available at \url{https://archive.ics.uci.edu/dataset/60/liver+disorders}.

%%%%%%%%%%%%%% Appendix

\section*{Appendix}
\label{appendix}

\renewcommand{\thetheorem}{{\sc A.\arabic{theorem}}} 
\renewcommand{\theproposition}{{\sc A.\arabic{proposition}}}
\renewcommand{\thelemma}{{\sc A.\arabic{lemma}}} 
\renewcommand{\thecorollary}{{\sc A.\arabic{corollary}}}
\renewcommand{\theequation}{A.\arabic{equation}} %
\renewcommand{\thesection}{{\it A.\arabic{section}}} 
\renewcommand{\thesubsection}{{\it A.\arabic{subsection}}}
\renewcommand{\thetable}{{\it A.\arabic{table}}} 
\renewcommand{\theassumption}{{\sc A.\arabic{assumption}}}

\setcounter{equation}{0}
\setcounter{lemma}{0}
\setcounter{proposition}{0} 
\setcounter{theorem}{0}
\setcounter{subsection}{0}
\setcounter{corollary}{0}
\setcounter{table}{0}

\renewcommand{\theHtable}{thetable}
\renewcommand{\theHproposition}{ A.\arabic{proposition}}
\renewcommand{\theHlemma}{ A.\arabic{lemma}}

In the Appendix, we provide the discussions
of sparse grids approximation and the technical proofs.

\subsection{Sparse grids approximation \label{SEC:discrete}}

In this section, we introduce a hierarchical basis of piecewise linear
functions. We have discussed a connection between this hierarchical basis
and a ReLU network in Section \ref{SEC:NN}. For any functions in the Korobov
spaces satisfying Assumption \ref{ass1}, they have a unique representation
in this hierarchical basis. To approximate functions of one variable $x$ on $%
[0,1]$, a simple choice of a basis function is the standard hat function $%
\phi (x)$:%
\begin{equation*}
\phi (x)=\left\{
\begin{array}{cc}
1-|x|, & \text{if }x\in \lbrack -1,1] \\
0, & \text{otherwise.}%
\end{array}%
\right.
\end{equation*}%
To generate a one-dimensional hierarchical basis, we consider a family of
grids $\Omega _{\ell }$ of level $\ell $ characterized by a grid size $%
h_{\ell }=2^{-\ell }$ and $2^{\ell }+1$ points $x_{\ell ,s}=sh_{\ell }$ for $%
0\leq s\leq 2^{\ell }$. On each $\Omega _{\ell }$, the piecewise linear
basis functions $\phi _{\ell ,s}$ are given as
\begin{equation*}
\phi _{\ell ,s}(x)=\phi (\frac{x-x_{\ell ,s}}{h_{\ell }})\text{, }0\leq
s\leq 2^{\ell },
\end{equation*}%
on the support $\left[ x_{\ell ,s}-h_{\ell },x_{\ell ,s}+h_{\ell }\right]
\cap \left[ 0,1\right] $. Note that $||\phi _{\ell ,s}||_{\infty }\leq 1$
for all $\ell $ and $s$. The hierarchical increment spaces $W_{\ell }$ on
each $\Omega _{\ell }$ are given by
\begin{equation*}
W_{\ell }=\text{span}\{\phi _{\ell ,s}:s\in I_{\ell }\},
\end{equation*}%
where $I_{\ell }=\{s\in \mathbb{N}:0\leq s\leq 2^{\ell }$; $s$ are odd
numbers for $\ell \geq 1\}$. We can see that for each $\ell \geq 1$, the
supports of all basis functions $\phi _{\ell ,s}$ spanning $W_{\ell }$ are
mutually disjoint. Then the hierarchical space of functions up to level $L$
is
\begin{equation*}
V_{L}=\bigoplus\limits_{0\leq \ell \leq L}W_{\ell }=\text{span}\{\phi _{\ell
,s}:s\in I_{\ell },0\leq \ell \leq L\}.
\end{equation*}

To approximate functions of $d$-dimensional variables $\boldsymbol{x=}%
\mathbf{(}x_{1},\ldots ,x_{d})^{\top }$ on $\mathcal{X}=[0,1]^{d}$, \ we
employ a tensor product construction of the basis functions. We consider a
family of grids $\Omega _{\boldsymbol{\ell }}$ of level $\boldsymbol{\ell }%
=(\ell _{1},...,\ell _{d})^{\top }$ with interior points $\boldsymbol{x}_{%
\boldsymbol{\ell },\boldsymbol{s}}=\boldsymbol{s}\cdot \boldsymbol{h}_{%
\boldsymbol{\ell }}$, where $\boldsymbol{h}_{\boldsymbol{\ell }}=(h_{\ell
_{1}},...,h_{\ell _{d}})^{\top }$ with $h_{\ell _{j}}=2^{-\ell _{j}}$ and $%
\boldsymbol{s}=(s_{1},...,s_{d})^{\top }$ for $0\leq s_{j}\leq 2^{\ell _{j}}$
and $j=1,...,d$. On each $\Omega _{\boldsymbol{\ell }}$, the basis functions
$\phi _{\boldsymbol{\ell },\boldsymbol{s}}$ are given as
\begin{equation*}
\phi _{\boldsymbol{\ell },\boldsymbol{s}}(\boldsymbol{x})=\prod%
\limits_{j=1}^{d}\phi _{\ell _{j},s_{j}}(x_{j})\text{, }\boldsymbol{0}%
_{d}\leq \boldsymbol{s}\leq 2^{\boldsymbol{\ell }},
\end{equation*}%
and they satisfy $||\phi _{\boldsymbol{\ell },\boldsymbol{s}}||_{\infty
}\leq 1$. The hierarchical increment spaces $W_{\boldsymbol{\ell }}$ are
given by%
\begin{equation*}
W_{\boldsymbol{\ell }}=\text{span}\{\phi _{\boldsymbol{\ell },\boldsymbol{s}%
}(\boldsymbol{x}):\boldsymbol{s}\in I_{\boldsymbol{\ell }}\},
\end{equation*}%
where $I_{\boldsymbol{\ell }}=I_{\ell _{1}}\times \cdots \times I_{\ell
_{d}} $, and $I_{\ell j}=\{s_{j}\in \mathbb{N}:0\leq s_{j}\leq 2^{\ell
_{j}},s_{j}$ are odd numbers for $\ell _{j}\geq 1\}$. Then the hierarchical
space of functions up to level $\boldsymbol{L}=(L_{1},...,L_{d})^{\top }$ is
\begin{equation*}
V_{\boldsymbol{L}}=\bigoplus\limits_{\boldsymbol{0}\leq \boldsymbol{\ell }%
\leq \boldsymbol{L}}W_{\boldsymbol{\ell }}=\text{span}\{\phi _{\boldsymbol{%
\ell },\boldsymbol{s}}:\boldsymbol{s}\in I_{\boldsymbol{\ell }},\boldsymbol{0%
}_{d}\leq \boldsymbol{\ell }\leq \boldsymbol{L}\}.
\end{equation*}%
For any function $f\in W^{2,p}(\mathcal{X)}$, it has a unique expression in
the hierarchical basis \cite{BG04}:
\begin{equation}
f(\boldsymbol{x})=\sum_{\boldsymbol{0}_{d}\leq \boldsymbol{\ell }\leq
\boldsymbol{\infty }}\sum_{s\in I_{\boldsymbol{\ell }}}\gamma _{_{%
\boldsymbol{\ell },\boldsymbol{s}}}^{0}\phi _{\boldsymbol{\ell },\boldsymbol{%
s}}(\boldsymbol{x})=\sum_{\boldsymbol{0}_{d}\leq \boldsymbol{\ell }\leq
\boldsymbol{\infty }}g_{\boldsymbol{\ell }}(\boldsymbol{x}),  \label{EQ:f}
\end{equation}%
where $g_{\boldsymbol{\ell }}(\boldsymbol{x})=\sum_{s\in I_{\boldsymbol{\ell
}}}\gamma _{_{\boldsymbol{\ell },\boldsymbol{s}}}^{0}\phi _{\boldsymbol{\ell
},\boldsymbol{s}}(\boldsymbol{x})\in W_{\boldsymbol{\ell }}$. The
hierarchical coefficients $\gamma _{_{\boldsymbol{\ell },\boldsymbol{s}%
}}^{0}\in \mathbb{R}$ are given as (Lemma 3.2 of \cite{BG04}):
\begin{equation}
\gamma _{_{\boldsymbol{\ell },\boldsymbol{s}}}^{0}=\int_{\mathcal{X}%
}\prod\limits_{j=1}^{d}\left( -2^{-(\boldsymbol{\ell }_{j}+1)}\phi _{\ell
_{j},s_{j}}(x_{j})\right) D^{\boldsymbol{2}}f(\boldsymbol{x})d\boldsymbol{x},
\label{gammaexpression}
\end{equation}%
where $\boldsymbol{2}=2\boldsymbol{1}_{d}$, and satisfy (Lemma 3.3 of \cite%
{BG04})%
\begin{equation}
|\gamma _{_{\boldsymbol{\ell },\boldsymbol{s}}}^{0}|\leq 6^{-d/2}2^{-(3/2)|%
\boldsymbol{\ell }|_{1}}(\int\nolimits_{\boldsymbol{x}_{\boldsymbol{\ell },%
\boldsymbol{s}}-\boldsymbol{h}_{\boldsymbol{\ell }}}^{\boldsymbol{x}_{%
\boldsymbol{\ell },\boldsymbol{s}}+\boldsymbol{h}_{\boldsymbol{\ell }}}|D^{%
\boldsymbol{2}}f(\boldsymbol{x})|^{2}d\boldsymbol{x})^{1/2}\leq
6^{-d/2}2^{-(3/2)|\boldsymbol{\ell }|_{1}}||D^{\boldsymbol{2}}f||_{L^{2}}.
\label{EQ:gammals}
\end{equation}%
Moreover, the above result leads to (Lemma 3.4 of \cite{BG04})
\begin{equation}
||g_{\boldsymbol{\ell }}||_{L^{2}}\leq 3^{-d}2^{-2|\boldsymbol{\ell }%
|_{1}}(\int\nolimits_{\mathcal{X}}|D^{\boldsymbol{2}}f(\boldsymbol{x})|^{2}d%
\boldsymbol{x})^{1/2}=3^{-d}2^{-2|\boldsymbol{\ell }|_{1}}||D^{\boldsymbol{2}%
}f||_{L^{2}}.  \label{EQ:gsup}
\end{equation}%
Assumption \ref{ass3} and (\ref{EQ:gsup}) imply that
\begin{equation}
||g_{\boldsymbol{\ell }}||_{2}\leq c_{\mu }||g_{\boldsymbol{\ell }%
}||_{L^{2}}\leq c_{\mu }3^{-d}2^{-2|\boldsymbol{\ell }|_{1}}||D^{\boldsymbol{%
2}}f||_{L^{2}}.  \label{EQ:gsup2}
\end{equation}

In practice, one can use a truncated version to approximate the function $%
f(\cdot )$ given in (\ref{EQ:f}), so that
\begin{equation*}
f(\boldsymbol{x})\approx \sum\nolimits_{0\leq |\boldsymbol{\ell |}_{\infty
}\leq m}\sum\nolimits_{s\in I_{\boldsymbol{\ell }}}\gamma _{_{\boldsymbol{%
\ell },\boldsymbol{s}}}^{0}\phi _{\boldsymbol{\ell },\boldsymbol{s}}(%
\boldsymbol{x})=\sum\nolimits_{0\leq |\boldsymbol{\ell |}_{\infty }\leq m}g_{%
\boldsymbol{\ell }}(\boldsymbol{x}),
\end{equation*}%
which is constructed based on the space with full grids: $V_{m}^{\left(
\infty \right) }=\bigoplus\limits_{0\leq |\boldsymbol{\ell |}_{\infty }\leq
m}W_{\boldsymbol{\ell }}= \allowbreak $span$\{\phi _{\boldsymbol{\ell },\boldsymbol{s}%
}:s\in I_{\boldsymbol{\ell }},0\leq |\boldsymbol{\ell }|_{\infty }\leq m\}$.
The dimension of the space $V_{m}^{\left( \infty \right) }$ is $%
|V_{m}^{\left( \infty \right) }|=(2^{m}+1)^{d}$, which increases with $d$ in
an exponential order. For dimension reduction, we consider the hierarchical
space with sparse grids:
\begin{equation*}
V_{m}^{\left( 1\right) }=\bigoplus\limits_{|\boldsymbol{\ell |}_{1}\leq m}W_{%
\boldsymbol{\ell }}=\text{span}\{\phi _{\boldsymbol{\ell },\boldsymbol{s}%
}:s\in I_{\boldsymbol{\ell }},|\boldsymbol{\ell }|_{1}\leq m\}.
\end{equation*}%
The function $f(\cdot )$ given in (\ref{EQ:f}) can be approximated by
\begin{equation}
f_{m}(\boldsymbol{x})=\sum\nolimits_{|\boldsymbol{\ell |}_{1}\leq
m}\sum\nolimits_{s\in I_{\boldsymbol{\ell }}}\gamma _{_{\boldsymbol{\ell },%
\boldsymbol{s}}}^{0}\phi _{\boldsymbol{\ell },\boldsymbol{s}}(\boldsymbol{x}%
)=\sum\nolimits_{|\boldsymbol{\ell |}_{1}\leq m}g_{\boldsymbol{\ell }}(%
\boldsymbol{x}).  \label{EQ:fm}
\end{equation}%
Clearly, when $d=1$, the dimension of the hierarchical space with sparse
grids is the same as that of the space with full grids. The dimensionality
issue does not exist. For the rest of this paper, we assume that $d\geq 2$.
Table \ref{TAB:number} provides the number of basis functions for the
hierarchical space with sparse grids $V_{m}^{\left( 1\right) }$ and the
space with full grids $V_{m}^{\left( \infty \right) }$ when the dimension of
the covariates $d$ increases from 2 to 8 and the $m$ value increases from 0
to 4. We see that the number of basis functions for the space with sparse
grids is dramatically reduced compared to the space with full grids, when
the dimension $d$ or $m$ value become larger, so that the dimensionality
problem can be lessened.

\begin{table}[tbph]
\caption{The number of basis functions for the space with sparse grids and
the space with full grids.}
\label{TAB:number} \centering
\vskip 0.4cm
\scalebox{0.8}{
\begin{tabular}{cccccc|ccccc}
%\begin{tabular*}{1\textwidth}{@{\extracolsep{\fill}}cccccc|ccccc}
\cline{1-11}
& \multicolumn{5}{c}{Sparse grids} & \multicolumn{5}{|c}{Full grids} \\
\cline{1-11}
& $m=0$ & $m=1$ & $m=2$ & $m=3$ & $m=4$ & $m=0$ & $m=1$ & $m=2$ & $m=3$ & $m=4$ \\ \cline{1-11}
$d=2$ & $4$ & $8$ & $17$ & $37$ & $81$ & $4$ & $9$ & $25$ & $81$ & $289$ \\
$d=3$ & $8$ & $20$ & $50$ & $123$ & $297$ & $8$ & $27$ & $125$ & $729$ & $4913$ \\
$d=4$ & $16$ & $48$ & $136$ & $368$ & $961$ & $16$ & $81$ & $625$ & $6561$ &
$83521$ \\
$d=5$ & $32$ & $112$ & $352$ & $1032$ & $2882$ & $32$ & $243$ & $3125$ & $59049$ & $1419857$ \\
$d=6$ & $64$ & $256$ & $880$ & $2768$ & $8204$ & $64$ & $729$ & $15625$ & $531441$ & $24137569$ \\
$d=7$ & $128$ & $576$ & $2144$ & $7184$ & $22472$ & $128$ & $2187$ & $78125$
& $4782969$ & $410338673$ \\
$d=8$ & $256$ & $1280$ & $5120$ & $18176$ & $59744$ & $256$ & $6561$ & $390625$ & $43046721$ & $6975757441$ \\ \cline{1-11}
%\end{tabular*}
\end{tabular}
} 
\end{table}

The following proposition provides the approximation error of the
approximator $f_{m}\left( \cdot \right) $ obtained from the sparse grids to
the true unknown function $f\in W^{2,p}(\mathcal{X)}$.

\begin{proposition}
\label{PROP:error_sparse}For any $f\in W^{2,p}(\mathcal{X)}$, $2\leq p\leq
\infty $, under Assumption \ref{ass3}, one has that 
\begin{equation}
||f_{m}-f||_{2}\leq 4^{-1}c_{\mu }2^{-2m}(3/2)^{-d}(m+3)^{d-1}||D^{%
\boldsymbol{2}}f||_{L^{2}}.  \label{EQ:fm-f}
\end{equation}
\end{proposition}

Proposition \ref{PROP:error_sparse} shows that the approximator error
decreases as the $m$ value increases.

\textcolor{black}{In the following proposition, we provide an upper and a lower bounds for the
dimension (cardinality) of the space $V_{m}^{\left( 1\right) }.$
\begin{proposition}
\label{PROP:cardinality} The dimension of the space $V_{m}^{\left( 1\right)
} $ satisfies that for $d\geq 2$,
\textcolor{black}{
\begin{align*}
 2^{d-1}(2^{m}+1) \leq |V_{m}^{\left( 1\right) }| & \leq 4\frac{d}{(d-1)!}%
2^{d}2^{m}(m+d)^{d-1} \\
& \leq 4\sqrt{\frac{2}{\pi }}\frac{d}{\sqrt{d-1}}%
2^{m}\left( 2e\frac{m+d}{d-1}\right) ^{d-1}.
\end{align*}
}
\end{proposition}
\begin{remark}
\cite{BG04} gives an asymptotic order for the cardinality of $V_{m}^{\left(
1\right) }$ which is $|V_{m}^{\left( 1\right) }|=\mathcal{O(}%
c(d)2^{m}m^{d-1})$, where $c(d)$ is a constant depending on $d$, and the
form of $c(d)$ has not been provided. \cite{MD19} numerically demonstrated
how quickly $c(d)$ can increase with $d$. In Proposion \ref{PROP:cardinality}%
, we give an explicit form for the upper bound of $|V_{m}^{\left( 1\right)
}| $ that has not been derived in the literature. From this explicit form, we
can more clearly see how the dimension of $V_{m}^{\left( 1\right) }$
increases with $d$.
\end{remark}
}

\subsection{Proof of Proposition \protect\ref{PROP:error_sparse}} 

This section provides the proof of Proposition \ref{PROP:error_sparse}.
Based on (\ref{EQ:f}) and (\ref{EQ:fm}), one has
\begin{equation*}
||f_{m}-f||_{2}=||\sum_{\boldsymbol{0}_{d}\leq \boldsymbol{\ell }\leq
\boldsymbol{\infty }}g_{\boldsymbol{\ell }}(\boldsymbol{x})-\sum_{|%
\boldsymbol{\ell |}_{1}\leq m}g_{\boldsymbol{\ell }}(\boldsymbol{x}%
)||_{2}=||\sum_{|\boldsymbol{\ell |}_{1}>m}g_{\boldsymbol{\ell }}(%
\boldsymbol{x})||_{2}.
\end{equation*}%
By (\ref{EQ:gsup2}) and Assumption 3, one has%
\begin{eqnarray*}
||\sum_{|\boldsymbol{\ell |}_{1}>m}g_{\boldsymbol{\ell }}(\boldsymbol{x}%
)||_{2} &\leq &\sum_{|\boldsymbol{\ell |}_{1}>m}||g_{\boldsymbol{\ell }%
}||_{2}\leq \sum_{|\boldsymbol{\ell |}_{1}>m}c_{\mu }3^{-d}2^{-2|\boldsymbol{%
\ell }|_{1}}||D^{\boldsymbol{2}}f||_{L^{2}} \\
&=&c_{\mu }3^{-d}||D^{\boldsymbol{2}}f||_{L^{2}}\sum_{|\boldsymbol{\ell |}%
_{1}>m}2^{-2|\boldsymbol{\ell }|_{1}}.
\end{eqnarray*}%
Then, one has that for arbitrary $s\in \mathbb{N}$,
\begin{eqnarray*}
\sum_{|\boldsymbol{\ell |}_{1}>m}2^{-s|\boldsymbol{\ell }|_{1}}
&=&\sum_{k^{\prime }=m+1}^{\infty }2^{-sk^{\prime }}\binom{k^{\prime }+d-1}{%
d-1}\\
&=&\sum_{k=0}^{\infty }2^{-s\left( k+m+1\right) }\binom{k+m+1+d-1}{d-1} \\
&=& \textcolor{black}{2^{-s\left( m+1\right) }\sum_{k=0}^{\infty }2^{-sk}\binom{k+m+1+d-1}{d-1}} \\
& \leq & 2^{-s\left( m+1\right) }2A(d,m),
\end{eqnarray*}%
where $A(d,m)=\sum\nolimits_{k=0}^{d-1}\binom{m+d}{k}$, where the last
inequality follows from Lemma 3.7 of \cite{BG04}.%
\begin{eqnarray*}
||f_{m}-f||_{2} &\leq &\sum_{|\boldsymbol{\ell |}_{1}>m}||g_{\boldsymbol{%
\ell }}||_{2}\leq c_{\mu }3^{-d}||D^{\boldsymbol{2}}f||_{L^{2}}2^{-2\left(
m+1\right) }2A(d,m) \\
&=&2^{-1}c_{\mu }2^{-2m}3^{-d}A(d,m)||D^{\boldsymbol{2}}f||_{L^{2}}.
\end{eqnarray*}%
Moreover,
\begin{eqnarray*}
A(d,m) &=&\sum\nolimits_{k=0}^{d-1}\binom{d-1}{k}\frac{(m+d)!(d-1-k)!}{%
(m+d-k)!(d-1)!} \\
&=&\sum\nolimits_{k=0}^{d-1}\binom{d-1}{k}\left( \frac{m+d}{d-1}\right)
\left( \frac{m+d-1}{d-2}\right) \cdots \left( \frac{m+d-k+1}{d-k}\right) .
\end{eqnarray*}%
Since $\left( \frac{m+d-j+1}{d-j}\right) >\left( \frac{m+d-j^{\prime }+1}{%
d-j^{\prime }}\right) $ for any $1\leq j^{\prime }<j\leq d-1$,  then $\left(
\frac{m+2}{1}\right) >\left( \frac{m+d-j^{\prime }+1}{d-j^{\prime }}\right) $
for any $1\leq j^{\prime }\leq d-1$ by letting $j=d-1$, and thus%
\begin{equation*}
A(d,m)\leq \sum\nolimits_{k=0}^{d-1}\binom{d-1}{k}\left( m+2\right)
^{k}=\left( m+3\right) ^{d-1}.
\end{equation*}
Therefore,%
\begin{align*}
||f_{m}-f||_{2} & \leq 2^{-1}c_{\mu }2^{-2m}3^{-d}2^{d-1}(m+3)^{d-1}||D^{%
\boldsymbol{2}}f||_{L^{2}} \\ & \textcolor{black}{ =4^{-1}c_{\mu }2^{-2m}(3/2)^{-d}(m+3)^{d-1}||D^{%
\boldsymbol{2}}f||_{L^{2}}}.
\end{align*}
\normalsize

\subsection{Proof of Proposition \protect\ref{PROP:cardinality}}
In this section, we provide the proof of Proposition \ref{PROP:cardinality}. 
The dimension of $W_{\boldsymbol{\ell }}$ satisfies
\begin{equation*}
|W_{\boldsymbol{\ell }}|\leq \prod\nolimits_{j=1}^{d}2^{\ell _{j}\vee
2-1}=2^{\sum_{j=1}^{d}\ell _{j}\vee 2-d}.
\end{equation*}%
Thus,%
\begin{eqnarray*}
|V_{m}^{\left( 1\right) }| &\leq &\sum\nolimits_{|\boldsymbol{\ell |}%
_{1}\leq m}2^{\sum_{j=1}^{d}\ell _{j}\vee 2-d}\leq \sum\nolimits_{|%
\boldsymbol{\ell |}_{1}\leq m}2^{|\boldsymbol{\ell |}_{1}+d}=%
\sum_{k=0}^{m}2^{k+d}\binom{d-1+k}{d-1} \\
&=&2^{d}\sum_{k=0}^{m}2^{k}\binom{d-1+k}{d-1}=2^{d}\left\{
(-1)^{d}+2^{m+1}\sum_{k=0}^{d-1}\binom{m+d}{k}(-2)^{d-1-k}\right\} ,
\end{eqnarray*}%
where the last equality follows from (3.62) of \cite{BG04}. We assume that $d
$ is even. The result for odd $d$ can be proved similarly. Then
\begin{equation*}
\sum_{k=0}^{d-1}\binom{m+d}{k}(-2)^{d-1-k}=\sum_{v=0}^{d/2-1}2^{2v}\left\{
\binom{m+d}{d-(1+2v)}-2\binom{m+d}{d-(2+2v)}\right\} .
\end{equation*}%
Moreover,
\begin{eqnarray*}
&&\binom{m+d}{d-(1+2v)}-2\binom{m+d}{d-(2+2v)} \\
&=&\frac{(m+d)!}{(d-(1+2v))!(m+1+2v)!}-2\frac{(m+d)!}{(d-(2+2v))!(m+2v+2)!}
\\
&=&\frac{(m+d)!}{(d-(2+2v))!(m+1+2v)!}\left\{ \frac{1}{d-(1+2v)}-\frac{2}{%
m+2v+2}\right\}  \\
&=&\frac{(m+d)!(m+6v-2d+4)}{(d-(1+2v))!(m+2v+2)!}\\
&=&\frac{(m+d)(m+d-1)\times \cdot \cdot \cdot \times (m+2v+3)(m+6v-2d+4)}{%
(d-(1+2v))!} \\
&\leq &\frac{(m+d)^{(d-(1+2v))}}{(d-(1+2v))!}.
\end{eqnarray*}%
Thus,
\begin{eqnarray*}
\sum_{k=0}^{d-1}\binom{m+d}{k}(-2)^{d-1-k} &\leq &\sum_{v=0}^{d/2-1}2^{2v}%
\frac{(m+d)^{(d-(1+2v))}}{(d-(1+2v))!}\\
& \leq & \textcolor{black}{\frac{(m+d)^{(d-1)}}{(d-1)!}%
\sum_{v=0}^{d/2-1}\frac{2^{2v}}{(m+d)^{2v}}} \\
&\leq &\frac{(m+d)^{(d-1)}}{(d-1)!}\frac{m+d}{m+d-2}\leq \frac{(m+d)^{d-1}}{%
(d-1)!}d,
\end{eqnarray*}%
for $m+d\geq 3$ and $d\geq 3$. It is easy to verify that $\sum_{k=0}^{d-1}%
\binom{m+d}{k}(-2)^{d-1-k}\leq \textcolor{black}{ (m+d)^{(d-1)}d/{(d-1)!}}$ when $d\leq 2$
or $m+d\leq 2$.

By stirling's formula,
\begin{equation*}
(d-1)!\geq \sqrt{2\pi }(d-1)^{d-1/2}e^{-(d-1)}.
\end{equation*}%
Therefore, for $d\geq 2$,
\begin{equation*}
\sum_{k=0}^{d-1}\binom{m+d}{k}(-2)^{d-1-k}\leq \frac{(m+d)^{d-1}}{(d-1)!}%
d\leq \frac{(m+d)^{d-1}d}{\sqrt{2\pi }(d-1)^{d-1/2}e^{-(d-1)}},
\end{equation*}%
and hence%
\begin{eqnarray*}
|V_{m}^{\left( 1\right) }| &\leq &2^{d+1}2^{m+1}\frac{(m+d)^{d-1}}{(d-1)!}%
d\leq 2^{d+1}2^{m+1}\frac{(m+d)^{d-1}d}{\sqrt{2\pi }(d-1)^{d-1/2}e^{-(d-1)}}
\\
&=&4\sqrt{\frac{2}{\pi }}\frac{d}{\sqrt{d-1}}2^{m}\left( 2e\frac{m+d}{d-1}%
\right) ^{d-1}.
\end{eqnarray*}%
Moreover, let $\boldsymbol{\ell }_{-1}=(\ell _{2},\ldots ,\ell _{d})^{\top }$%
. Then, $|V_{0}^{\left( 1\right) }|=\sum\nolimits_{|\boldsymbol{\ell |}%
_{1}=0}\prod\nolimits_{j=1}^{d}2=2^{d}$, and for $m\geq 1$,
\begin{eqnarray*}
|V_{m}^{\left( 1\right) }| &\geq &\sum\nolimits_{|\boldsymbol{\ell |}%
_{1}=0}\prod\nolimits_{j=1}^{d}2+\sum\nolimits_{1\leq \ell _{1}\leq m,|%
\boldsymbol{\ell }_{-1}\boldsymbol{|}_{1}=0}(\prod\nolimits_{j=2}^{d}2)2^{%
\ell _{1}-1}\\
&=& \textcolor{black}{2^{d}+2^{d-1}\sum\nolimits_{1\leq \ell _{1}\leq m}2^{\ell
_{1}-1} =2^{d}+2^{d-1}(2^{m}-1)}\\
&=& 2^{d-1}(2^{m}-1+2)\geq 2^{d-1}(2^{m}+1).
\end{eqnarray*}%
Therefore, $|V_{m}^{\left( 1\right) }|\geq 2^{d-1}(2^{m}+1)$ for any $m\geq 0
$.

\subsection{Proof of Proposition \protect\ref{THM:error_ReLU}}

In this section, we provide the proof of Proposition \ref{THM:error_ReLU}.
It is clear that $||\widetilde{f}-f||_{2}=||\widetilde{f}%
-f_{m}+f_{m}-f||_{2}\leq ||\widetilde{f}-f_{m}||_{2}+||f_{m}-f||_{2}$. The
rate of $||f_{m}-f||_{2}$ is provided in Proposition \ref{PROP:error_sparse}%
. Next we derive the rate of $||\widetilde{f}-f_{m}||_{2}$ as follows. By (%
\ref{EQ:fm}) and (\ref{EQ:ReLUapproximator}), we have
\begin{equation*}
||\widetilde{f}-f_{m}||_{2}\leq \sup_{\boldsymbol{x}\in \mathcal{X}}\sum_{|%
\boldsymbol{\ell |}_{1}\leq m}\sum_{s\in I_{\boldsymbol{\ell }}}|\gamma _{%
\boldsymbol{\ell },\boldsymbol{s}}^{0}||\textcolor{black}{\widetilde{\varphi}_R({\phi }_{\boldsymbol{\ell },%
\boldsymbol{s}}(\boldsymbol{x}))}-\phi _{\boldsymbol{\ell },\boldsymbol{s}}(%
\boldsymbol{x})|.
\end{equation*}%
Since a given $\boldsymbol{x}$ belongs to at most one of the disjoint
supports for $\phi _{\boldsymbol{\ell },\boldsymbol{s}}(\boldsymbol{x})$,
this result together with (\ref{EQ:phitilda}) lead to%
\begin{equation*}
||\widetilde{f}-f_{m}||_{2}\leq 3\cdot 2^{-2R-2}(d-1)\sum_{|\boldsymbol{\ell
|}_{1}\leq m}|\gamma _{\boldsymbol{\ell },\boldsymbol{s}_{\boldsymbol{\ell }%
}}^{0}|,
\end{equation*}%
for some $\boldsymbol{s}_{\boldsymbol{\ell }}$. Moreover, by (\ref%
{EQ:gammals}), we have%
\begin{eqnarray*}
\sum_{|\boldsymbol{\ell |}_{1}\leq m}|\gamma _{_{\boldsymbol{\ell },%
\boldsymbol{s}_{\boldsymbol{\ell }}}}^{0}| &\leq &\sum_{|\boldsymbol{\ell |}%
_{1}\leq m}6^{-d/2}2^{-(3/2)|\boldsymbol{\ell }|_{1}}||D^{\boldsymbol{2}%
}f||_{L^{2}} \\
&=& \textcolor{black}{ 6^{-d/2}||D^{\boldsymbol{2}}f||_{L^{2}}\sum_{|\boldsymbol{\ell |}_{1}\leq
m}2^{-(3/2)|\boldsymbol{\ell }|_{1}} }\\
&=&6^{-d/2}||D^{\boldsymbol{2}%
}f||_{L^{2}}\sum\nolimits_{k=0}^{m}2^{-(3/2)k}\binom{k+d-1}{d-1}.
\end{eqnarray*}%

Since $\sum\nolimits_{k=0}^{\infty }\left( \frac{1}{\sqrt{8}}\right)
^{k}\left( 1-\frac{1}{\sqrt{8}}\right) ^{d}\binom{k+d-1}{d-1}=1$, it implies
that
\begin{equation*}
\sum\nolimits_{k=0}^{m}2^{-(3/2)k}\binom{k+d-1}{d-1}\leq
\sum\nolimits_{k=0}^{\infty }2^{-(3/2)k}\binom{k+d-1}{d-1}=\left( 1-\frac{1}{%
\sqrt{8}}\right) ^{-d},
\end{equation*}%
and thus
\begin{equation}
\sum_{|\boldsymbol{\ell |}_{1}\leq m}|\gamma _{\boldsymbol{\ell },%
\boldsymbol{s}_{\boldsymbol{\ell }}}^{0}|\leq \left\{ \sqrt{6}(1-\frac{1}{\sqrt{8%
}})\right\} ^{-d}||D^{\boldsymbol{2}}f||_{L^{2}}\leq (\sqrt{3/2})^{-d}||D^{%
\boldsymbol{2}}f||_{L^{2}}.  \label{EQ:Gammals}
\end{equation}%
Therefore,
\begin{eqnarray*}
||\widetilde{f}-f_{m}||_{2} &\leq &3\cdot 2^{-2R-2}(d-1)(\sqrt{3/2}%
)^{-d}||D^{\boldsymbol{2}}f||_{L^{2}} \\
&=&(3/4)2^{-2R}(d-1)(\sqrt{3/2})^{-d}||D^{\boldsymbol{2}}f||_{L^{2}}.
\end{eqnarray*}%

The above result and (\ref{EQ:fm-f}) lead to
\begin{eqnarray*}
||\widetilde{f}-f||_{2} &\leq &||\widetilde{f}-f_{m}||_{2}+||f_{m}-f||_{2} \\
&\leq &\left\{ (3/4)2^{-2R}(d-1)(\sqrt{3/2})^{-d}+4^{-1}c_{\mu
}2^{-2m}(3/2)^{-d}(m+3)^{d-1}\right\} ||D^{\boldsymbol{2}}f||_{L^{2}} \\
&\leq &\left\{ (3/2)2^{-2R}+4^{-1}c_{\mu
}2^{-2m}(3/2)^{-d}(m+3)^{d-1}\right\} ||D^{\boldsymbol{2}}f||_{L^{2}} \\
&=&\left\{ (3/2)2^{-2R}+6^{-1}c_{\mu }2^{-2m}\{(2/3)(m+3)\}^{d-1}\right\}
||D^{\boldsymbol{2}}f||_{L^{2}}.
\end{eqnarray*}%

Moreover, the ReLU network used to construct the approximator $\widetilde{f}$
has depth \linebreak $\mathcal{O}(R\log _{2}d)$, \textcolor{black}{ the number of weights $\mathcal{O%
}(Rd)\times |V_{m}^{\left( 1\right) }|$, and the computational units $\mathcal{O}%
(Rd)\times |V_{m}^{\left( 1\right) }|$,. By the upper bound for $%
|V_{m}^{\left( 1\right) }|$ established in Proposition (\ref{PROP:cardinality}), we have
that the number of weights is $\mathcal{O}\left(
2^{m}d^{3/2}R\left( 2e\frac{m+d}{d-1}\right) ^{d-1}\right) $, and the number of the computational units is $\mathcal{O}(Rd)\times
\mathcal{O}\left( 2^{m}\frac{d}{\sqrt{d-1}}\left( 2e\frac{m+d}{d-1}\right)
^{d-1}\right) =\mathcal{O}\left( 2^{m}d^{3/2}R\left( 2e\frac{m+d}{d-1}%
\right) ^{d-1}\right) $}.

\subsection{Proofs of Proposition \protect\ref{THM:error_ReLUbound}}

Under Condition (\ref{Bernstein2}) given in Assumption \ref{ass5}, by the
definition of $f^{\ast }$ given in (\ref{EQ:f0RN}) and Proposition \ref%
{THM:error_ReLU}, the approximation error satisfies that for $d\geq 2$ and $%
R\geq m$,
\begin{eqnarray*}
\mathcal{E}(f^{\ast })-\mathcal{E}(f_{0}) &\leq &\mathcal{E}(\widetilde{f})-%
\mathcal{E}(f_{0})\leq b_{\rho }||\widetilde{f}-f_{0}||_{2}^{2} \\
&\leq &b_{\rho }\left\{ (3/2)2^{-2R}+6^{-1}c_{\mu
}2^{-2m}\{(2/3)(m+3)\}^{d-1}\right\} ^{2}||D^{\boldsymbol{2}%
}f_{0}||_{L^{2}}^{2} \\
&\leq &b_{\rho }\left\{ (3/2)2^{-2m}+6^{-1}c_{\mu
}2^{-2m}\{(2/3)(m+3)\}^{d-1}\right\} ^{2}||D^{\boldsymbol{2}%
}f_{0}||_{L^{2}}^{2} \\
&\leq &\zeta _{m,d},
\end{eqnarray*}%
where
\begin{equation*}
\zeta _{m,d}=4^{-1}b_{\rho }(3+c_{\mu }/3)^{2}2^{-4m}\{(2/3)(m+3)\}^{2\left(
d-1\right) }||D^{\boldsymbol{2}}f_{0}||_{L^{2}}^{2}.
\end{equation*}

\subsection{Proofs of theorems \protect\ref{THM:sampliingerror1} and \protect
\ref{THM:sampliingerror2}}

We first introduce a Bernstein inequality which will be used to establish
the bounds in Theorems \ref{THM:sampliingerror1} and \ref%
{THM:sampliingerror2}.

\begin{lemma}
\label{LEM:bernstein}Let $\mathcal{G}$ be a set of scalar-valued functions on%
$\ \mathcal{X}\times \mathcal{Y}$ such that for each $\xi \left( \boldsymbol{%
X},Y\right) \in \mathcal{G}$, $\mathbb{E\{}\xi \left( \boldsymbol{X}%
,Y\right) \}\geq 0$, $\mathbb{E\{}\xi \left( \boldsymbol{X},Y\right)
^{2}\}\leq c_{1}\mathbb{E\{}\xi \left( \boldsymbol{X},Y\right) \}$ and $|\xi
\left( \boldsymbol{X},Y\right) -\mathbb{E\{}\xi \left( \boldsymbol{X}%
,Y\right) \}|\leq c_{2}$ almost everywhere for some constants $%
c_{1},c_{2}\in (0,\infty )$. Then for every $\epsilon >0$ and $0<\alpha \leq
1$, we have
\begin{eqnarray*}
&&P\left\{ \sup_{\xi \in \mathcal{G}}\frac{\mathbb{E\{}\xi \left(
\boldsymbol{X},Y\right) \}-n^{-1}\sum\nolimits_{i=1}^{n}\xi \left(
\boldsymbol{X}_{i},Y_{i}\right) }{\sqrt{\mathbb{E\{}\xi \left( \boldsymbol{X}%
,Y\right) \}+\epsilon }}>4\alpha \sqrt{\epsilon }\right\} \\
&\leq &\mathcal{N}(\alpha \epsilon ,\mathcal{G},||\cdot ||_{\infty })\exp
\left( -\frac{\alpha ^{2}n\epsilon }{2c_{1}+2c_{2}/3}\right) .
\end{eqnarray*}
\end{lemma}

\begin{proof}
Let $\{\xi _{j}\}_{j=1}^{J}\in \mathcal{G}$ $\ $with $J=\mathcal{N}(\alpha
\epsilon ,\mathcal{G},||\cdot ||_{\infty })$ being such that $\mathcal{G}$
is covered by $||\cdot ||_{\infty }$- balls centered on $\xi _{j}$ with
radius $\alpha \epsilon $. Denote $\mu (\xi )=\mathbb{E\{}\xi \left(
\boldsymbol{X},Y\right) \}$ and $\sigma ^{2}(\xi )=$var$\mathbb{\{}\xi
\left( \boldsymbol{X},Y\right) \}$. For each $j$, by applying the one-side
Bernstein inequality in Corollary 3.6 of \cite{CZ07}, one has
\begin{eqnarray}
&&P\left\{ \frac{\mu (\xi _{j})-n^{-1}\sum\nolimits_{i=1}^{n}\xi _{j}\left(
\boldsymbol{X}_{i},Y_{i}\right) }{\sqrt{\mu (\xi _{j})+\epsilon }}>\alpha
\sqrt{\epsilon }\right\}  \notag \\
&\leq &\exp \left( -\frac{\alpha ^{2}n(\mu (\xi _{j})+\epsilon )\epsilon }{%
2\{\sigma ^{2}(\xi _{j})+c_{2}\alpha \sqrt{\mu (\xi _{j})+\epsilon }\sqrt{%
\epsilon }/3\}}\right) .  \label{Pxi}
\end{eqnarray}%
Since $\sigma ^{2}(\xi _{j})\leq \mathbb{E\{}\xi _{j}\left( \boldsymbol{X}%
,Y\right) ^{2}\}\leq c_{1}\mu (\xi _{j})$, then%
\begin{eqnarray*}
&&\sigma ^{2}(\xi _{j})+c_{2}\alpha \sqrt{\mu (\xi _{j})+\epsilon }\sqrt{%
\epsilon }/3 \\
&\leq &c_{1}\mu (\xi _{j})+c_{2}(\mu (\xi _{j})+\epsilon )/3 \\
&\leq &c_{1}\left( \mu (\xi _{j})+\epsilon \right) +c_{2}(\mu (\xi
_{j})+\epsilon )/3 \\
&=&(c_{1}+c_{2}/3)(\mu (\xi _{j})+\epsilon ).
\end{eqnarray*}%
The above result together with (\ref{Pxi}) implies that
\begin{eqnarray}
&&P\left\{ \frac{\mu (\xi _{j})-n^{-1}\sum\nolimits_{i=1}^{n}\xi _{j}\left(
\boldsymbol{X}_{i},Y_{i}\right) }{\sqrt{\mu (\xi _{j})+\epsilon }}>\alpha
\sqrt{\epsilon }\right\}  \notag \\
&\leq &\exp \left( -\frac{\alpha ^{2}n(\mu (\xi _{j})+\epsilon )\epsilon }{%
2(c_{1}+c_{2}/3)(\mu (\xi _{j})+\epsilon )}\right) =\exp \left( -\frac{%
\alpha ^{2}n\epsilon }{2(c_{1}+c_{2}/3)}\right) .  \label{xij}
\end{eqnarray}%
For each $\xi \in \mathcal{G}$, there exists some $j$ such that $||\xi -\xi
_{j}||_{\infty }\leq \alpha \epsilon $. Then $|\mu (\xi )-\mu (\xi _{j})|$
and $|n^{-1}\sum\nolimits_{i=1}^{n}\xi \left( \boldsymbol{X}%
_{i},Y_{i}\right) -n^{-1}\sum\nolimits_{i=1}^{n}\xi _{j}\left( \boldsymbol{X}%
_{i},Y_{i}\right) |$ are both bounded by $\alpha \epsilon $. Hence,%
\begin{equation*}
\frac{|\mu (\xi )-\mu (\xi _{j})|}{\sqrt{\mu (\xi )+\epsilon }}\leq \alpha
\sqrt{\epsilon },\frac{|n^{-1}\sum\nolimits_{i=1}^{n}\xi \left( \boldsymbol{X%
}_{i},Y_{i}\right) -n^{-1}\sum\nolimits_{i=1}^{n}\xi _{j}\left( \boldsymbol{X%
}_{i},Y_{i}\right) |}{\sqrt{\mu (\xi )+\epsilon }}\leq \alpha \sqrt{\epsilon
}.
\end{equation*}%
This implies that
\begin{eqnarray*}
\mu (\xi _{j})+\epsilon &=&\mu (\xi _{j})-\mu (\xi )+\mu (\xi )+\epsilon \\
&\leq &\alpha \sqrt{\epsilon }\sqrt{\mu (\xi )+\epsilon }+\{\mu (\xi
)+\epsilon \} \\
&\leq &\sqrt{\epsilon }\sqrt{\mu (\xi )+\epsilon }+\{\mu (\xi )+\epsilon \}
\\
&\leq &2\{\mu (\xi )+\epsilon \},
\end{eqnarray*}%
so that $\sqrt{\mu (\xi _{j})+\epsilon }\leq 2\sqrt{\{\mu (\xi )+\epsilon \}}
$. Therefore, $\textcolor{black}{\{\mu (\xi )-n^{-1}\sum\nolimits_{i=1}^{n}\xi \left(
\boldsymbol{X}_{i},Y_{i}\right) \}/ } \linebreak \textcolor{black}{\sqrt{\mu (\xi )+\epsilon }\geq 4\alpha\sqrt{\epsilon }}$ implies that $\{\mu (\xi
_{j})-n^{-1}\sum\nolimits_{i=1}^{n}\xi _{j}\left( \boldsymbol{X}%
_{i},Y_{i}\right) \}/\sqrt{\mu (\xi )+\epsilon }\geq 2\alpha \sqrt{\epsilon }$ and thus $\{\mu (\xi _{j})-n^{-1}\sum\nolimits_{i=1}^{n}\xi _{j}\left(
\boldsymbol{X}_{i},Y_{i}\right) \}/\sqrt{\mu (\xi _{j})+\epsilon }%
\geq \alpha \sqrt{\epsilon }$. This result together with (\ref{xij}) implies
\begin{eqnarray*}
&&P\left\{ \sup_{\xi \in \mathcal{G}}\frac{\mu (\xi
)-n^{-1}\sum\nolimits_{i=1}^{n}\xi \left( \boldsymbol{X}_{i},Y_{i}\right) }{%
\sqrt{\mu (\xi )+\epsilon }}>4\alpha \sqrt{\epsilon }\right\} \\
&\leq &\sum\nolimits_{j=1}^{J}P\left\{ \frac{\mu (\xi
_{j})-n^{-1}\sum\nolimits_{i=1}^{n}\xi _{j}\left( \boldsymbol{X}%
_{i},Y_{i}\right) }{\sqrt{\mu (\xi _{j})+\epsilon }}>\alpha \sqrt{\epsilon }%
\right\} \leq J\exp \left( -\frac{\alpha ^{2}n\epsilon }{2c_{1}+2c_{2}/3}%
\right).
\end{eqnarray*}
\end{proof}

Based on the Bernstein inequality given in Lemma \ref{LEM:bernstein}, we
next provide a probability bound that will be used for establishing an upper
bound for the sampling error $\mathcal{E}(\widehat{f})-\mathcal{E}(f^{\ast
}) $. For notational simplicity, we denote $\mathcal{E}_{n}(f;\boldsymbol{%
\Theta }):=\mathcal{E}_{n}(f)$.

\begin{lemma}
\label{LEM:probbound}Under Assumptions \ref{ass1}-\ref{ass4}, we have that
for any $\epsilon >0$ and $0<\alpha \leq 1,$
\begin{eqnarray*}
&&P\left\{ \sup_{f\in \textcolor{black}{ \mathcal{F(%
}R ,m,B_{0},B_{1})}}\frac{\mathcal{E}(f)-%
\mathcal{E}(f^{\ast })-(\mathcal{E}_{n}(f)-\mathcal{E}_{n}(f^{\ast }))}{%
\sqrt{\mathcal{E}(f)-\mathcal{E}(f^{\ast })+\epsilon }}>4\alpha \sqrt{%
\epsilon }\right\} \\
&\leq &\mathcal{N}(\alpha C_{\rho }^{-1}\epsilon ,\textcolor{black}{ \mathcal{F(%
}R ,m,B_{0},B_{1})}),||\cdot ||_{\infty })\exp \left( -\frac{\alpha ^{2}n\epsilon
}{2C_{\rho }^{2}a_{\rho }^{-1}+8M_{\rho }/3}\right) ,
\end{eqnarray*}%
where $C_{\rho },a_{\rho }$ and $M_{\rho }$ are constants given in
Assumptions \ref{ass2} and \ref{ass4} and Remark \ref{rho}.
\end{lemma}

\begin{proof}
Let $\mathcal{G}=\{\xi \left( \boldsymbol{x},y\right) =\rho \left( f(%
\boldsymbol{x}),y\right) -\rho \left( f^{\ast }(\boldsymbol{x}),y\right)
;f\in \textcolor{black}{ \mathcal{F(%
}R ,m,B_{0},B_{1})},(\boldsymbol{x},y)\in \mathcal{X}%
\times \mathcal{Y}\}$. For any $f\in \textcolor{black}{ \mathcal{F(%
}R ,m,B_{0},B_{1})}$,
\begin{equation*}
\mathbb{E\{}\xi \left( \boldsymbol{X},Y\right) \}=\mathbb{E}\{\rho \left( f(%
\boldsymbol{X}),Y\right) \}-\mathbb{E}\{\rho \left( f^{\ast }(\boldsymbol{X}%
),Y\right) \}\geq 0,
\end{equation*}%
based on the definition of $f^{\ast }$ given in (\ref{EQ:f0RN}). By Remark %
\ref{rho}, we have $|\xi \left( \boldsymbol{x},y\right) |\leq 2M_{\rho }$,
for almost every $(\boldsymbol{x},y)\in \mathcal{X}\times \mathcal{Y}$, so
that
\begin{equation*}
|\xi \left( \boldsymbol{X},Y\right) -\mathbb{E\{}\xi \left( \boldsymbol{X}%
,Y\right) \}|\leq 4M_{\rho },
\end{equation*}%
almost surely. Moreover, Assumption \ref{ass2} further implies that $|\xi
\left( \boldsymbol{x},y\right) |\leq C_{\rho }|f(\boldsymbol{x})-f^{\ast }(%
\boldsymbol{x})|$ for almost every $(\boldsymbol{x},y)\in \mathcal{X}\times
\mathcal{Y}$. Then
\begin{equation}
\mathbb{E}\{\xi \left( \boldsymbol{X},Y\right) ^{2}\}\leq C_{\rho }^{2}\int_{%
\mathcal{X}}|f(\boldsymbol{x})-f^{\ast }(\boldsymbol{x})|^{2}d\mu _{X}(%
\boldsymbol{x})=C_{\rho }^{2}||f-f^{\ast }||_{2}^{2}.  \label{EQ:Exi}
\end{equation}%
Moreover, under Condition (\ref{Bernstein}) in Assumption \ref{ass4},%
\begin{equation*}
||f-f^{\ast }||_{2}^{2}\leq a_{\rho }^{-1}\{\mathcal{E}(f)-\mathcal{E}%
(f^{\ast })\}.
\end{equation*}%
Thus%
\begin{equation}
\mathbb{E}\{\xi \left( \boldsymbol{X},Y\right) ^{2}\}\leq a_{\rho
}^{-1}C_{\rho }^{2}\{\mathcal{E}(f)-\mathcal{E}(f^{\ast })\}=a_{\rho
}^{-1}C_{\rho }^{2}\mathbb{E\{}\xi \left( \boldsymbol{X},Y\right) \}.
\label{EQ:meanxi}
\end{equation}%
By the Bernstein inequality given in Lemma \ref{LEM:bernstein}, for every $%
\epsilon >0$ and $0<\alpha \leq 1$, we have
\begin{eqnarray*}
&&P\left\{ \sup_{f\in \textcolor{black}{ \mathcal{F(%
}R ,m,B_{0},B_{1})}}\frac{\mathcal{E}(f)-\mathcal{%
E}(f^{\ast })-(\mathcal{E}_{n}(f)-\mathcal{E}_{n}(f^{\ast }))}{\sqrt{%
\mathcal{E}(f)-\mathcal{E}(f^{\ast })+\epsilon }}>4\alpha \sqrt{\epsilon }%
\right\} \\
&\leq &\mathcal{N}(\alpha \epsilon ,\mathcal{G},||\cdot ||_{\infty })\exp
\left( -\frac{\alpha ^{2}n\epsilon }{2C_{\rho }^{2}a_{\rho }^{-1}+8M_{\rho
}/3}\right) .
\end{eqnarray*}%
Since $|\rho \left( f(\boldsymbol{x}),y\right) -\rho \left( f^{\ast }(%
\boldsymbol{x}),y\right) |\leq C_{\rho }|f(\boldsymbol{x})-f^{\ast }(%
\boldsymbol{x})|$ for almost every $(\boldsymbol{x},y)\in \mathcal{X}\times
\mathcal{Y}$, it follows that
\begin{equation*}
\mathcal{N}(\alpha \epsilon ,\mathcal{G},||\cdot ||_{\infty })\leq \mathcal{N%
}(\alpha C_{\rho }^{-1}\epsilon ,\textcolor{black}{ \mathcal{F(%
}R ,m,B_{0},B_{1})},||\cdot
||_{\infty }).
\end{equation*}
\end{proof}

\begin{proof}[Proof of Theorem \protect\ref{THM:sampliingerror1}]
Let $f=\widehat{f}$, $\Delta =\mathcal{E}(\widehat{f})-\mathcal{E}(f^{\ast })
$ and $\alpha =\sqrt{2}/8$. From the result in Lemma \ref{LEM:probbound}, we
have
\begin{equation}
P\left\{ \frac{\Delta -(\mathcal{E}_{n}(\widehat{f})-\mathcal{E}_{n}(f^{\ast
}))}{\sqrt{\Delta +\epsilon }}>\sqrt{\epsilon /2}\right\} \leq Q,  \label{Q}
\end{equation}%
where
\begin{equation*}
Q=\mathcal{N}(\sqrt{2}C_{\rho }^{-1}\epsilon /8,\textcolor{black}{ \mathcal{F(%
}R ,m,B_{0},B_{1})},||\cdot ||_{\infty })\exp \left( -n\epsilon /C^{\ast
}\right),
\end{equation*}%
in which $C^{\ast }=64(C_{\rho }^{2}a_{\rho }^{-1}+4M_{\rho }/3)$. Let $%
\Delta _{n}=(\mathcal{E}_{n}(\widehat{f})+2^{-1}\lambda \widehat{\boldsymbol{%
\Theta }}^{\top }\widehat{\boldsymbol{\Theta }}-\mathcal{E}_{n}(f^{\ast
})-2^{-1}\lambda \boldsymbol{\Theta }^{\ast \top }\boldsymbol{\Theta }^{\ast
})$. We have $\Delta _{n}\leq \varpi _{n}$, $\widehat{\boldsymbol{\Theta }}%
^{\top }\widehat{\boldsymbol{\Theta }}\leq B_{1}^{2}$ and $\boldsymbol{%
\Theta }^{\ast \top }\boldsymbol{\Theta }^{\ast }\leq B_{1}^{2}$. Then
\begin{eqnarray}
\frac{\Delta -(\mathcal{E}_{n}(\widehat{f})-\mathcal{E}_{n}(f^{\ast }))}{%
\sqrt{\Delta +\epsilon }} &=&\frac{\Delta +(2^{-1}\lambda \widehat{%
\boldsymbol{\Theta }}^{\top }\widehat{\boldsymbol{\Theta }}-2^{-1}\lambda
\boldsymbol{\Theta }^{\ast \top }\boldsymbol{\Theta }^{\ast })-\Delta _{n}}{%
\sqrt{\Delta +\epsilon }}  \notag \\
&\geq &\frac{\Delta -\varpi _{n}}{\sqrt{\Delta +\epsilon }}+\frac{%
2^{-1}\lambda (\widehat{\boldsymbol{\Theta }}^{\top }\widehat{\boldsymbol{%
\Theta }}-\boldsymbol{\Theta }^{\ast \top }\boldsymbol{\Theta }^{\ast })}{%
\sqrt{\Delta +\epsilon }}.  \notag \\
&\geq &\frac{\Delta -\varpi _{n}}{\sqrt{\Delta +\epsilon }}-\frac{%
2^{-1}\lambda 2B_{1}^{2}}{\sqrt{\Delta +\epsilon }} \notag \\
&=&\frac{\Delta }{\sqrt{%
\Delta +\epsilon }}-\frac{(\varpi _{n}+\lambda B_{1}^{2})}{\sqrt{\Delta
+\epsilon }}.  \label{EQ:Delta}
\end{eqnarray}%
The above result and (\ref{Q}) lead to
\begin{equation*}
P\left\{ \Delta >\sqrt{\epsilon /2}\sqrt{\Delta +\epsilon }+(\varpi
_{n}+\lambda B_{1}^{2})\right\} \leq Q.
\end{equation*}%
Moreover,%
\begin{align*}
\text{ \ }\Delta & >\sqrt{\epsilon /2}\sqrt{\Delta +\epsilon }+\varpi
_{n}+\lambda B_{1}^{2} \\
& \Longleftrightarrow (\Delta -(\varpi _{n}+\lambda
B_{1}^{2}))^{2}>(\epsilon /2)(\Delta +\epsilon ) \\
\text{ \ \ \ \ \ \ }& \Longleftrightarrow (\Delta -(\varpi _{n}+\lambda
B_{1}^{2})-\epsilon /4)^{2}>(9/16)\epsilon ^{2} \\
\text{ \ \ \ }& \Longleftrightarrow \Delta -(\varpi _{n}+\lambda
B_{1}^{2})>\epsilon \text{ or }\Delta -(\varpi _{n}+\lambda
B_{1}^{2})<-(1/2)\epsilon .
\end{align*}%
Since $\Delta \geq 0$ and $\varpi _{n}+\lambda B_{1}^{2}<(1/2)\epsilon $,
then $P\left\{ \Delta >\sqrt{\epsilon /2}\sqrt{\Delta +\epsilon }+(\varpi
_{n}+\lambda B_{1}^{2})\right\} \leq Q$ is equivalent to $P(\Delta >\epsilon
+\varpi _{n}+\lambda B_{1}^{2})\leq Q$. Given that $\varpi _{n}+\lambda
B_{1}^{2}<(1/2)\epsilon $, one has $P(\Delta >(3/2)\epsilon )\leq P(\Delta
>\epsilon +\varpi _{n}+\lambda B_{1}^{2})\leq Q.$
\end{proof}

\begin{proof}[Proof of Theorem \protect\ref{THM:sampliingerror2}]
For the sparse deep ReLU network given in Section \ref{SEC:ReLUestimator},
the number of parameters is $W=|V_{m}^{\left( 1\right)
}|+(12R-6)+4+3=|V_{m}^{\left( 1\right) }|+12R+1\asymp |V_{m}^{\left(
1\right) }|+R$ and the number of layers is $L\asymp R\log _{2}d$. Let $%
d(W,L) $ denote the pseudo-dimension of the ReLU network class $\textcolor{black}{ \mathcal{F(%
}R ,m,B_{0},B_{1})}$ defined in \cite{BHLM19}. By the result
given in equation (2) of \cite{BHLM19}, one has that there exist constants
$c$, $C$ such that
\begin{equation}
c\cdot WL\log (W/L)\leq d(W,L)\leq C\cdot WL\log (W).
\label{pseudodimension}
\end{equation}%
By Theorem 12.2 in \cite{AB09}, one has
\begin{equation*}
\mathcal{N}(\sqrt{2}C_{\rho }^{-1}\epsilon /8,\textcolor{black}{ \mathcal{F(%
}R ,m,B_{0},B_{1})},||\cdot ||_{\infty })\leq (\frac{8C_{\rho }B_{0}e}{\sqrt{2}%
\epsilon })^{d(W,L)}\leq (\frac{16C_{\rho }B_{0}}{\epsilon })^{d(W,L)}.
\end{equation*}

Let
\begin{equation}
\varsigma =(\frac{16C_{\rho }B_{\textcolor{black}{0}}}{\epsilon })^{d(W,L)}\exp \left( -n\epsilon
/C^{\ast }\right) .  \label{EQ:omega}
\end{equation}%
By the above results and Theorem \ref{THM:sampliingerror1}, we have
\begin{equation}
P\left( \mathcal{E}(\widehat{f})-\mathcal{E}(f^{\ast })>(3/2)\epsilon
\right) \leq \varsigma .  \label{EQ:inequality_prob}
\end{equation}%
Moreover, (\ref{EQ:omega}) leads to\ $\exp \left( n\epsilon /C^{\ast
}\right) \epsilon ^{d(W,L)}=(16C_{\rho }B_0\varsigma ^{-1/d(W,L)})^{d(W,L)}$
which is equivalent to%
\begin{equation*}
\exp (\varkappa \epsilon )\epsilon =\nu \Longleftrightarrow \exp (\varkappa
\epsilon )(\varkappa \epsilon )=\varkappa \nu ,
\end{equation*}%
where $\varkappa =n/(C^{\ast }d(W,L))$ and $\nu =16C_{\rho }B_0\varsigma
^{-1/d(W,L)}$. Applying the monotone increasing Lambert W-function: $%
W:[0,\infty )\rightarrow \lbrack 0,\infty )$ defined by $W(t\exp (t))=t$ on
both sides of the above equation, we have
\begin{equation*}
W\left( \varkappa \nu \right) =\varkappa \epsilon ,
\end{equation*}%
which is equivalent to $\epsilon =W\left( \varkappa \nu \right) /\varkappa
\leq \max (1,\log (\varkappa \nu ))/\varkappa $, since $W(s)\leq \log (s)$
for all $s\geq e$. Then,
\begin{eqnarray*}
\epsilon &\leq &\max (1,\log (\varkappa \nu ))/\varkappa =\frac{C^{\ast
}d(W,L)}{n}\max (1,\log \frac{16C_{\rho }B_0n}{C^{\ast }d(W,L)\varsigma
^{1/d(W,L)}}) \\
&\leq &\frac{C^{\ast }d(W,L)}{n}\max (1,\log \frac{16C_{\rho }B_0n}{C^{\ast
}d(W,L)\varsigma }).
\end{eqnarray*}%
Therefore, we have
\begin{equation*}
P\left( \mathcal{E}(\widehat{f})-\mathcal{E}(f^{\ast })>(3/2)\frac{C^{\ast
}d(W,L)}{n}\max (1,\log \frac{C^{\ast \ast }n}{d(W,L)\varsigma })\right)
\leq \varsigma ,
\end{equation*}%
for $C^{\ast \ast }=16C_{\rho }B_0C^{\ast -1}$ based on (\ref{EQ:inequality_prob}).
Moreover, the result (\ref{pseudodimension}) implies that
\begin{equation*}
C^{\ast }d(W,L)\max (1,\log \frac{C^{\ast \ast }n}{d(W,L)\varsigma })\leq
C^{\ast }CWL\log (W)\max (1,\log \frac{C^{\ast \ast }n}{cWL\log
(W/L)\varsigma }).
\end{equation*}%
By (\ref{pseudodimension}) and (\ref{EQ:omega}), one has $(\textcolor{black}{\frac{16C_{\rho }B_0}{\epsilon} )^{cWL\log (W/L)}\exp \left(\frac{-n\epsilon}{C^{\ast }}\right) \leq
\varsigma \leq (\frac{16C_{\rho }B_0}{\epsilon})^{CWL\log (W)}} \linebreak \textcolor{black}{\exp \left( \frac{-n\epsilon}{C^{\ast }} \right)} $. Thus, one has
\begin{equation*}
P\left( \mathcal{E}(\widehat{f})-\mathcal{E}(f^{\ast })>\frac{3C^{\ast
}CWL\log (W)}{2n}\max (1,\log \frac{C^{\ast \ast }n}{cWL\log (W/L)\varsigma }%
)\right) \leq \varsigma .
\end{equation*}%
The proof is complete.
\end{proof}

\subsection{Proofs of Theorem \protect\ref{THM:rate}}

Let $2^{m}\asymp n^{1/5}$ and $m\leq R\asymp \log _{2}n$. When $d\lesssim
(\log _{2}n)^{1-\kappa }$ for an arbitrary small constant $\kappa >0\ $,
then for any constant $c\in (0,\infty )$, $\left( c(m+d)\right) ^{2d}\ll
n^{\nu }$ and $\left( c(m+3)\right) ^{2d}\ll n^{\nu }$ for an arbitrary
small $\nu >0$. Therefore, the bias term satisfies
\begin{eqnarray*}
\mathcal{E}(f^{\ast })-\mathcal{E}(f_{0}) &\leq &\zeta _{m,d}\lesssim
n^{-4/5}\{(2/3)(m+3)\}^{-2}\{(2/3)(m+3)\}^{2d} \\
&\ll &n^{-4/5}(\log _{2}n)^{-2}n^{\nu }=n^{-4/5+\nu }(\log _{2}n)^{-2},
\end{eqnarray*}%
where $\zeta _{m,d}$ is given in (\ref{XiRmd}). Then the bias term satisfies
\begin{equation}
\mathcal{E}(f^{\ast })-\mathcal{E}(f_{0})=o(n^{-4/5+\nu }(\log _{2}n)^{-2}).
\label{EQ:biasorder}
\end{equation}%
Moreover, Proposition \ref{PROP:cardinality} leads to
\begin{eqnarray*}
|V_{m}^{\left( 1\right) }| &\lesssim &n^{1/5}d^{1/2}\left( \frac{2e}{d-1}%
\right) ^{d-1}\left( m+d\right) ^{-1}\left( 2e(m+d)\right) ^{d} \\
&\ll &n^{1/5}(\log _{2}n)^{1/2-\kappa /2}(\log _{2}n)^{-1}n^{\nu
/2}=n^{1/5+\nu /2}(\log _{2}n)^{-\kappa /2-1/2},
\end{eqnarray*}%
and $n^{1/5}\lesssim |V_{m}^{\left( 1\right) }|$. Let $\epsilon =n^{-4/5+\nu
/2}(\log _{2}n)^{7/2-\kappa /2}$. By (\ref{pseudodimension}), one has
\begin{eqnarray*}
d(W,L) &\asymp &(|V_{m}^{\left( 1\right) }|+R)L\log (|V_{m}^{\left( 1\right)
}|+R)\\
&\ll & n^{1/5+\nu /2}(\log _{2}n)^{-\kappa /2-1/2}\left( \log _{2}n\right)
^{2}\log _{2}(\log _{2}n) \\
&=&n^{1/5+\nu /2}(\log _{2}n)^{3/2-\kappa /2}\log _{2}(\log _{2}n).
\end{eqnarray*}%
Then $\varsigma $ given in (\ref{EQ:omega}) satisfies
\textcolor{black}{
\begin{eqnarray*}
\varsigma  &\ll &\{n^{4/5-\nu /2}(\log _{2}n)^{\frac{\kappa -7}{2}}\}^{n^{1/5+\nu
/2}(\log _{2}n)^{\frac{3-\kappa}{2}}\log _{2}(\log _{2}n)}\exp \{-n^{1/5+\nu
/2}(\log _{2}n)^{\frac{7-\kappa}{2}}\} \\
&\leq &\exp \{n^{1/5+\nu /2}(\log _{2}n)^{\frac{3-\kappa}{2}}\left( \log
_{2}n\right) \log _{2}(\log _{2}n)\}\exp \{-n^{1/5+\nu /2}(\log
_{2}n)^{\frac{7-\kappa}{2}}\} \\
&=&\exp \{n^{1/5+\nu /2}(\log _{2}n)^{\frac{5-\kappa}{2}}(\log _{2}(\log
_{2}n)-\log _{2}n)\} \\ & \leq & \exp \{-\frac{1}{2}n^{1/5+\nu /2}(\log
_{2}n)^{\frac{5-\kappa}{2}}\}
\end{eqnarray*}%
}
, when $n$ is large. Thus, $\varsigma \rightarrow 0$ as $n\rightarrow \infty $%
. Therefore, the above results and (\ref{EQ:inequality_prob}) lead to
\begin{equation}
\mathcal{E}(\widehat{f})-\mathcal{E}(f^{\ast })=\mathcal{O}_{p}(n^{-4/5+\nu
/2}(\log _{2}n)^{7/2-\kappa /2}).  \label{EQ:errororder}
\end{equation}%
\textcolor{black}{The tuning parameter $\lambda =\mathcal{O}(\epsilon )=\mathcal{O}(n^{-4/5+\nu /2}(\log
_{2}n)^{7/2-\kappa /2})$, and $\varpi _{n}$ need to satisfy $\varpi _{n}=\mathcal{%
O}(\epsilon )=  \mathcal{O}(n^{-4/5+\nu /2} (\log _{2}n)^{7/2-\kappa /2})$.}

From the results in (\ref{EQ:biasorder}) and (\ref{EQ:errororder}), we have
\textcolor{black}{
\begin{align*}
\mathcal{E}(\widehat{f})-\mathcal{E}(f_{0}) & =\mathcal{E}(\widehat{f})-%
\mathcal{E}(f^{\ast })+\mathcal{E}(f^{\ast })-\mathcal{E}(f_{0}) \\
& =\mathcal{O}_{p}(n^{-4/5+\nu /2}(\log _{2}n)^{7/2-\kappa /2})+o(n^{-4/5+\nu
}(\log _{2}n)^{-2})\\
& = o_{p}(n^{-4/5+\nu }(\log _{2}n)^{-2}).
\end{align*}%
}
When $R\asymp \log _{2}n$, the ReLU network that is used to construct the
estimator $\widehat{f}$ has depth $\mathcal{O}(R\log _{2}d)=\mathcal{O}[\log
_{2}n\{\log _{2}(\log _{2}n)\}]$, the number of computational units $%
\mathcal{O}(Rd)\times |V_{m}^{\left( 1\right) }|=\mathcal{O}\{\left( \log
_{2}n\right) ^{3/2(1-\kappa )}n^{1/5+\nu /2}\}$, and the number of weights $%
\mathcal{O}(Rd)\times |V_{m}^{\left( 1\right) }|=\mathcal{O}\{\left( \log
_{2}n\right) ^{3/2(1-\kappa )}n^{1/5+\nu /2}\}$.

\subsection{Proofs of Proposition \protect\ref{THM:rate_fixed}}

Let $2^{m}\asymp n^{1/5}$ and $m\leq R\asymp \log _{2}n$. When $d\ $is a
fixed constant not depending on $n$, the bias term satisfies $\ $
\begin{equation*}
\mathcal{E}(f^{\ast })-\mathcal{E}(f_{0})\leq \zeta _{m,d}\lesssim
n^{-4/5}\{(2/3)(m+3)\}^{2\left( d-1\right) }\lesssim n^{-4/5}(\log
_{2}n)^{2d-2},
\end{equation*}%
where $\zeta _{m,d}$ is given in (\ref{XiRmd}). Then,
\begin{equation}
\mathcal{E}(f^{\ast })-\mathcal{E}(f_{0})=\mathcal{O}(n^{-4/5}(\log
_{2}n)^{2d-2}).  \label{EQ:bias1}
\end{equation}%
Moreover, Proposition \ref{PROP:cardinality} leads to
\begin{equation*}
|V_{m}^{\left( 1\right) }|\lesssim n^{1/5}d^{1/2}\left( m+d\right)
^{-1}\left( 2e(m+d)\right) ^{d}=\mathcal{O}(n^{1/5}(\log _{2}n)^{d-1}),
\end{equation*}%
and $n^{1/5}\lesssim |V_{m}^{\left( 1\right) }|$. Let $\epsilon
=n^{-4/5}(\log _{2}n)^{d+3}$. By (\ref{pseudodimension}), one has $%
d(W,L)\asymp (|V_{m}^{\left( 1\right) }|+R)L\log (|V_{m}^{\left( 1\right)
}|+R)=\mathcal{O}(n^{1/5}(\log _{2}n)^{d+1})$. $\ $Then $\varsigma $ given
in (\ref{EQ:omega}) satisfies
\begin{eqnarray*}
\varsigma  &\lesssim &\{n^{4/5}(\log _{2}n)^{-d-3}\}^{n^{1/5}(\log
_{2}n)^{d+1}}\exp \{-n^{1/5}(\log _{2}n)^{d+3}\} \\
&\lesssim &\exp \{n^{1/5}(\log _{2}n)^{d+2}\}\exp \{-n^{1/5}(\log
_{2}n)^{d+3}\} \\
&=&\exp \{n^{1/5}(\log _{2}n)^{d+2}(1-\log _{2}n)\}\lesssim \exp \{-\frac{1}{%
2}n^{1/5}(\log _{2}n)^{d+2}\}.
\end{eqnarray*}%
Thus, $\varsigma \rightarrow 0$ as $n\rightarrow \infty $. Therefore, the
above results and (\ref{EQ:inequality_prob}) lead to%
\begin{equation}
\mathcal{E}(\widehat{f})-\mathcal{E}(f^{\ast })=\mathcal{O}%
_{p}(n^{-4/5}(\log _{2}n)^{d+3}).  \label{EQ:samplingerror1}
\end{equation}%
The tuning parameter and $\varpi _{n}$ need to satisfy $\varpi _{n}=\mathcal{%
O}(\epsilon )=\mathcal{O}(n^{-4/5}(\log _{2}n)^{d+3})$ and $\lambda =%
\mathcal{O}(\epsilon )=\mathcal{O}(n^{-4/5}(\log _{2}n)^{d+3})$.

From the results in (\ref{EQ:bias1}) and (\ref{EQ:samplingerror1}), we have
\begin{eqnarray*}
\mathcal{E}(\widehat{f})-\mathcal{E}(f_{0}) &=&\mathcal{E}(\widehat{f})-%
\mathcal{E}(f^{\ast })+\mathcal{E}(f^{\ast })-\mathcal{E}(f_{0}) \\
&=&\mathcal{O}_{p}(n^{-4/5}(\log _{2}n)^{d+3})+\mathcal{O}(n^{-4/5}(\log
_{2}n)^{2d-2}) \\
&=&\mathcal{O}_{p}(n^{-4/5}(\log _{2}n)^{\left( d+3\right) \vee \left(
2d-2\right) }).
\end{eqnarray*}%
When $R\asymp \log _{2}(n)$, the ReLU network that is used to construct the
estimator $\widehat{f}$ has depth $\mathcal{O}(R\log _{2}d)=\mathcal{O}(\log
_{2}n)$, the number of computational units
\begin{equation*}
\mathcal{O}(Rd)\times |V_{m}^{\left( 1\right) }|=\mathcal{O}\{\left( \log
_{2}n\right) n^{1/5}(\log _{2}n)^{d-1}\}=\mathcal{O}\{\left( \log
_{2}n\right) ^{d}n^{1/5}\},
\end{equation*}%
and the number of weights $\mathcal{O}(Rd)\times |V_{m}^{\left( 1\right) }|=%
\mathcal{O}\{\left( \log _{2}n\right) ^{d}n^{1/5}\}$.

\subsection{{\protect\normalsize Proofs of Proposition \protect\ref%
{THM:rate_fixed_lower}}}

Let $d^{\ast }(W,L)$ denote the VC-dimension of SDRN, where $W$ is the
number of parameters and $L$ is the number of layers of SDRN. By the
construction of SDRN in Section \ref{sec: SDRN architecture}, one has for $n$
sufficiently large, {\normalsize $W\asymp |V_{m}^{\left( 1\right) }|+R$ and $%
L\asymp R\log _{2}d$. By Proposition \ref{PROP:cardinality} and }$%
2^{m}\asymp n^{1/5}$, one has $n^{1/5}\lesssim |V_{m}^{\left( 1\right) }|$.
\ Moreover, since $R\asymp \log _{2}n$, one has
\begin{equation}
{ L\asymp \log _{2}n\log _{2}d\ }\text{{\normalsize \ and }}%
n^{1/5}\lesssim { W.\ }  \label{EQ:LW}
\end{equation}%
By equation (2) in \cite{BHLM19}, one has $d^{\ast }(W,L)\geq
c^{\ast }\cdot WL\log (W/L)$, for some constant $c^{\ast }${\normalsize $>0$%
. Moreover, (\ref{EQ:LW}) implies that there exists a constant  }$c^{\ast
\ast} >0$ such that $\log (W/L)\geq c^{\ast \ast }%
{ \log _{2}n}$. Therefore, by the above results, one
has that there exists a constant $c_{0}>0$ such that for $n$ sufficiently
large,   
\begin{equation}
d^{\ast }(W,L)-1\geq c_{0}n^{1/5}({ \log _{2}n)^{2}\log _{2}d.}
\label{EQ:dstar}
\end{equation}%
Furthermore, since $\mathcal{E}(f^{\ast })-\mathcal{E}(f_{0})\geq 0$, one
has
\begin{equation}
\mathcal{E}(\widehat{f})-\mathcal{E}(f_{0})=\mathcal{E}(\widehat{f})-%
\mathcal{E}(f^{\ast })+\mathcal{E}(f^{\ast })-\mathcal{E}(f_{0})\geq
\mathcal{E}(\widehat{f})-\mathcal{E}(f^{\ast }).  \label{EQ:fhat-f0}
\end{equation}%
By Theorem 3.6 of {\normalsize \cite{fM2018}, one has
\begin{equation*}
P\left( \mathcal{E}(\widehat{f})-\mathcal{E}(f^{\ast })>\frac{d^{\ast
}(W,L)-1}{32n}\right) \geq 1/100.
\end{equation*}%
The above result together with (}\ref{EQ:dstar}{\normalsize ) and (\ref%
{EQ:fhat-f0}) implies that
\begin{equation*}
P\left( \mathcal{E}(\widehat{f})-\mathcal{E}(f_0)>\frac{c_{0}n^{1/5}(%
{\log _{2}n)^{2}\log _{2}d}}{32n}\right) \geq 1/100.
\end{equation*}%
Let }$c_{1}=c_{0}/32$. Then, one has%
\begin{equation*}
P\left( \mathcal{E}(\widehat{f})-\mathcal{E}(f_0)>c_{1}n^{-4/5}(%
{ \log _{2}n)^{2}\log _{2}d}\right) \geq 1/100.
\end{equation*}%

\subsection{Proofs of Lemmas \protect\ref{LEM:f-f0}-\protect\ref{LEM:Huber}}

\label{proof_LEM:f-f0}

A lemma is presented below and it is used to prove the lemmas given in
Section \ref{discussion_ass}.

\begin{lemma}
\label{LEM:convexity}For any $f\in \textcolor{black}{ \mathcal{F(%
}R ,m,B_{0},B_{1})}$, one
has
\begin{equation*}
\lim_{\delta \rightarrow 0^{+}}\frac{\mathcal{E}(f^{\ast}+\delta
(f-f^{\ast}))-\mathcal{E}(f^{\ast})}{\delta }\geq 0.
\end{equation*}
\end{lemma}

\begin{proof}
Let $\delta \in (0,1)$. Based on the definition of $\textcolor{black}{ \mathcal{F(%
}R ,m,B_{0},B_{1})}$ given in (\ref{DEF:F}), we can see that $f_{0}+\delta
(f-f_{0})\in \textcolor{black}{ \mathcal{F(%
}R ,m,B_{0},B_{1})}$. Moreover $\mathcal{E}%
(f^{\ast}+\delta (f-f^{\ast}))-\mathcal{E}(f^{\ast})\geq 0$ by the
definition of $f^{\ast}$ given in (\ref{EQ:f0RN}).
\end{proof}

\begin{proof}[Proof of Lemma \protect\ref{LEM:f-f0}]
Denote $t_{0}=f^{\ast}(\boldsymbol{x})$ and $t=$ $f(\boldsymbol{x})$. By
Taylor's expansion and Assumption \ref{ass6}, we have
\begin{equation*}
\rho \left( t,y\right) -\rho \left( t_{0},y\right) =\rho ^{\prime }\left(
t_{0},y\right) (t-t_{0})+\int_{0}^{1}2^{-1}\rho ^{\prime \prime }\left(
t_{0}+(t-t_{0})\omega ,y\right) (t-t_{0})^{2}d\omega .
\end{equation*}%
Moreover, by the dominated convergence theorem and Lemma \ref{LEM:convexity}%
,
\begin{align*}
& \int\nolimits_{\mathcal{X}\times \mathcal{Y}}\rho ^{\prime }\left(
f^{\ast}(\boldsymbol{x}),y\right) (f(\boldsymbol{x})-f^{\ast}(%
\boldsymbol{x}))d\mu (\boldsymbol{x,}y) \\
& =\int\nolimits_{\mathcal{X}\times \mathcal{Y}}\lim_{\delta \rightarrow
0^{+}}\frac{\rho \left( f^{\ast}+\delta (f-f^{\ast}),y\right) -\rho
\left( f^{\ast},y\right) }{\delta }d\mu (\boldsymbol{x,}y) \\
& =\lim_{\delta \rightarrow 0^{+}}\int\nolimits_{\mathcal{X}\times \mathcal{Y%
}}\frac{\rho \left( f^{\ast}+\delta (f-f^{\ast}),y\right) -\rho \left(
f^{\ast},y\right) }{\delta }d\mu (\boldsymbol{x,}y) \\
& =\lim_{\delta \rightarrow 0^{+}}\frac{\mathcal{E}(f^{\ast}+\delta
(f-f^{\ast}))-\mathcal{E}(f^{\ast})}{\delta }\geq 0.
\end{align*}%
Therefore,
\begin{eqnarray*}
\mathcal{E}(f)-\mathcal{E}(f^{\ast}) &=&\int\nolimits_{\mathcal{X}\times
\mathcal{Y}}\{\rho \left( f(\boldsymbol{x}),y\right) -\rho \left( f^{\ast}(%
\boldsymbol{x}),y\right) \}d\mu (\boldsymbol{x,}y) \\
&\geq &\int\nolimits_{\mathcal{X}\times \mathcal{Y}}a_{\rho }(f(\boldsymbol{x%
})-f^{\ast}(\boldsymbol{x}))^{2}d\mu (\boldsymbol{x,}y)=a_{\rho
}||f-f^{\ast}||_{2}^{2}.
\end{eqnarray*}%
Since $\int_{\mathcal{Y}}\rho ^{\prime }(f_{0}(\boldsymbol{x}),y)d\mu (y|%
\boldsymbol{x})=0$, then \textcolor{black}{$\int\nolimits_{\mathcal{X}\times \mathcal{Y}}\rho
^{\prime }\left( f_{0}(\boldsymbol{x}),y\right) (f(\boldsymbol{x})-f_{0}(%
\boldsymbol{x})) d\mu (\boldsymbol{x,}y) \linebreak =0$}. Thus,
\begin{eqnarray*}
\mathcal{E}(f)-\mathcal{E}(f_{0}) &=&\int\nolimits_{\mathcal{X}\times
\mathcal{Y}}(\rho \left( f(\boldsymbol{x}),y\right) -\rho \left( f_{0}(%
\boldsymbol{x}),y\right) )d\mu (\boldsymbol{x,}y) \\
&\leq &\int\nolimits_{\mathcal{X}\times \mathcal{Y}}b_{\rho }(f(\boldsymbol{x%
})-f_{0}(\boldsymbol{x}))^{2}d\mu (\boldsymbol{x,}y)=b_{\rho
}||f-f_{0}||_{2}^{2}.
\end{eqnarray*}
\end{proof}

\begin{proof}[Proof of Lemma \protect\ref{LEM:quantile}]
In the following, we will show the results in Lemma \ref{LEM:quantile} when
the loss function $\rho \left( f(\boldsymbol{x}),y\right) $ is the quantile
loss given in (\ref{def:quantile}). We follow a proof procedure from \cite%
{ACL19}. We have
\begin{eqnarray*}
\mathcal{E}(f)-\mathcal{E}(f^{\ast}) &=&\int\nolimits_{\mathcal{X}\times
\mathcal{Y}}(\rho \left( f(\boldsymbol{x}),y\right) -\rho \left( f^{\ast}(%
\boldsymbol{x}),y\right) )d\mu (\boldsymbol{x,}y) \\
&=&\int\nolimits_{\mathcal{X}}\int\nolimits_{\mathcal{Y}}(\rho \left( f(%
\boldsymbol{x}),y\right) -\rho \left( f^{\ast}(\boldsymbol{x}),y\right)
)d\mu (y|\boldsymbol{x})d\mu _{X}(\boldsymbol{x})
\end{eqnarray*}%
Then for all $\boldsymbol{x}\in \mathcal{X}$,
\begin{eqnarray*}
&&\int\nolimits_{\mathcal{Y}}\rho \left( f(\boldsymbol{x}),y\right) d\mu (y|%
\boldsymbol{x}) \\
&=&\int\nolimits_{\mathcal{Y}}I\{y>f(\boldsymbol{x})\}(y-f(\boldsymbol{x}%
))d\mu (y|\boldsymbol{x})+(\tau -1)\int\nolimits_{\mathcal{Y}}(y-f(%
\boldsymbol{x}))d\mu (y|\boldsymbol{x}) \\
&=&g(\boldsymbol{x},f(\boldsymbol{x}))+(\tau -1)\int\nolimits_{\mathcal{Y}%
}yd\mu (y|\boldsymbol{x}),
\end{eqnarray*}%
where $g(\boldsymbol{x},u)=\int\nolimits_{\mathcal{Y}}I\{y>u\}(1-\mu (y|%
\boldsymbol{x}))dy+(1-\tau )u$, and $\mathcal{E}(f)-\mathcal{E}%
(f^{\ast})=\int\nolimits_{\mathcal{X}}g(\boldsymbol{x},f(\boldsymbol{x}%
))d\mu _{X}(\boldsymbol{x})-\int\nolimits_{\mathcal{X}}g(\boldsymbol{x}%
,f^{\ast}(\boldsymbol{x}))d\mu _{X}(\boldsymbol{x})$. Denote $%
t_{0}=f^{\ast}(\boldsymbol{x})$ and $t=$ $f(\boldsymbol{x})$. By Taylor's
expansion, we have
\begin{equation*}
g(\boldsymbol{x},t)-g(\boldsymbol{x},t_{0})=g^{\prime }(\boldsymbol{x}%
,t_{0})(t-t_{0})+\int_{0}^{1}2^{-1}g^{\prime \prime }\left( \boldsymbol{x}%
,t_{0}+(t-t_{0})\omega \right) (t-t_{0})^{2}d\omega .
\end{equation*}

Since $(g\left( \boldsymbol{x},f^{\ast}+\delta (f-f^{\ast})\right)
-g(\boldsymbol{x},f^{\ast}(\boldsymbol{x})))/{\delta }\leq (2-\tau )|f(%
\boldsymbol{x})-f^{\ast}(\boldsymbol{x})|$, by the dominated convergence
theorem and Lemma \ref{LEM:convexity},
\begin{align*}
& \int\nolimits_{\mathcal{X}}g^{\prime }(\boldsymbol{x},f^{\ast}(%
\boldsymbol{x}))(f(\boldsymbol{x})-f^{\ast}(\boldsymbol{x}))d\mu _{X}(%
\boldsymbol{x}) \\
& =\int\nolimits_{\mathcal{X}}\lim_{\delta \rightarrow 0^{+}}\frac{g\left(
\boldsymbol{x},f^{\ast}+\delta (f-f^{\ast})\right) -g(\boldsymbol{x}%
,f^{\ast}(\boldsymbol{x}))}{\delta }d\mu _{X}(\boldsymbol{x}) \\
& =\lim_{\delta \rightarrow 0^{+}}\int\nolimits_{\mathcal{X}\times \mathcal{Y%
}}\frac{g\left( \boldsymbol{x},f^{\ast}+\delta (f-f^{\ast})\right) -g(%
\boldsymbol{x},f^{\ast}(\boldsymbol{x}))}{\delta }d\mu _{X}(\boldsymbol{x})
\\
& =\lim_{\delta \rightarrow 0^{+}}\frac{\mathcal{E}(f^{\ast}+\delta
(f-f^{\ast}))-\mathcal{E}(f^{\ast})}{\delta }\geq 0.
\end{align*}%
The above results together with $\partial ^{2}g(\boldsymbol{x},u)/\partial
u^{2}=\mu ^{\prime }(u|\boldsymbol{x})$ imply that
\begin{eqnarray*}
&&\mathcal{E}(f)-\mathcal{E}(f^{\ast }) \\
&=&\int\nolimits_{\mathcal{X}}\{g(\boldsymbol{x},f(\boldsymbol{x}))-g(%
\boldsymbol{x},f^{\ast }(\boldsymbol{x}))\}d\mu _{X}(\boldsymbol{x}) \\
&\geq &2^{-1}\int\nolimits_{\mathcal{X}}(f(\boldsymbol{x})-f^{\ast }(%
\boldsymbol{x}))^{2}\int_{0}^{1}g^{\prime \prime }\left( \boldsymbol{x}%
,f^{\ast }(\boldsymbol{x})+(f(\boldsymbol{x})-f^{\ast }(\boldsymbol{x}%
))\omega \right) d\omega d\mu _{X}(\boldsymbol{x})\\
&=&2^{-1}\int\nolimits_{\mathcal{X}}(f(\boldsymbol{x})-f^{\ast }(\boldsymbol{%
x}))^{2}\int_{0}^{1}\mu ^{\prime }(f^{\ast }(\boldsymbol{x})+(f(\boldsymbol{x%
})-f^{\ast }(\boldsymbol{x}))\omega |\boldsymbol{x})d\omega d\mu _{X}(%
\boldsymbol{x}) \\
&\geq &\frac{1}{2C_{1}}\int\nolimits_{\mathcal{X}}(f(\boldsymbol{x})-f^{\ast
}(\boldsymbol{x}))^{2}d\mu _{X}(\boldsymbol{x})=\frac{1}{2C_{1}}||f-f^{\ast
}||_{2}^{2}.
\end{eqnarray*}

Note that $\partial g(\boldsymbol{x},u)/\partial u\left\vert
_{u=f_{0}}\right. =0$ and $\partial ^{2}g(\boldsymbol{x},u)/\partial
u^{2}=\mu ^{\prime }(u|\boldsymbol{x})$. Thus by Taylor's expansion,%
\begin{eqnarray*}
&&\mathcal{E}(f)-\mathcal{E}(f_{0}) \\
&=&\int\nolimits_{\mathcal{X}}\{g(\boldsymbol{x},f(\boldsymbol{x}))-g(%
\boldsymbol{x},f_{0}(\boldsymbol{x}))\}d\mu _{X}(\boldsymbol{x}) \\
&=&2^{-1}\int\nolimits_{\mathcal{X}}(f(\boldsymbol{x})-f_{0}(\boldsymbol{x}%
))^{2}\int_{0}^{1}g^{\prime \prime }\left( \boldsymbol{x},f_{0}(\boldsymbol{x%
})+(f(\boldsymbol{x})-f_{0}(\boldsymbol{x}))\omega \right) d\omega d\mu _{X}(%
\boldsymbol{x}) \\
&=&2^{-1}\int\nolimits_{\mathcal{X}}(f(\boldsymbol{x})-f_{0}(\boldsymbol{x}%
))^{2}\int_{0}^{1}\mu ^{\prime }(f_{0}(\boldsymbol{x})+(f(\boldsymbol{x}%
)-f_{0}(\boldsymbol{x}))\omega |\boldsymbol{x})d\omega d\mu _{X}(\boldsymbol{%
x}) \\
&\leq &\frac{1}{2C_{2}}\int\nolimits_{\mathcal{X}}(f(\boldsymbol{x})-f_{0}(%
\boldsymbol{x}))^{2}d\mu _{X}(\boldsymbol{x})=\frac{1}{2C_{2}}%
||f-f_{0}||_{2}^{2}.
\end{eqnarray*}

\end{proof}

\begin{proof}[Proof of Lemma \protect\ref{LEM:Huber}]
The proof of Lemma \ref{LEM:Huber} follows the same procedure as the proof
of Lemma \ref{LEM:quantile}, and thus it is omitted.
\end{proof}

\bibliography{NNbib}

\end{document}